\documentclass[11pt]{article}

\makeatletter
\def\maxwidth{ %
  \ifdim\Gin@nat@width>\linewidth
    \linewidth
  \else
    \Gin@nat@width
  \fi
}
\makeatother

\RequirePackage{amsmath,amssymb}
\RequirePackage{amsthm}
\RequirePackage{natbib}
\RequirePackage[colorlinks,allcolors=blue]{hyperref}

\usepackage{fullpage}
\usepackage{setspace}

\usepackage{bbm,graphicx,bm} 
\usepackage{enumitem}
\usepackage{multirow}
\usepackage[all,import]{xy}
\usepackage{appendix}
\usepackage{comment}
\usepackage{color}
\usepackage{soul}
\usepackage{pdflscape}
\usepackage{subcaption}

\usepackage{booktabs}
\usepackage{longtable}
\usepackage{array}
\usepackage[table]{xcolor}
\usepackage{wrapfig}
\usepackage{float}
\usepackage{colortbl}
\usepackage{pdflscape}
\usepackage{tabu}
\usepackage{threeparttable}
\usepackage{threeparttablex}
\usepackage[normalem]{ulem}
\usepackage{makecell}
\usepackage{afterpage}

\theoremstyle{definition}

\newtheorem{theorem}{Theorem}
\newtheorem{lemma}{Lemma}

\newtheorem{corollary}{Corollary}
\newtheorem{proposition}{Proposition}

\newcommand{\R}{\ensuremath{\mathbb{R}}}

\newcommand{\bbone}{\ensuremath{\mathbbm{1}}}

\newcommand{\E}{\ensuremath{\mathbb{E}}}

\newcommand{\calL}{\ensuremath{\mathcal{L}}}

\newcommand{\trace}{\text{trace}}

\newcommand{\scm}{\text{scm}}
\newcommand{\adj}{\text{adj}}
\newcommand{\aug}{\text{aug}}
\newcommand{\ridge}{\text{ridge}}

\newcommand{\cov}{\text{cov}}

\def\super{\textsuperscript}

\def\b1{\boldsymbol{1}}

\newcommand\undermat[2]{%
  \makebox[0pt][l]{$\smash{\underbrace{\phantom{%
    \begin{matrix}#2\end{matrix}}}_{\text{$#1$}}}$}#2}

\usepackage{xr}

\makeatletter
\newcommand*{\addFileDependency}[1]{
  \typeout{(#1)}
  \@addtofilelist{#1}
  \IfFileExists{#1}{}{\typeout{No file #1.}}
}
\makeatother

\definecolor{RED}{RGB}{255,0,0}

\usepackage{soul}
\usepackage[utf8]{inputenc}

\title{
The Augmented Synthetic Control Method%
\thanks{email: \texttt{afeller@berkeley.edu}. We thank Alberto Abadie, Josh Angrist, Matias Cattaneo, Alex D'Amour, Peng Ding, Erin Hartman, Chad Hazlett, Steve Howard, Guido Imbens, Brian Jacob, Pat Kline, Caleb Miles, Luke Miratrix, Sam Pimentel, Fredrik S{\"a}vje, Jas Sekhon, Jake Soloff, Panos Toulis, Stefan Wager, Yiqing Xu, Alan Zaslavsky, and Xiang Zhou for thoughtful comments and discussion, as well as seminar participants at Stanford, UC Berkeley, UNC, the 2018 Atlantic Causal Inference Conference, COMPIE 2018, and the 2018 Polmeth Conference. We also thank editors and referees for constructive feedback.}}

\author{Eli Ben-Michael, Avi Feller, and Jesse Rothstein\\[1em] UC Berkeley}
\date{July 2020}

\begin{document}

\maketitle
\thispagestyle{empty}
\pagenumbering{gobble}

\begin{abstract}

The  synthetic  control  method  (SCM)  is  a  popular  approach  for  estimating  the  impact  of a treatment on a single unit in panel data settings.
The ``synthetic control'' is a weighted average of control units that balances the treated unit's pre-treatment outcomes as closely as possible. 
A critical feature of the original proposal is to use SCM only when the fit on pre-treatment outcomes is excellent. 
We propose Augmented SCM as an extension of SCM to settings where such pre-treatment fit is infeasible. 
Analogous to bias correction for inexact matching, Augmented SCM uses an outcome model to estimate the bias due to imperfect pre-treatment fit and then de-biases the original SCM estimate. 
Our main proposal, which uses ridge regression as the outcome model, directly controls pre-treatment fit while minimizing extrapolation from the convex hull.
This estimator can also be expressed as a solution to a modified synthetic controls problem that allows negative weights on some donor units.   
We bound the estimation error of this approach under different data generating processes, including a linear factor model, and show how regularization helps to avoid over-fitting to noise.
We demonstrate gains from Augmented SCM with extensive simulation studies and apply this framework to estimate the impact of the 2012 Kansas tax cuts on economic growth.
We implement the proposed method in the new \texttt{augsynth} \texttt{R} package.
\end{abstract}

\clearpage
\pagenumbering{arabic}

\def\spacingset#1{\renewcommand{\baselinestretch}%
{#1}\small\normalsize} \spacingset{1}

\onehalfspacing

\section{Introduction}

The \emph{synthetic control method} (SCM) is a popular approach for estimating the impact of a treatment on a single unit in panel data settings with a modest number of control units and with many pre-treatment periods \citep{Abadie2003,AbadieAlbertoDiamond2010,Abadie2015}. 
The idea is to construct a weighted average of control units, known as a synthetic control, that matches the treated unit's pre-treatment outcomes. The estimated impact is then the difference in post-treatment outcomes between the treated unit and the synthetic control.
SCM has been widely applied --- the main SCM papers have over 4,000 citations --- and has been called ``arguably the most important innovation in the policy evaluation literature in the last 15 years'' \citep{athey2017state}.

A critical feature of the original proposal, not always followed in practice, is to use SCM only when the synthetic control's pre-treatment outcomes closely match the pre-treatment outcomes for the treated unit \citep{Abadie2015}. 
When it is not possible to construct a synthetic control that fits pre-treatment outcomes well, the original papers advise against using SCM.
At that point, researchers often fall back to linear regression. This allows better (often perfect) pre-treatment fit, but does so by applying negative weights to some control units, extrapolating outside the support of the data.

We propose the \emph{augmented synthetic control method} (ASCM) as a middle ground in settings where excellent pre-treatment fit using SCM alone is not feasible. 
Analogous to bias correction for inexact matching \citep{rubin1973adjust, abadie2011bias}, ASCM begins with the original SCM estimate, uses an outcome model to estimate the bias due to imperfect pre-treatment fit, and then uses this to de-bias the estimate.
If pre-treatment fit is good, the estimated bias will be small, and the SCM and ASCM estimates will be similar. Otherwise, the estimates will diverge, and ASCM will rely more heavily on extrapolation.

Our primary proposal is to augment SCM with a ridge regression model, which we call \emph{Ridge ASCM}. 
We show that, like SCM, the Ridge ASCM estimator
can be written as a weighted average of the control unit outcomes. 
We also show that Ridge ASCM weights can be written as the solution to a modified synthetic controls problem, targeting the same imbalance metric as traditional SCM.
However, where SCM weights are always non-negative, Ridge ASCM admits negative weights, using extrapolation to improve pre-treatment fit.
The regularization parameter in Ridge ASCM directly parameterizes the level of extrapolation by penalizing the distance from SCM weights.
By contrast, 
(ridge) regression alone, which can also be written as a modified synthetic controls problem with possibly negative weights, allows for arbitrary extrapolation and possibly unchecked extrapolation bias.

We relate Ridge ASCM's improved pre-treatment fit to a finite sample bound on estimation error
under several data generating processes, including a simple linear model and the linear factor model often invoked in this setting \citep{AbadieAlbertoDiamond2010}. 
Under a linear model, improving pre-treatment fit directly reduces bias, and the Ridge ASCM penalty term negotiates a bias-variance trade-off. Under a latent factor model, improving pre-treatment fit again reduces bias, though there is now a risk of over-fitting. The penalty term also directly parameterizes this trade-off. Thus, choosing the hyperparameter will be important for practice; we propose a cross-validation procedure in Section \ref{sec:cv}.

Finally, we also describe how the Augmented SCM approach can be extended to incorporate auxiliary covariates other than pre-treatment outcomes. 
We first propose to include the auxiliary covariates in parallel to the lagged outcomes in both the SCM and outcome models. 
We also propose an alternative when there relatively few covariates, extending a suggestion from \citet{Doudchenko2017}: first residualize both the pre- and post-treatment outcomes against the auxiliary covariates, then fit Ridge ASCM on the residualized outcome series. 
We show that this controls the estimation error under a linear factor model with auxiliary covariates.

We demonstrate the properties of Augmented SCM both via calibrated simulation studies and by using it to examine the effect of an aggressive tax cut in Kansas in 2012 on economic output, finding a substantial negative effect. Overall, we see large gains from ASCM relative to alternative estimators, especially under model mis-specification, in terms of both bias and root mean squared error.
We implement the proposed methodology in the \texttt{augsynth} package for \texttt{R}, available at \href{https://github.com/ebenmichael/augsynth}{\texttt{https://github.com/ebenmichael/augsynth}}.

The paper proceeds as follows.
Section \ref{sec:related_work} briefly reviews related work.
Section \ref{sec:scm_overview} introduces notation and the SCM estimator.
Section \ref{sec:aug_scm_overview} gives an overview of Augmented SCM.
Section \ref{sec:ridge_ascm} gives key numerical results for Ridge ASCM.
Section \ref{sec:bias_section} bounds the estimation error under a linear model and a linear factor model, the standard setting for SCM.
Section \ref{sec:extensions} extends the ASCM framework to incorporate auxiliary covariates and alternative outcome models.
Section \ref{sec:empirical} reports on extensive simulation studies as well as the application to the Kansas tax cuts.
Finally, Section \ref{sec:discussion} discusses some possible directions for further research. 
The appendix includes all of the proofs, as well as additional derivations and technical discussion, including a discussion of inference.

\subsection{Related work}
\label{sec:related_work}

SCM was introduced by \citet{Abadie2003} and \citet{AbadieAlbertoDiamond2010, Abadie2015} and is the subject of an extensive methodological literature; see \citet{abadie2019synthreview} and \citet{samartsidis2019assessing} for recent reviews.
We briefly highlight some relevant aspects of this literature. 

A first set of papers adapts the original SCM proposal to allow for more robust estimation while retaining the simplex constraint on the weights. 
\citet{Robbins2017, Doudchenko2017, Abadie_LHour, minard2018dispersion} incorporate a penalty on the weights into the SCM optimization problem, building on a suggestion in \citet{Abadie2015}.
\citet{gobillon2016regional} and \citet{hazlett2018trajectory} explore dimension reduction strategies and other data transformations that can improve the performance of the subsequent estimator.
See also \citet{bilinski2020goldilocks}, who propose a cross-validation procedure for selecting the number of lagged outcomes to include.

A second set of papers relaxes various constraints imposed in the original SCM problem.
In particular, \citet{Doudchenko2017} relax the SCM restriction that control unit weights be non-negative, arguing that there are many settings in which negative weights would be desirable. 
\citet{amjad2018robust} propose an interesting variant that combines negative weights with a pre-processing step.
\citet{powell2018imperfect} instead allows for extrapolation via a Frisch-Waugh-Lovell-style projection, which similarly generalizes the typical SCM setting.
\citet{Doudchenko2017} and \citet{ferman2018revisiting} both propose to incorporate an intercept into the SCM problem, which we discuss in Section \ref{sec:estimator_overview}. 
There have also been several other proposals to reduce bias in SCM, developed independently and contemporaneously with ours. 
\citet{Abadie_LHour} also propose bias correcting SCM using regression, but focus on bias due to interpolation rather than poor pre-treatment fit. \citet{kellogg2020combining} propose using a weighted average of SCM and matching, trading off interpolation and extrapolation bias. Finally, \citet{arkhangelsky2018synthetic} propose the \emph{Synthetic Difference-in-Differences} estimator, which can be seen as a special case of our proposal with a constrained outcome regression.

Finally, there have been several recent proposals to use outcome modeling rather than SCM-style weighting in this setting. These include the matrix completion method in \citet{athey2017mcp}, the generalized synthetic control method in \citet{Xu2017}, and the combined approaches in \citet{hsiao2018panel}. We explore the performance of select methods, both in isolation and in combination with SCM, in Section \ref{sec:sim_results}.

\section{Overview of the Synthetic Control Method} \label{sec:scm_overview}
\subsection{Notation and setup}
\label{sec:setup}

We consider the canonical SCM panel data setting with $i = 1, \ldots, N$ units observed for $t = 1, \ldots, T$ time periods; for the theoretical discussion below, we will consider both $N$ and $T$ to be fixed.
Let $W_i$ be an indicator that unit $i$ is treated at time $T_0<T$ where units with $W_i=0$ never receive the treatment. 
We restrict our attention to the case where a single unit receives treatment, and follow the convention that this is the first one, $W_1 = 1$; see \citet{benmichael2019multisynth} for an extension to multiple treated units.
The remaining $N_0=N-1$ units are possible controls, often referred to as \emph{donor units} in the SCM context.
To simplify notation, we limit to one post-treatment observation, $T=T_0+1$, though our results are easily extended to larger $T$.

We adopt the potential outcomes framework~\citep{neyman1923, rubin1974} and invoke SUTVA, which assumes a well-defined treatment and excludes interference between units \citep{rubin1980}; the potential outcomes for unit $i$ in period $t$ under control and treatment are $Y_{it}(0)$ and $Y_{it}(1)$, respectively. 
The observed outcomes are then:
\begin{equation}
Y_{it}=\begin{cases} Y_{it}(0) & \text{if } W_i=0 \text{ or } t\leq T_0 \\
Y_{it}(1) & \text{if } W_i=1 \text{ and } t>T_0. \\  
\end{cases}
\end{equation}

\noindent We next assume that control potential outcomes are generated as a fixed component $m_{it}$ plus mean zero, additive noise $\varepsilon_{it}$ drawn from some distribution $P(\cdot)$,
$$Y_{it}(0) = m_{it} + \varepsilon_{it}.$$
The treated potential outcome is then $Y_{it}(1) = Y_{it}(0) + \tau_{it}$, where the treatment effects $\tau_{it}$ are the key estimands, and are fixed parameters.
The treatment effect of interest is thus $\tau = \tau_{1T} = Y_{1T}(1) - Y_{1T}(0)$.

In Section \ref{sec:bias_section}, we consider special cases where $m_{it}$ is a linear function of lagged outcomes or where $m_{it}$ is a linear factor model; in the Appendix, we also consider the case where $m_{it}$ is a linear model with Lipshitz deviations from linearity.
We make two assumptions about the distribution of the noise terms $\varepsilon_{it}$. 
First, we assume that treatment assignment $W_i$ is ignorable given $m_{it}$; specifically that the noise terms in the post-treatment time periods $\bm{\varepsilon}_T = (\varepsilon_{1T},\ldots,\varepsilon_{NT})$ are mean-zero and uncorrelated with treatment assignment,
\begin{equation}
  \label{eq:ignore}
    \E_{\bm{\varepsilon}_{T}}\left[W_i\varepsilon_{iT}\right] = \E_{\bm{\varepsilon}_{T}}\left[(1 - W_i)\varepsilon_{iT}\right] =
    \E_{\bm{\varepsilon}_{T}}\left[\varepsilon_{iT}\right] = 0,
\end{equation}
where the expectation is taken with respect to the noise term at the post-treatment time $T$, $\bm{\varepsilon}_T$. As a result, the noise terms for the treated and control units do not systematically deviate from each other. We discuss this in more detail in the context of our application in Section \ref{sec:empirical}.
Second, for the theoretical results in Sections \ref{sec:bias_section} and \ref{sec:extensions}, we assume that the error terms $\varepsilon_{it}$ are independent (across units and over time) sub-Gaussian random variables with scale parameter $\sigma$. See, for example, \citet{chernozhukov2017exact} for an extended discussion of this general setup in panel data settings.

To emphasize that pre-treatment outcomes serve as covariates in SCM, we use $X_{it}$, for $t\leq T_0$, to represent pre-treatment outcomes; we use the terms \emph{pre-treatment fit} and \emph{covariate balance} interchangeably.
With some abuse of notation, we use $\bm{X}_{0\cdot}$ to represent the $N_0$-by-$T_0$ matrix of control unit pre-treatment outcomes and $\bm{Y}_{0T}$ for the $N_0$-vector of control unit outcomes in period $T$. 
With only one treated unit, $Y_{1T}$ is a scalar, and $\bm{X}_{1\cdot}$ is a $T_0$-row vector of treated unit pre-treatment outcomes.
The data structure is then:

\begin{align}\label{eq:scm_setup}
\left( \begin{array}{c c c c c}
Y_{11} & Y_{12} & \ldots & Y_{1T_0} & Y_{1T} \\
Y_{21} & Y_{22} & \ldots & Y_{2T_0} & Y_{2T}\\
\vdots & & & & \vdots\\
Y_{N1} & Y_{N2} & \ldots & Y_{NT_0} & Y_{NT}\\
\end{array} \right) \equiv
\left(\begin{array}{c  c c c | c} 
X_{11} & X_{12} & \hdots & X_{1T_0} & Y_{1T} \\
\hline
X_{21} & X_{22} & \hdots & X_{2T_0} & Y_{2T} \\
\vdots &  & & & \vdots \\
\undermat{\text{pre-treatment outcomes}}{X_{N1} & X_{N2} & \hdots & X_{NT_0}} & Y_{NT} \\
\end{array}\right) 
\equiv
\left(\begin{array}{c | c}
  \bm{X}_{1\cdot} & Y_{1T} \\
\hline
\bm{X}_{0\cdot} & \bm{Y}_{0T}
\end{array}\right)\\
\nonumber
\end{align}

\subsection{Synthetic Control Method}
\label{sec:scm}	

The Synthetic Control Method imputes the missing potential outcome for the treated unit, $Y_{1T}(0)$, as a weighted average of the control outcomes, $\bm{Y}_{0T}'\bm{\gamma}$ \citep{Abadie2003, AbadieAlbertoDiamond2010, Abadie2015}. Weights are chosen to balance pre-treatment outcomes and possibly other covariates. We consider a version of SCM that chooses weights $\bm{\gamma}$ as a solution to the constrained optimization problem:
\begin{equation}
\begin{aligned}
\min_{\gamma} \;\;\;\;& \|\bm{V_x}^{1/2}( \bm{X}_{1\cdot} - \bm{X}_{0\cdot}'\bm{\gamma})\|_2^2 \;+\; \zeta \sum_{W_i = 0} f(\gamma_i) \\
\text{subject to} \;\;\;\; & \sum_{W_i=0} \gamma_i = 1\\
& \gamma_i \geq 0 \;\;\; i: W_i = 0
\end{aligned}
 \label{eq:vanillaSCM}
\end{equation}
\noindent where the constraints limit $\bm{\gamma}$ to the simplex $\Delta^{N_0} = \{\bm{\gamma} \in \R^{N_0} \mid \gamma_i \geq 0 ~\forall i, ~\sum_i\gamma_i = 1\}$, and where $\bm{V_x} \in \R^{T_0 \times T_0}$ is a symmetric importance matrix and $\|\bm{V_x}^{1/2} (\bm{X}_{1\cdot} - \bm{X}_{0\cdot}'\gamma)\|^2_2 \equiv (\bm{X}_{1\cdot} - \bm{X}_{0\cdot}'\gamma)' \bm{V_x} (\bm{X}_{1\cdot} - \bm{X}_{0\cdot}'\bm{\gamma})$ is the 2-norm on $\mathbb{R}^{T_0}$ after applying $\bm{V_x}^{1/2}$ as a linear transformation. 
To simplify the exposition and notation below, we will generally take $\bm{V_x}$ to be the identity matrix.
The simplex constraint in Equation \eqref{eq:vanillaSCM} ensures that the weights will be sparse and non-negative; \citet{AbadieAlbertoDiamond2010, Abadie2015} argue that enforcing this constraint is important for preserving interpretability.

Equation \eqref{eq:vanillaSCM} modifies the original SCM proposal in two ways. 
First, Equation \eqref{eq:vanillaSCM} excludes auxiliary covariates; we re-introduce them in Section \ref{sec:extensions}. 
Second, Equation \eqref{eq:vanillaSCM} penalizes the dispersion of the weights with hyperparameter $\zeta \geq 0$, following a suggestion in \citet{Abadie2015}.
Possible penalization functions include an elastic net penalty \citep{Doudchenko2017}, an entropy penalty \citep{Robbins2017}, and a penalty based on a measure of pairwise distance \citep{Abadie_LHour}. 
The choice of penalty is less central when weights are constrained to be on the simplex, but becomes more important below when we relax this constraint \citep{Doudchenko2017}.

We can view the SCM optimization problem in Equation \eqref{eq:vanillaSCM} as an approximate balancing weights estimator \citep[see, e.g.][]{Zubizarreta2015,hirshberg2019minimax}.
As with all balancing estimators, a central question is what quantity to balance.
Following the recent methodological literature \citep[see][]{Doudchenko2017, ferman2018revisiting},
Equation \eqref{eq:vanillaSCM} directly optimizes for the pre-treatment fit, minimizing the (possibly weighted) imbalance of pre-treatment outcomes between the treated unit and the weighted control mean.
In Section \ref{sec:bias_section}, we argue that this a natural quantity to target under both linearity and a latent factor model.
Many choices are possible, however, and we can easily modify Equation \eqref{eq:vanillaSCM} to balance other summary measures and functions of the lagged outcomes; see, for example, \citet{gobillon2016regional}, \citet{amjad2018robust}, and \citet{hazlett2018trajectory}.

When the treated unit's vector of lagged outcomes, $\bm{X}_{1\cdot}$, is inside the convex hull of the control units' lagged outcomes, $\bm{X}_{0\cdot}$, the SCM weights in Equation \eqref{eq:vanillaSCM} achieve perfect pre-treatment fit, and the resulting estimator has many attractive properties, including a bias bound established by \citet{AbadieAlbertoDiamond2010}.
Due to the curse of dimensionality, however, achieving perfect (or nearly perfect) pre-treatment fit is not always feasible with weights constrained to be on the simplex \citep[see][]{ferman2018revisiting}.
When ``the pre-treatment fit is poor or the number of pre-treatment periods is small,'' \citet{Abadie2015} recommend against using SCM.
Even if the pre-treatment fit is excellent,  \citet{AbadieAlbertoDiamond2010, Abadie2015} propose extensive placebo checks to ensure that SCM weights do not overfit to noise. 
Thus, the conditional nature of the analysis is critical to deploying SCM, excluding many practical settings.
Our proposal enables the use of (a modified) SCM approach in many of the cases where SCM alone is infeasible.

\section{Augmented SCM}
\label{sec:aug_scm_overview}

\subsection{Overview}

We now show how to modify the SCM approach to adjust for poor pre-treatment fit. 
Let $\hat{m}_{iT}$ be an estimator for the post-treatment control potential outcome $Y_{iT}(0)$. The \emph{Augmented SCM} (ASCM) estimator for $Y_{1T}(0)$ is:
\begin{align}
   \hat{Y}_{1T}^{\text{aug}}(0)  &= \sum_{W_i=0} \hat{\gamma}_i^\scm Y_{iT} \;\; + \; \left( \hat{m}_{1T} -  \sum_{W_i=0} \hat{\gamma}_i^\scm\hat{m}_{iT} \right)\label{eq:ascm_2}\\[0.5em]
    &= \hat{m}_{1T} \;\;\; + \;\;\; \sum_{W_i=0} \hat{\gamma}_i^\scm (Y_{iT} - \hat{m}_{iT}),     \label{eq:ascm_1}
\end{align}
where weights $\hat{\gamma}_i^\scm$ are the SCM weights defined above.
Standard SCM is a special case, where $\hat{m}_{iT}$ is a constant. We will largely focus on estimators that are a function of pre-treatment outcomes, $\hat{m}_{iT} \equiv \hat{m}(\bm{X}_i)$, where $\hat{m}:\R^{T_0} \to \R$.

Equations \eqref{eq:ascm_2} and \eqref{eq:ascm_1}, while equivalent, highlight two distinct motivations for ASCM. Equation \eqref{eq:ascm_2} directly corrects the SCM estimate, $\sum \hat{\gamma}_i^\scm Y_{iT}$, by the imbalance in a particular function of the pre-treatment outcomes $\hat{m}(\cdot)$. Intuitively, since $\hat{m}$ estimates the post-treatment outcome, we can view this as an estimate of the bias due to imbalance, analogous to bias correction for inexact matching \citep{rubin1973adjust, abadie2011bias}.
In this form, we can see that SCM and ASCM estimates will be similar if the estimated bias is small, as measured by imbalance in $\hat{m}(\cdot)$. 
In independent work, \citet{Abadie_LHour} also consider a bias-corrected estimator of this form, though their paper focuses on reducing bias due to interpolation rather than extrapolation. 

Equation \eqref{eq:ascm_1}, by contrast, is analogous to standard doubly robust estimation \citep{robins1994estimation}, which begins with the outcome model but then re-weights to balance residuals.
This is also comparable in form to the generalized regression estimator in survey sampling \citep{Cassel1976, Breidt2017}, which has been adapted to the causal inference setting by, among others, \citet{Athey2016} and \citet{Hirshberg2018}. 
We discuss a connection to inverse propensity score weighting in Appendix \ref{sec:ipw_connection}.

\subsection{Choice of estimator}
\label{sec:estimator_overview}

While this setup is general, the choice of estimator $\hat{m}$ is important both for understanding the procedure's properties and for practical performance. We give a brief overview of two special cases: (1) when $\hat{m}$ is linear in the pre-treatment outcomes; and (2) when $\hat{m}$ is linear in the comparison units. Ridge regression is an important example that is linear in both pre-treatment outcomes and comparison units;
we explore this estimator further in Sections \ref{sec:ridge_ascm} and \ref{sec:bias_section}.

First, consider a model that is linear in pre-treatment outcomes, $\hat{m}(\bm{X}) = \hat{\eta}_0 + \bm{\hat{\eta}} \cdot \bm{X}$. The augmented estimator \eqref{eq:ascm_2} is then:
\begin{equation}
    \label{eq:ridge_ascm}
    \hat{Y}_{1T}^{\aug}(0) = \sum_{W_i=0} \hat{\gamma}_i^{\scm} Y_{iT} \;\; + \;  \sum_{t=1}^{T_0} \hat{\eta}_t \left(X_{1t} -  \sum_{W_i=0} \hat{\gamma}_i^{\scm} X_{it} \right).
\end{equation}
Pre-treatment periods that are more predictive of the post-treatment outcome will have larger (absolute) regression coefficients and so imbalance in these periods will lead to a larger adjustment.
Thus, even if we do not \emph{a priori} prioritize balance in any particular pre-treatment time periods (via the choice of $\bm{V_x}$),
the linear model augmentation will adjust for the time periods that are empirically more predictive of the post-treatment outcome. As we show in Section \ref{sec:ridge_ascm}, the ridge-regularized linear model is an important special case in which the resulting augmented estimator is itself a penalized synthetic control estimator. This allows for a more direct analysis of the role of bias correction.

Second, consider an outcome model that is a linear combination of comparison units, $\hat{m}(\bm{X}) = \sum_{W_i = 0}\hat{\alpha}_i(\bm{X}) Y_{iT}$, for some weighting function $\hat{\alpha}:\R^{T_0} \to \R^{N_0}$.
Examples include $k$-nearest neighbor matching and kernel weighting as well as other  ``vertical'' regression approaches \citep{athey2017mcp}.
The augmented estimator \eqref{eq:ascm_2} is itself a weighting estimator that adjusts the SCM weights:\footnote{We thank an anonymous reviewer for suggesting this presentation.}
\begin{equation}
    \label{eq:augmented_with_weights}
    \hat{Y}_{1T}^\aug(0) = \sum_{W_i = 0}\left(\hat{\gamma}_i^\scm + \hat{\gamma}^\adj_i \right) Y_{iT}, \;\;\text{ where } \;\; \hat{\gamma}_i^\adj \equiv \hat{\alpha}_i(\bm{X}_1) - \sum_{W_j = 0}\hat{\gamma}_j^\scm \hat{\alpha}_i(\bm{X}_j).
\end{equation}
Here the adjustment term for unit $i$, $\hat{\gamma}^\adj_i$, is the imbalance in a unit $i$-specific transformation of the lagged outcomes that depends on the weighting function $\alpha(\cdot)$.
While $\hat{\gamma}^\scm$ are constrained to be on the simplex, the form of  $\hat{\gamma}^\adj$ makes clear that the overall weights can be negative.

There are many special cases to consider, though we do not explore them in depth. One is the linear-in-lagged-outcomes model with equal coefficients, $\hat{\eta}_t = \frac{1}{T_0}$, which estimates a fixed-effects outcome model as $\hat{m}(\bm{X}_i) = \bar{X}_i$. The corresponding treatment effect estimate adjusts for imbalance in all pre-treatment time periods equally, and yields a weighted difference-in-differences estimator:
\begin{equation}
  \label{eq:demeaned_scm_tau}
  \hat{\tau}^{\text{de}} = \left(Y_{1T} - \bar{X}_1\right) - \left(\sum_{W_i=0}\hat{\gamma}_{i}(Y_{iT} - \bar{X}_i)\right) = \frac{1}{T_0}\sum_{t=1}^{T_0}\left[ \left(Y_{1T} - X_{1t}\right) - \left(\sum_{W_i=0}\hat{\gamma}_i(Y_{iT} - X_{it})\right)\right].
\end{equation}
An augmented estimator of this form has appeared as the \emph{de-meaned} or \emph{intercept shift SCM} \citep{Doudchenko2017, ferman2018revisiting}.\footnote{In these proposals, the SCM weights balance the \emph{residual} outcomes $X_{it}-\bar{X}_i$ rather than the raw outcomes $X_{it}$.  We further consider balancing residuals in Section \ref{sec:extensions}.} 
See also \citet{arkhangelsky2018synthetic}, 
who extend this to weight across both units and time.

In Section \ref{sec:sim_results} we conduct a simulation study to inspect the performance of a range of estimators, including other penalized linear models, such as the LASSO, flexible machine learning models, such as random forests, and panel data methods, such as fixed effects models and low-rank matrix completion methods \citep{Xu2017, athey2017mcp}.

\section{Ridge ASCM improves pre-treatment fit while penalizing extrapolation}
\label{sec:ridge_ascm}

We now inspect the algorithmic and numerical properties for the special case where $\hat{m}(\bm{X}_i)$ is estimated via a ridge-regularized linear model, which we refer to as \emph{Ridge Augmented SCM} (Ridge ASCM). 
With Ridge ASCM, the estimator for the post-treatment outcome is $\hat{m}(\bm{X}_i)=\hat{\eta}_0^\ridge + \bm{X}_{i}^{\prime}\bm{\hat{\eta}}^\ridge$, where $\hat{\eta}_0^\ridge$ and $\bm{\hat{\eta}}^\ridge$ are the coefficients of a ridge regression of control post-treatment outcomes $\bm{Y}_{0T}$ on centered pre-treatment outcomes $\bm{X}_{0\cdot}$ with penalty hyper-parameter $\lambda^\ridge$:\footnote{Similar to the synthetic controls problem, we can regularize time periods differently with a generalized ridge penalty $\bm{\eta}' \bm{\Lambda} \bm{\eta}$ using an importance matrix $\bm{\Lambda}$. Following the typical case with diagonal elements, the generalized ridge penalty reduces to separate regularization on each time period.}
\begin{equation}
  \label{eq:ridge_params}
  \left\{\hat{\eta}^\ridge_0,\bm{\hat{\eta}}^\ridge\right\}=\arg\min_{\eta_0, \bm{\eta}} \;\; \frac{1}{2}\sum_{W_i=0}(Y_i - (\eta_0 + X_i'\bm{\eta}))^2 + \lambda^\ridge \|\bm{\eta}\|_2^2.
\end{equation}
\noindent The Ridge Augmented SCM estimator is then:
\begin{equation} \label{eq:ridge_ascm_intro}
\hat{Y}_{1T}^{\aug}(0) = \sum_{W_i=0} \hat{\gamma}_i^{\scm} Y_{iT} +  \left(\bm{X}_{1} -  \sum_{W_i=0} \hat{\gamma}_i^{\scm} \bm{X}_{i\cdot} \right) \cdot \bm{\hat{\eta}}^\ridge.
\end{equation}

\noindent We first show that Ridge ASCM is a linear weighting estimator as in Equation \eqref{eq:augmented_with_weights}. Unlike augmenting with other linear weighting estimators, when augmenting with ridge regression the implied weights are themselves the solution to a penalized synthetic control problem, as in Equation \eqref{eq:vanillaSCM}. 
Using this representation, we show that when the treated unit lies outside the convex hull of the control units, Ridge ASCM improves the pre-treatment fit relative to SCM alone by allowing for negative weights and extrapolating away from the convex hull.

Allowing for negative weights is an important departure from the original SCM proposal,
which constrains weights to be on the simplex.
Ridge regression alone also allows for negative weights, and may have negative weights even when the treated unit is inside of the convex hull.
In contrast, Ridge ASCM directly penalizes distance from the sparse, non-negative SCM weights, controlling the amount of extrapolation by the choice of $\lambda^{\ridge}$, and only resorts to negative weights if the treated unit is outside of the convex hull.

\subsection{Ridge ASCM as a penalized SCM estimator}

We now express both Ridge ASCM and ridge regression alone as special cases of the penalized SCM problem in Equation \eqref{eq:vanillaSCM}.
The Ridge ASCM estimate of the counterfactual is the solution to \eqref{eq:vanillaSCM}, replacing the simplex constraint with a penalty  $f(\gamma_i) = \left(\gamma_i - \hat{\gamma}^\scm_i\right)^2$ that penalizes \emph{deviations from the SCM weights}.

\begin{lemma}
\label{lem:ridge_ascm_weights}

\noindent The ridge-augmented SCM estimator \eqref{eq:ridge_ascm} is:
\begin{equation}
    \label{eq:ridge_ascm_weights}
    \hat{Y}_{1T}^{\aug}(0) = \sum_{W_i=0}\hat{\gamma}_i^{\aug}Y_{iT},
\end{equation}
where 
\begin{equation}
    \label{eq:greg}
    \hat{\gamma}_i^{\aug} = \hat{\gamma}_i^{\scm} + (\bm{X}_{1} - \bm{X}_{0\cdot}'\bm{\hat{\gamma}}^{\scm})'(\bm{X}_{0\cdot}'\bm{X}_{0\cdot} + \lambda^\ridge \bm{I_{T_0}})^{-1} \bm{X}_{i\cdot}.
\end{equation}
Moreover, the Ridge ASCM weights $\bm{\hat{\gamma}}^{\aug}$ are the solution to 
    \begin{equation}
        \label{eq:ridge_ascm_primal}
        \min_{\bm{\gamma}~ \text{s.t.} \sum_i\gamma_i=1} \frac{1}{2\lambda^\ridge}\|\bm{X}_{1\cdot} - \bm{X}_{0\cdot}'\gamma\|_2^2 + \frac{1}{2}\left\|\bm{\gamma} - \bm{\hat{\gamma}}^\scm\right\|_2^2.
    \end{equation}
\end{lemma}

\noindent When the treated unit is in the convex hull of the control units --- so the SCM weights exactly balance the lagged outcomes --- the Ridge ASCM and SCM weights are identical.
When SCM weights do not achieve exact balance, the Ridge ASCM solution will use negative weights to extrapolate from the convex hull of the control units. The amount of extrapolation is determined both by the amount of imbalance and by the hyperparameter $\lambda^\ridge$. 
When SCM yields good pre-treatment fit or when $\lambda^\ridge$ is large, the adjustment term will be small and $\bm{\hat{\gamma}}^\aug$ will remain close to the SCM weights.

We can similarly characterize ridge regression alone as a solution to a penalized SCM problem
where the penalty term, $f(\gamma_i) = \left(\gamma_i - \frac{1}{N_0}\right)^2$, penalizes the variance of the weights. 
Other penalized linear models, such as the LASSO or elastic net, do not have this same representation as a penalized SCM estimator.

\begin{lemma}
  \label{lem:ridge_weights_main}
  The ridge regression estimator $\hat{Y}_{1T}^{\ridge}(0) \equiv \hat{\eta}_0^\ridge + \bm{X}_{1} \cdot \bm{\hat{\eta}}^\ridge$ can be written as $\hat{Y}_{1T}^{\ridge}(0) = \sum_{W_i=0} \hat{\gamma}_i^{\ridge} Y_{iT},$
where the ridge weights $\bm{\hat{\gamma}}^{\ridge}$ are the solution to:
      \begin{equation}
        \label{eq:ridge_primal_main}
        \min_{\bm{\gamma}~ \mid~ \sum_i\gamma_i=1} \;\; \frac{1}{2\lambda^\ridge}\|\bm{X}_{1} - \bm{X}_{0\cdot}'\bm{\gamma}\|_2^2 + \frac{1}{2}\left\|\bm{\gamma} - \frac{1}{N_0}\right\|_2^2.
    \end{equation}

\end{lemma}

\noindent For ridge regression alone, the hyperparameter $\lambda^{\ridge}$ controls the variance of the weights rather than the degree of extrapolation from the simplex.
Thus, in order to reduce variance, the ridge regression weights might still be negative even if the treated unit is inside of the convex hull and SCM achieves perfect fit.

\begin{figure}[tbp]
  \centering
    \begin{subfigure}[t]{0.46\textwidth}  
  {\centering \includegraphics[width=0.8\textwidth]{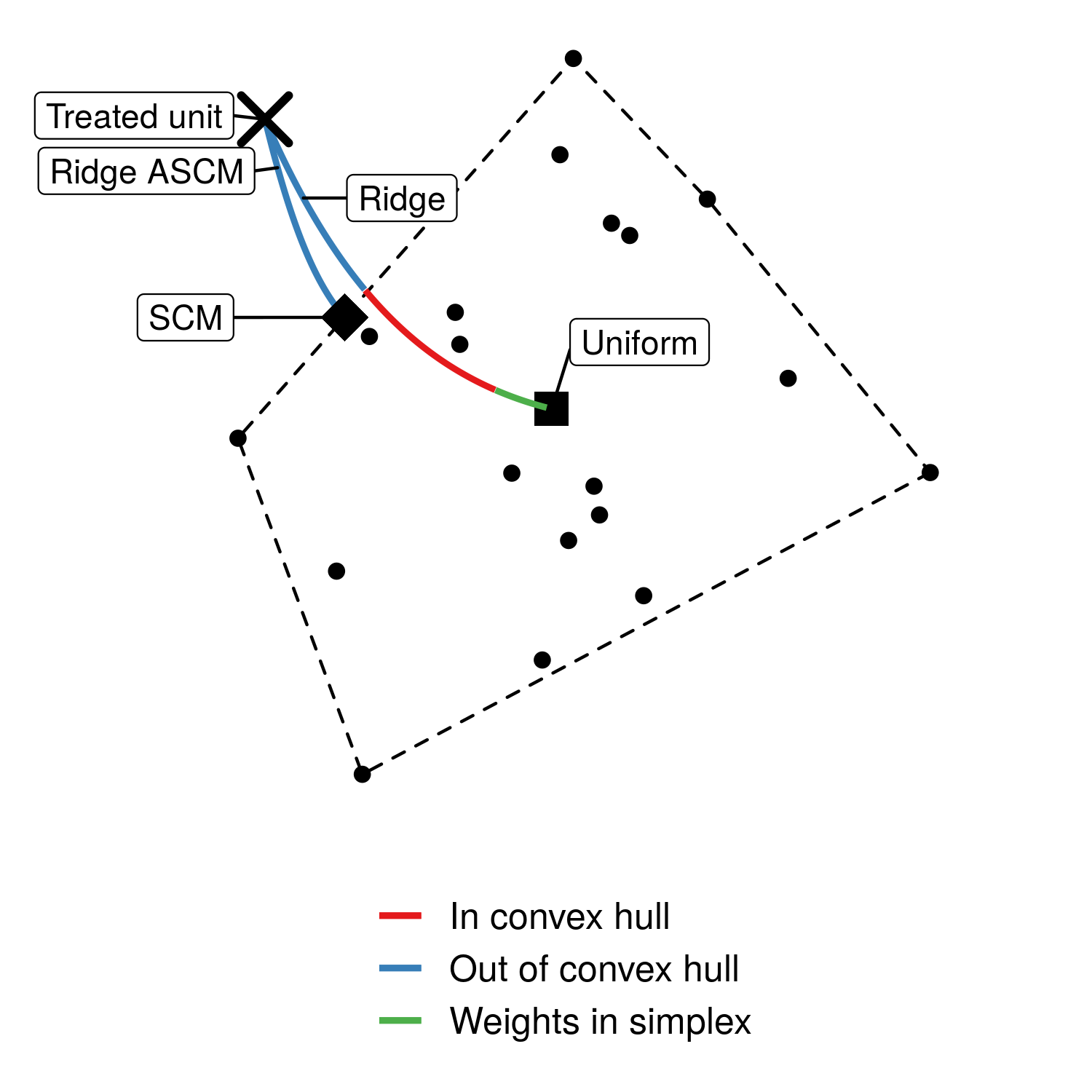} 
  }
  \caption{Treated and control units with the convex hull marked as a dashed line. Ridge and Ridge ASCM estimates in solid.} 
    \label{fig:convex_hull}
    \end{subfigure}%
    \quad
    \begin{subfigure}[t]{0.46\textwidth}  
    {\centering \includegraphics[width=0.8\textwidth]{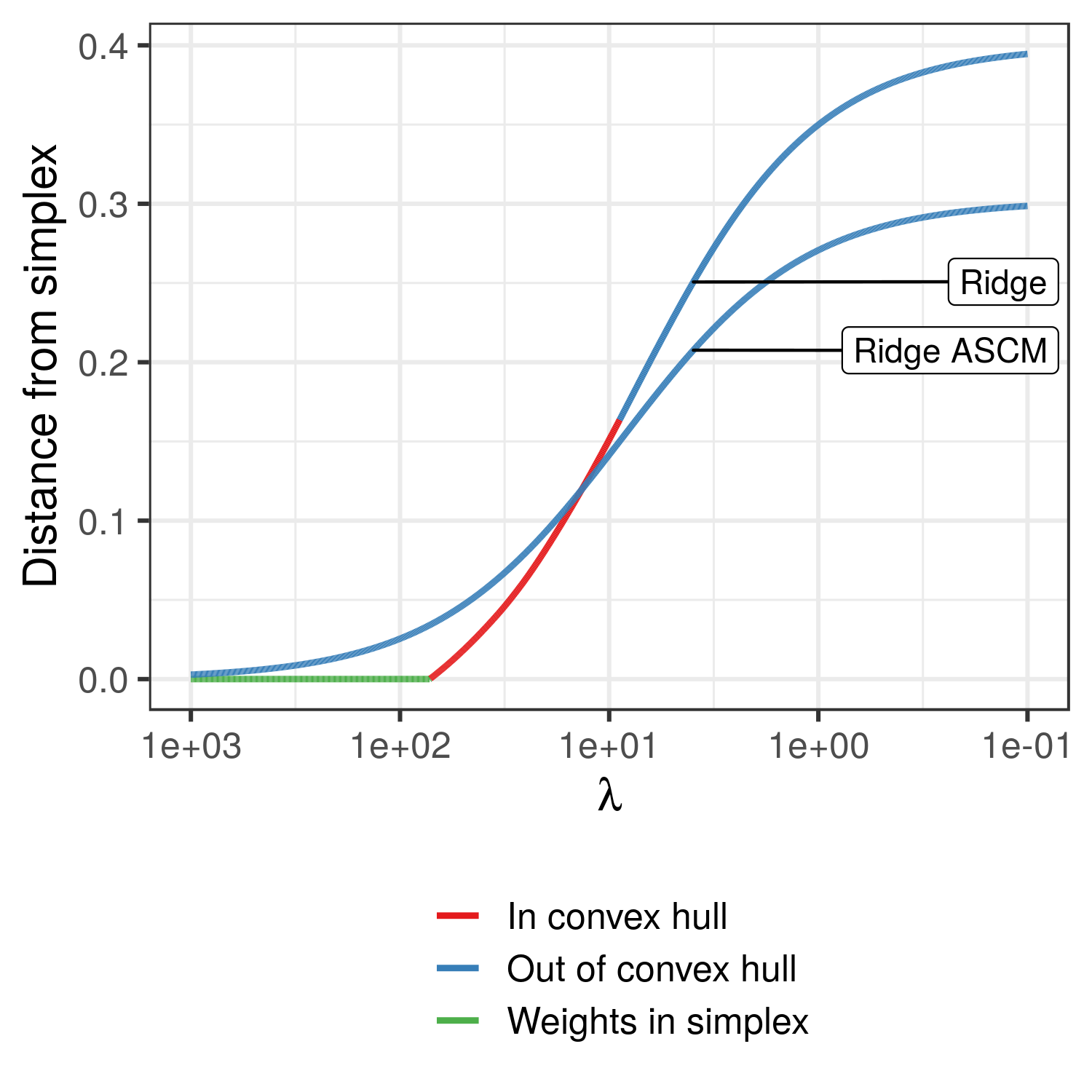} 
    }
    \caption{Distance of ridge and Ridge ASCM weights from the simplex.}
      \label{fig:dist_simplex}
      \end{subfigure}
\caption{Ridge ASCM vs. ridge regression alone for a two-dimensional example with the treated unit outside of the convex hull of the control units. Results shown varying $\lambda^\ridge$ from $10^{3}$ to $10^{-1}$. Green denotes that the weights are inside the simplex, red that the weights are outside the simplex but the weighted average is inside of the convex hull, and blue that the weighted average is outside the convex hull.} 
\label{fig:conceptual_numeric}
\end{figure}

Figure \ref{fig:conceptual_numeric} visualizes these results in two dimensions.
Figure \ref{fig:convex_hull} shows the distribution of control units and the treated unit outside the convex hull, along with the weighted average of control units using the ridge and Ridge ASCM weights, varying $\lambda^\ridge$. We see that ridge regression alone begins (for large $\lambda^\ridge$) at the center of the control units, while Ridge ASCM begins at the SCM solution; both move smoothly towards an exact fit solution as $\lambda^\ridge$ is reduced.
Figure \ref{fig:dist_simplex} shows the distance of the ridge and Ridge ASCM weights from the simplex as $\lambda^\ridge$ varies. Ridge ASCM weights begin at the SCM solution, which is on the boundary of the simplex, then extrapolate outside the convex hull. In contrast, ridge regression weights begin in the center of the simplex (i.e., uniform weights), but then leave the simplex (i.e., some negative weights) before the corresponding weighted average is outside of the convex hull. Over this range, marked as red in Figure \ref{fig:convex_hull}, it is possible to achieve the same level of balance with non-negative weights, but ridge regression uses negative weights in order to reduce the variance. Eventually, as $\lambda^{\ridge} \to 0$, both ridge and Ridge ASCM use negative weights to achieve perfect balance; the weight vectors are different, however, with the Ridge ASCM vector closer to the simplex.

As this example makes clear, both ridge and Ridge ASCM extrapolate from the support of the data to improve pre-treatment fit relative to SCM alone. 
Ridge ASCM, however, is equivalent to SCM weights when SCM achieves exact pre-treatment fit.
When achieving excellent pre-treatment fit with SCM is possible, \citet{Abadie2015} argue that we should prefer SCM weights over possibly negative weights: a slight balance improvement is not worth the extrapolation and the loss of interpretability.
In this case, the Ridge ASCM weights will be close to the simplex, while the ridge regression weights may be quite far away.
When this is not possible, however, and SCM has poor fit, 
some degree of extrapolation is critical; Ridge ASCM allows the researcher to directly penalize the amount of extrapolation in these cases.\footnote{See \citet{king2006dangers} for a discussion of extrapolation in constructing counterfactuals. As they note: ``If we learn that a counterfactual question involves extrapolation, we still might wish to proceed if the question is sufficiently important, but we would be aware of how much more model dependent our answers would be.''}

\subsection{Ridge ASCM improves pre-treatment fit relative to SCM alone}
\label{sec:better_fit}

Just as the hyper-parameter $\lambda^\ridge$ parameterizes the level of extrapolation, it also parameterizes the level of improvement in pre-treatment fit over the SCM solution.
Because we are removing the non-negativity constraint and allowing for extrapolation outside of the convex hull, the pre-treatment fit from Ridge ASCM will be at least as good as the pre-treatment fit from SCM alone, i.e., $\|\bm{X}_{1} - \bm{X}_{0\cdot}'\bm{\hat{\gamma}}^\aug\|_2 \leq \|\bm{X}_{1} - \bm{X}_{0\cdot}'\bm{\hat{\gamma}}^\scm\|_2$. We can exactly characterize the pre-treatment fit of Ridge ASCM using the singular value decomposition of the matrix of control outcomes.
\begin{lemma}
\label{lem:aug_imbal}
Let $\frac{1}{\sqrt{N_0}}\bm{X}_{0\cdot} = \bm{U} \bm{D} \bm{V}'$ 
be the singular value decomposition of the matrix of control pre-intervention outcomes, where $m$ is the rank of $\bm{X}_{0\cdot}$, $\bm{U} \in \R^{N_0 \times m}, \bm{V} \in \R^{T_0 \times m}$, and $\bm{D} = \text{diag}(d_1,\ldots,d_m) \in \R^{m \times m}$ is the diagonal matrix of singular values, where $d_1$ and $d_m$ are the largest and smallest singular values, respectively. Furthermore, let $\bm{\tilde{X}}_i = \bm{V}'\bm{X}_i$ be the rotation of $\bm{X}_i$ along the singular vectors of $\bm{X}_{0\cdot}$. Then $\bm{\hat{\gamma}}^{\aug}$, the Ridge ASCM weights with hyper-parameter $\lambda^\ridge = \lambda N_0$ satisfy
\begin{equation}
    \label{eq:greg_ridge_imbal}
    \begin{aligned}
      \|\bm{X}_{1\cdot} - \bm{X}_{0\cdot}'\bm{\hat{\gamma}}^\aug\|_2 = \lambda \left\|\left(\bm{D} + \lambda \bm{I}\right)^{-1}(\bm{\widetilde{X}}_{1} - \bm{\widetilde{X}}_{0\cdot}'\hat{\gamma}^\scm)\right\|_2 \leq \frac{\lambda}{d_m^2 + \lambda}\|\bm{X}_1 - \bm{X}_{0\cdot}' \hat{\gamma}^\scm\|_2,
        \end{aligned}   
\end{equation}
and the weights from ridge regression alone $\bm{\hat{\gamma}}^\ridge$ satisfy
\begin{equation}
  \label{eq:ridge_imbal}
  \begin{aligned}
    \|\bm{X}_1 - \bm{X}_{0\cdot}'\bm{\hat{\gamma}}^\ridge\|_2 = \lambda \left\|\left(\bm{D} + \lambda \bm{I}\right)^{-1}\bm{\widetilde{X}}_1\right\|_2 \leq \frac{\lambda}{d_m^2 + \lambda}\|\bm{X}_1\|_2.
      \end{aligned}   
\end{equation}
\end{lemma}

\noindent From Equation \eqref{eq:greg_ridge_imbal}, we see that the pre-treatment imbalance for Ridge ASCM weights is smaller than that of SCM weights by at least a factor of $\frac{\lambda}{d_m^2 + \lambda}$.
Thus, Ridge ASCM will achieve strictly better pre-treatment fit than SCM alone, except in corner cases where pre-treatment fit will be equal.
For example, these will be equal when the pre-treatment SCM residual $\bm{X}_1 - \bm{X}_{0\cdot}'\hat{\bm{\gamma}}^\scm$
is orthogonal to the lagged outcomes of the control units $\bm{X}_{0\cdot}$.
From Equation \eqref{eq:ridge_imbal}, we see that the relationship for pre-treatment imbalance fit between SCM and ridge regression is less clear; ridge regression penalizes deviations from uniformity, rather than deviations from SCM weights.

\section{Error under restrictions on the data generating processes}
\label{sec:bias_section}

We now relate Ridge ASCM's improved pre-treatment fit to
improved estimation error under several data generating processes. 
We first consider the case where the post-treatment outcome is linear in the pre-treatment outcomes, and then extend to the canonical linear factor model considered by 
\citet{AbadieAlbertoDiamond2010}.
In the appendix, we also consider a more general formulation, including when the outcome model is approximately linear, with Lipshitz deviations from linearity. 

Under a linear model, improving pre-treatment fit directly 
reduces bias, and the Ridge ASCM penalty term negotiates a bias-variance trade-off.
Under a latent factor model, improving pre-treatment fit again reduces bias,
though there is now a risk of over-fitting.
The penalty term also directly parameterizes this trade-off. 
Thus, choosing the hyper-parameter $\lambda^\ridge$ is important in practice.
In Section \ref{sec:cv}, we describe a cross-validation hyper-parameter selection procedure that builds on the in-time placebo check in \citet{Abadie2015}.
In Section \ref{sec:sim_results} we explore these trade-offs through simulations, finding substantial gains to augmentation via ridge regression in a variety of settings.

\subsection{Error under linearity}
\label{sec:linear}

We first illustrate the key balancing idea in the simple case where the post-treatment outcome is a linear combination of lagged outcomes plus additive noise:
\begin{equation}
  \label{eq:ark}
  Y_{iT}(0) = \sum_{t=1}^{T_0} \beta_t X_{it} + \varepsilon_{iT}.
\end{equation}
A special case of this setup is an auto-regressive process of order $K \leq T_0$.  Here we consider the pre-treatment outcomes $\bm{X}_i$ fixed 
and so the randomness in the post-treatment outcome $Y_{iT}(0)$ is due to the noise term $\varepsilon_{iT}$.
As in Section \ref{sec:scm_overview}, we assume that the $N$ units' noise terms are independent, mean zero sub-Gaussian random variables with scale parameter $\sigma$ that are uncorrelated with treatment assignment $W_i$.

We consider a generic weighting estimator with weights $\bm{\hat{\gamma}}$ that are independent of the post-treatment outcomes $Y_{1T},\ldots,Y_{NT}$; both SCM and Ridge ASCM take this form.
Under linearity, the difference between the counterfactual outcome $Y_{1T}(0)$ and the weighting estimator $\hat{Y}_{1T}(0)$ decomposes into: (1) systemic error due to imbalance in the lagged outcomes $\bm{X}$, and (2) idiosyncratic error due to the noise in the post-treatment period: 
\begin{equation}
  \label{eq:weights_err_lin}
  Y_{1T}(0) - \sum_{W_i = 0} \hat{\gamma}_i Y_{iT} = \underbrace{\bm{\beta} \cdot \left(\bm{X}_1 - \sum_{W_i = 0} \bm{X}_i\right)}_{\text{imbalance in $\bm{X}$}} + \underbrace{\varepsilon_{1T} - \sum_{W_i = 0}\hat{\gamma}_i \varepsilon_{iT}}_{\text{post-treatment noise}}.
\end{equation}
With this setup, a weighting estimator that exactly balances the lagged outcomes $\bm{X}$ will eliminate all systematic error.
Furthermore, if $\bm{\beta}$ is sparse, then it suffices to balance only the lagged outcomes with non-zero coefficients; for example, under an AR($K$) process, $(\beta_1,\ldots,\beta_{T_0 - K - 1}) = 0$, it is sufficient to balance only the first $K$ lags.

If the weighting estimator does not perfectly balance the pre-treatment outcomes $\bm{X}$, there will be a systematic component of the error, with the magnitude depending on the imbalance.
Below we construct a finite sample error bound for Ridge ASCM (and for SCM, the special case with $\lambda^\ridge=\infty$). This bound on the estimation error holds with high probability over the noise in the post-treatment period $\bm{\varepsilon}_T$.
\begin{proposition}
  \label{cor:ascm_error_ar}
  Under the linear model \eqref{eq:ark} with independent sub-Gaussian noise with scale parameter $\sigma$, for any $\delta > 0$, the Ridge ASCM weights with hyperparameter  $\lambda^\ridge = \lambda  N_0$ satisfy the bound
  \begin{equation}
    \label{eq:ascm_err_cor_ar}
  \left|Y_{1T}(0) - \sum_{W_i = 0}\hat{\gamma}^\aug_i Y_{iT}\right| \leq \|\bm{\beta}\|_2\underbrace{\left\|\text{diag}\left(\frac{\lambda}{d_j^2 + \lambda}\right)(\bm{\widetilde{X}}_1 - \bm{\widetilde{X}}_{0\cdot}'\bm{\hat{\gamma}}^\scm)\right\|_2}_{\text{imbalance in } \bm{X}} + \underbrace{\delta\sigma \left(1 + \|\bm{\hat{\gamma}}^\aug\|_2\right)}_{\text{post-treatment noise}},
  \end{equation}
  with probability at least $1 - 2e^{-\frac{\delta^2}{2}}$, where $\bm{\widetilde{X}}_i = \bm{V}'\bm{X}_i$ is the rotation of $\bm{X}_i$ along the singular vectors of $\bm{X}_{0\cdot}$, as above.
\end{proposition}

Proposition \ref{cor:ascm_error_ar} shows the finite sample error of Ridge ASCM weights is controlled by the imbalance in the lagged outcomes and the $L^2$ norm of the weights;
Lemma \ref{lem:aug_variance} in the Appendix gives a deterministic bound for $\|\bm{\hat{\gamma}}^{\aug}\|_2$. See \citet{Athey2016} for analogous results on balancing weights in high dimensional cross-sectional settings.

In the special case that SCM weights have perfect pre-treatment fit, ASCM and SCM weights will be equivalent, and the estimation error will only be due to variance of the weights and post-treatment noise.
When SCM weights do not achieve perfect pre-treatment fit, Ridge ASCM with finite $\lambda < \infty$ extrapolates outside the convex hull, improving pre-treatment fit and thus reducing bias.
This is subject to the usual bias-variance trade-off: The second term in \eqref{eq:ascm_err_cor_ar} is increasing in the $L^2$ norm of the weights, which will generally be larger for ASCM than for SCM. 
The hyperparameter $\lambda$ directly negotiates this trade off.

\subsection{Error under a latent factor model}

We now consider the case where control potential outcomes are generated according to a linear factor model.
Similar to the linear case above,
when pre-treatment fit is poor, improving the pre-treatment fit by extrapolating away from the convex hull reduces the bias,
with a possible bias-variance tradeoff. 
Unlike in the linear setting, however, there is additional 
possible error due to balancing noisy pre-treatment outcomes rather than the underlying factors. This raises the risk of over-fitting, though it can be reduced by appropriate choice of the extrapolation penalty. 
Specifically, in settings where the noise is high, ASCM can increase error relative to SCM alone.
SCM limits over-fitting in this high-noise case by constraining weights to the simplex, although SCM alone is likely to perform poorly here regardless.

Following the setup in \citet{AbadieAlbertoDiamond2010}, 
we assume that there are $J$ unknown, latent time-varying factors $\bm{\mu}_t=\{\mu_{jt}\} \in \R^{T}$, $j=1,\dots,J$,  with $\max_{jt} |\mu_{jt}| \leq M$, where $J$ will typically be small relative to $N$ and $T_0$. In addition, each unit has a vector of unknown factor loadings $\bm{\phi}_i \in \R^J$. We consider both the time-varying factors $\bm{\mu}_t$ and the unit-varying factor loadings $\bm{\phi}_i$ to be non-random quantities and will consider a fixed $N$ and $T_0$.
In this setup, the control potential outcomes are weighted averages of these factors plus additive noise:
\begin{equation}\label{eq:scm_factor_model}
Y_{it}(0) = \sum_{j=1}^J \phi_{ij} \mu_{jt} + \varepsilon_{it} = \bm{\phi}_i \cdot \bm{\mu}_t + \varepsilon_{it},
\end{equation}
where again the randomness in $Y_{it}(0)$ is only due to the noise term $\varepsilon_{it}$.
We assume that $\varepsilon_{it}$ are independent, mean zero sub-Gaussian random variables with scale parameter $\sigma$ that are uncorrelated with treatment assignment $W_i$.
Slightly abusing notation, we collect the pre-intervention factors into a matrix $\bm{\mu} \in \R^{T_0 \times J}$, where the $t$\super{th} row of $\bm{\mu}$ contains the factor values at time $t$, $\bm{\mu}_t^{\prime}$.
Following \citet{bai2009panel} and \citet{Xu2017}, we further assume that the factors are orthogonal and normalized, i.e., that $\frac{1}{T_0}\bm{\mu}'\bm{\mu} = \bm{I}_{J}$. 

Under this model, the finite-sample error of a weighting estimator depends on the imbalance in the latent $\bm{\phi}$ and a noise term due to the noise at time $T$:
\begin{equation}
    \label{eq:weights_err}
    Y_{1T}(0) - \hat{Y}_{1T}(0) = Y_{1T}(0) - \sum_{W_i=0}\hat{\gamma}_i Y_{iT} = \underbrace{\left(\bm{\phi}_1 - \sum_{W_i=0}\hat{\gamma}_i\bm{\phi}_i\right) \cdot \bm{\mu}_T}_{\text{imbalance in } \bm{\phi}} + \underbrace{\varepsilon_{1T} - \sum_{W_i=0}\hat{\gamma}_i \varepsilon_{it}}_{\text{noise}}.
\end{equation}
\noindent 
Balancing the observed pre-treatment outcomes $\bm{X}$ will not necessarily balance the latent factor loadings $\bm{\phi}$.
Following \citet{AbadieAlbertoDiamond2010}, we show in the appendix that, under Equation \eqref{eq:scm_factor_model}, we can decompose the imbalance term as:

\begin{equation}
\label{eq:factor_model_bias_approx}
    \left(\bm{\phi}_1 - \sum_{W_i=0}\gamma_i\bm{\phi}_i\right) \cdot \bm{\mu}_T = \frac{1}{T_0} \bm{\mu}' \underbrace{\left( \bm{X}_1 - \sum_{W_i=0} \gamma_i \bm{X}_i\right)}_{\text{imbalance in } \bm{X}}\cdot \mu_T - \frac{1}{T_0} \bm{\mu}' \underbrace{\left(\bm{\varepsilon}_{1(1:T_0)} - \sum_{W_i=0}\gamma_i\bm{\varepsilon}_{i(1:T_0)}\right) }_{\text{approximation error}}\cdot \bm{\mu}_T,
\end{equation}
where $\bm{\varepsilon}_{i(1:T_0)} = (\varepsilon_{i1},\ldots,\varepsilon_{iT_0})$ is the vector of pre-treatment noise terms for unit $i$.
The first term is the imbalance of observed lagged outcomes and the second term is an approximation error arising from the latent factor structure. 
With $\sigma=0$ and $J$ small, the approximation error is zero, and it is possible to express $Y_{iT}(0)$ as a linear combination of $Y_{it}(0)$, $t=1,\dots,T_0-1$, recovering the linear model \eqref{eq:ark}.
However, with $\sigma>0$ we cannot write the period-$T$ outcome as a linear combination of earlier outcomes plus independent, additive error.

With this setup, we can bound the finite-sample error in Equation \eqref{eq:weights_err} for Ridge ASCM weights (and for SCM weights as a special case). This bound is with high probability over the noise in all time periods $\varepsilon_{it}$, and accounts for the noise in the pre- and post-treatment outcomes separately. 

\begin{theorem}
  \label{thm:ascm_error}
  Under the linear factor model \eqref{eq:scm_factor_model} with independent sub-Gaussian noise, for any $\delta>0$, the Ridge ASCM weights with hyperparameter  $\lambda^\ridge = \lambda N_0$ satisfy the bound
  \begin{equation}
  \label{eq:ascm_error}
  \begin{aligned}
    \left|Y_{1T}(0) - \sum_{W_i=0}\hat{\gamma}^\aug_iY_{1T}(0)\right| & 
    \leq \frac{JM^2}{\sqrt{T_0}}\left(\vphantom{\left\|\text{diag}\left(\frac{\lambda}{d_j^2 + \lambda}\right)(\widetilde{X}_1 - \widetilde{X}_{0\cdot}'\hat{\gamma}^\scm)\right\|_2}\right.
    \underbrace{\left\|\text{diag}\left(\frac{\lambda}{d_j^2 + \lambda}\right)(\widetilde{\bm{X}}_1 - \widetilde{\bm{X}}_{0\cdot}'\bm{\hat{\gamma}}^\scm)\right\|_2}_{\text{imbalance in $\bm{X}$}} + \\[1em] 
    & \qquad\qquad\quad \underbrace{4(1 + \delta)\left\|\text{diag}\left(\frac{d_j\sigma}{d_j^2 + \lambda}\right)(\widetilde{\bm{X}}_1 - \widetilde{\bm{X}}_{0\cdot}'\hat{\gamma}^\scm)\right\|_2}_{\text{excess approximation error}} + \\[1em]
    & \qquad\qquad\quad \underbrace{2\left(\sqrt{\log 2 N_0} + \frac{\delta}{2}\right)}_{\text{SCM approximation error}}  \left.\vphantom{\left\|\text{diag}\left(\frac{\lambda}{d_j^2 + \lambda}\right)(\widetilde{\bm{X}}_1 - \widetilde{\bm{X}}_{0\cdot}'\bm{\hat{\gamma}}^\scm)\right\|_2}\right)\qquad +  \qquad \underbrace{\vphantom{\frac{JM^2}{\sqrt{T_0}}}\delta\sigma \left(1 + \|\bm{\hat{\gamma}}^\aug\|_2\right)}_{\text{post-treatment noise}}
  \end{aligned}
  \end{equation}
  with probability at least $1 - 6e^{-\frac{\delta^2}{2}} - e^{-2(\log 2 + N_0 \log 5)\delta^2}$.
  \end{theorem}
  
Theorem \ref{thm:ascm_error} shows that, relative to the linear case in Proposition \ref{cor:ascm_error_ar}, there is an additional source of error under a latent factor model: approximation error due to balancing lagged outcomes rather than balancing underlying factors.
In particular, it is now possible that a control unit only receives a large weight because of idiosyncratic noise, rather than because of similarity in the underlying factors.
See \citet{arkhangelsky2018synthetic} and \citet{Ferman2019} for asymptotic analogues of this finite sample bound.

In the special case where SCM achieves perfect pre-treatment fit, considered by 
\citet{AbadieAlbertoDiamond2010}, 
the ASCM and SCM weights are equivalent and 
the error is only due to post-treatment noise and the approximation error.
The bound in Theorem \ref{thm:ascm_error} accounts for the worst case scenario where the control unit with the largest weight is only similar to the treated unit due to idiosyncratic noise.
The approximation error, and thus the bias, converges to zero in probability as $T_0 \to \infty$ under suitable conditions on the factor loadings $\bm{\mu}_t$ \citep[see also][]{ferman2018revisiting}.
Intuitively, as we observe more $X_{it}$ --- and can exactly balance each one --- we are better able to match on the index $\bm{\phi}_i \cdot \bm{\mu}_t$ and, as a result, on the underlying factor loadings.\footnote{We show in the supplementary material that with dependent errors the probability of the worst-case error additionally scales with the maximum eigenvalue of the covariance matrix.
Dependence leads to a more complicated error structure overall; we leave a thorough analysis of this to future work.
} 

Without exact balance, Theorem \ref{thm:ascm_error} shows that a long pre-period may not be enough to control the error due to imbalance.
This emphasizes the critical role that conditioning on excellent pre-treatment fit plays in ensuring that SCM alone yields reasonable estimates of the counterfactual.

\begin{figure}[tbp]
  \centering
    \includegraphics[width=\maxwidth]{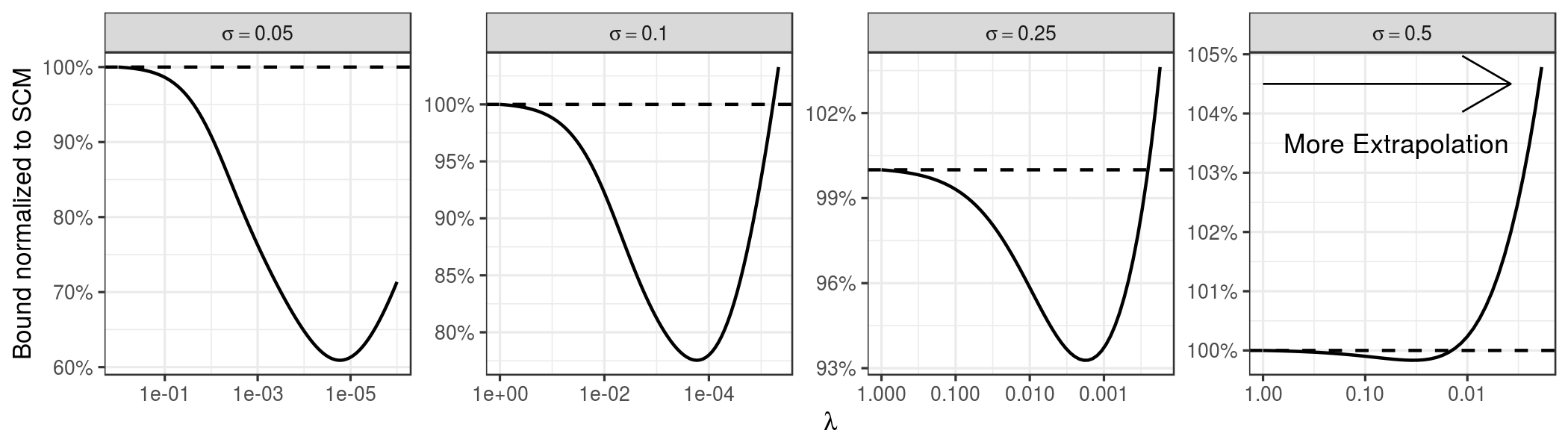} 
\caption{\label{fig:factor_bound} Sketch of the error due to imbalance and approximation error \eqref{eq:ascm_error} for the  linear factor model; the standard deviation of the treated unit's pre-treatment outcomes is normalized to one. We fit SCM weights on the empirical example in Section \ref{sec:empirical} and compute the vector of pre-treatment fit. Each line shows the sum of the error due to imbalance in $\bm{X}$, excess approximation error, and SCM approximation error in Theorem \ref{thm:ascm_error} (with $\delta = 0$) for different values of $\sigma$. These are normalized so that the SCM solution (with $\lambda$ large) equals 100\%; values below 100\% show improvement over the unadjusted weights for a given $\lambda$.}
\end{figure}

When perfect pre-treatment fit is not feasible with SCM, Ridge ASCM
with $\lambda < \infty$ will extrapolate outside the convex hull, where the hyper-parameter $\lambda$ controls the amount of extrapolation.
As in the linear case, extrapolation reduces the error due to imbalance in lagged outcomes. Unlike in the linear case, however, this could possibly 
lead to over-fitting to the noisy lagged outcomes.
While some extrapolation can be helpful on net, the optimal amount will depend on the synthetic control fit and the amount of noise. 
Figure \ref{fig:factor_bound} illustrates this using SCM weights from the empirical example we discuss in Section \ref{sec:empirical}, where pre-treatment fit is relatively poor.
For each value of $\sigma$, the figure plots the sum of the imbalance, SCM approximation error, and excess approximation error terms in the bound in Theorem \ref{thm:ascm_error}.
At each noise level, a small amount of extrapolation leads to a smaller error bound, but as $\lambda$ shrinks there is a point where further extrapolation leads to over-fitting and eventually to a worse error bound than without extrapolation. 
The risk of overfitting is greater when the noise is particularly large, though even here a sufficiently regularized ASCM estimate has lower error bound than SCM alone (represented as the $\lambda \to \infty$ bound). When noise is less extreme, the benefits of augmentation are larger and the optimal amount of regularization shrinks.

It is worth noting that Theorem \ref{thm:ascm_error} gives a worst-case bound. In Section \ref{sec:sim_results} we inspect the typical performance of the Ridge ASCM estimator via extensive simulation studies and find that gains to pre-treatment fit through augmentation outweigh increased approximation error in a range of practical settings, including when noise is very large.

\subsection{Hyper-parameter selection}
\label{sec:cv}

We propose a cross-validation approach for selecting $\lambda$ inspired by the in-time placebo check proposed by \citet{Abadie2015}. 
Let $\hat{Y}_{1k}^{(-t)} = \sum_{W_i=0}\hat{\gamma}^\aug_{i(-t)} Y_{ik}$ be the estimate of $Y_{1k}$ where time period $t$ is excluded from fitting the estimator in \eqref{eq:greg}. \citet{Abadie2015} propose to compare the difference $Y_{1t} - \hat{Y}_{1t}^{(-t)}$ for some $t \leq T_0$ as a placebo check. We can extend this idea to compute the leave-one-out cross validation MSE over time periods:
  \begin{equation}
    \label{eq:loocv}
    CV(\lambda) = \sum_{t=1}^{T_0}\left(Y_{1t} - \hat{Y}_{1t}^{(-t)}\right)^2.
  \end{equation}
We can then choose $\lambda$ to minimize $CV(\lambda)$ or follow a more conservative approach such as the ``one-standard-error'' rule \citep{hastie2009elements}.
This proposal is similar to the leave-one-out cross validation proposed by \citet{Doudchenko2017}, who select hyperparameters by holding out control units and minimizing the MSE of the control units in the post-treatment time $T$.
Finally, only excluding time period $t$ might be inappropriate for some outcome models, e.g. the linear model in Section \ref{sec:linear}. In these settings we can extend the procedure to exclude all time periods $\geq t$ when estimating $\hat{\gamma}_{(-t)}^\aug$. For related proposals, see \citet{kellogg2020combining} and \citet{bilinski2020goldilocks}.

\section{Auxiliary covariates}
\label{sec:extensions}

Thus far, we have focused exclusively on lagged outcomes as predictors. We now consider the case where there are also a small number of auxiliary covariates $\bm{Z}_i \in \R^K$ for unit $i$.
These auxiliary covariates may include summaries of lagged outcomes or time-varying covariates such as the pre-treatment mean $\bar{X}_i$.
Let $\bm{Z}_{0\cdot} \in \R^{N_0 \times K}$ denote the matrix of covariates for the donor units. We assume that the covariates are centered, $\bar{\bm{Z}}_{0\cdot} = \mathbf{0}$.

 These auxiliary covariates can be incorporated both into the balance objective for SCM and into the outcome model used for augmentation in ASCM. For example, we can extend SCM to choose weights to solve
\begin{equation}
\min_{\gamma\in\Delta^{N_0}} \;\;\;\; \theta_x\|\bm{X}_{1} - \bm{X}_{0\cdot}'\bm{\gamma}\|_2^2 + \theta_z\|\bm{Z}_1 - \bm{Z}_{0\cdot}\gamma\|_2^2 \;+\; \zeta \sum_{W_i = 0} f(\gamma_i),
 \label{eq:vanillaSCM_withZ}
\end{equation}
\noindent where $\Delta^{N_0}$ is the $N_0$-simplex.
Similarly, we can augment these weights with an outcome model $\hat{m}(\bm{X}_i, \bm{Z}_i)$ that is a function of both the lagged outcomes and auxiliary covariates. For example, we can extend Ridge ASCM to choose $
\hat{m}(\bm{X}, \bm{Z})=\hat{\eta}_0 + \bm{X}'\hat{\bm{\eta}}_x + \bm{Z}'\hat{\bm{\eta}}_z$ and fit via ridge regression:
\begin{equation}
    \label{eq:cov_regression}
    \min_{\eta_0, \bm{\eta_x}, \bm{\eta_z}} \;\; \frac{1}{2}\sum_{W_i=0}(Y_i - (\eta_0 + \bm{X}_i'\bm{\eta_x} + \bm{Z}_i'\bm{\eta_z}))^2 + \lambda_x \|\bm{\eta_x}\|_2^2 + \lambda_z\|\bm{\eta_z}\|_2^2.
\end{equation}
Both this SCM criterion and augmentation estimator incorporate user-specified weights that determine
the importance of balancing each set of covariates (Equation \ref{eq:vanillaSCM_withZ}) or 
the amount of regularization for each set of coefficients (Equation \ref{eq:cov_regression}). 
There are many potential choices for these weights. We discuss two, appropriate to different settings depending on the number of auxiliary covariates.

A sensible default when the dimension of the auxiliary covariates is moderate is to incorporate the lagged outcomes $\bm{X}$ and the auxiliary covariates $\bm{Z}$ equally in Equations \eqref{eq:vanillaSCM_withZ} and \eqref{eq:cov_regression}, setting $\theta_x = \theta_z = 1$ and $\lambda_x = \lambda_z = \lambda^\ridge$, after standardizing auxiliary covariates and lagged outcomes to have the same standard deviation. 
With this setup the numerical and algorithmic results in Section \ref{sec:ridge_ascm} apply for the combined vector of lagged outcomes and auxiliary covariates, $(\bm{X}_i, \bm{Z}_i) \in \R^{T_0 + K}$. In particular, Ridge ASCM is again a penalized SCM estimator that adjusts the synthetic control weights that solve optimization problem \eqref{eq:vanillaSCM_withZ} to achieve better balance by extrapolating outside of the convex hull.

An alternative approach when the dimension of the auxiliary covariates is small relative to $N$ (i.e., $K \ll N$) is to fit a regression model 
that regularizes the lagged outcome coefficients $\bm{\eta_x}$ but does \emph{not} regularize the auxiliary covariate coefficients $\bm{\eta_z}$ (i.e., set $\lambda_z = 0$).
Lemma \ref{lem:aug_aux} below writes the resulting augmented estimator as its corresponding penalized SCM optimization problem, with weights that perfectly balance the auxiliary covariates. This has two key implications. First, since the auxiliary covariates $\bm{Z}$ are exactly balanced regardless of the balance that the SCM weights achieve alone, we can exclude them from the optimization problem \eqref{eq:vanillaSCM_withZ}.
Second, as we show below, the pre-treatment fit on the lagged outcomes depends on how well the SCM weights balance the residualized lagged outcomes $\check{\bm{X}}$.
This suggests modifying Equation \eqref{eq:vanillaSCM_withZ} to balance $\check{\bm{X}}$ rather than the lagged outcomes $\bm{X}$, which leads to the two-step procedure:
(1) residualize the pre- and post-treatment outcomes on the auxiliary covariates $\bm{Z}$; and (2) estimate Ridge ASCM on the residualized outcomes.
This two-step procedure follows from a related proposal in \citet{Doudchenko2017}.

\begin{lemma}
  \label{lem:aug_aux}
  Let $\hat{\bm{\eta}}_x$ and $\hat{\bm{\eta}}_z$ be the solutions to \eqref{eq:cov_regression} with $\lambda_x = \lambda^\ridge$ and $\lambda_z = 0$. For any weight vector $\bm{\hat{\gamma}}$ that sums to one, the ASCM estimator from Equation \eqref{eq:ascm_1} with $\hat{m}(\bm{X}_i, \bm{Z}_i) = \bm{X}_i^\prime \bm{\hat{\eta}_x} +\bm{Z}_i^{\prime}\bm{\hat{\eta}_z}$ is
  \begin{equation}
      \label{eq:ols_scm}
   \sum_{W_i=0} \hat{\gamma}_i Y_{iT}\;\; + \left(\bm{X}_{1} -  \sum_{W_i=0} \hat{\gamma}_i \bm{X}_{i} \right)^{\prime}\bm{\hat{\eta}_x}\; + \;  \left(\bm{Z}_{1} -  \sum_{W_i=0} \hat{\gamma}_i \bm{Z}_{i} \right)^{\prime}\bm{\hat{\eta}_z} = \sum_{W_i=0}\hat{\gamma}_i^{\cov}Y_{iT},   
  \end{equation}
  where the weights $\bm{\hat{\gamma}}^\cov$ are 
  \begin{equation}
      \label{eq:ols_scm_weights}
      \hat{\gamma}_i^{\cov} = \hat{\gamma}_i + (\check{\bm{X}}_1 - \check{\bm{X}}_{0\cdot})(\check{\bm{X}}_{0\cdot}'\check{\bm{X}}_{0\cdot} + \lambda^\ridge \bm{I}_{T_0})^{-1} \check{\bm{X}}_i + (\bm{Z}_1 - \bm{Z}_{0\cdot}'\gamma)'(\bm{Z}_{0\cdot}'\bm{Z}_{0\cdot})^{-1} \bm{Z}_i,
      \end{equation}
      and $\check{\bm{X}}_i$ is the residual components of a regression of pre-treatment outcomes on the control auxiliary covariates:
  \begin{equation}
     \check{\bm{X}}_i = \bm{X}_i - \bm{Z}_i'(\bm{Z}_{0\cdot}^{\prime}\bm{Z}_{0\cdot})^{-1}\bm{Z}_{0\cdot}^{\prime}\bm{X}_{0\cdot}.
  \end{equation}
  These weights exactly balance the auxiliary covariates, $\bm{Z}_1 - \bm{Z}_{0\cdot}'\bm{\hat{\gamma}}^{\cov} = 0$; the imbalance in the lagged outcomes is
  \begin{equation}
      \label{eq:aux_imbal}
       \left\|\bm{X}_1 - \bm{X}_{0\cdot}'\bm{\hat{\gamma}}^{\cov}\right\|_2 \leq \left(\frac{\lambda^\ridge}{\lambda^\ridge + N_0\check{d}_r^2}\right)\left\|\check{\bm{X}}_1 - \check{\bm{X}}_{0\cdot}'\hat{\bm{\gamma}}\right\|_2,
  \end{equation}
  where $\check{d}_r$ is the minimal singular value of $\check{\bm{X}}_0$.
  \end{lemma}
Comparing to the numerical results in Section \ref{sec:ridge_ascm}, 
Lemma \ref{lem:aug_aux} shows that the two-step approach penalizes extrapolation from the convex hull \emph{in the residualized space} $\check{\bm{X}}$, rather than in the lagged outcomes themselves. In essence, by residualizing out the auxiliary covariates $\bm{Z}$ the two-step approach allows for a possibly large amount of extrapolation in the auxiliary covariates, while carefully penalizing extrapolation in the part of the lagged outcomes that is orthogonal to the covariates.

In the Appendix, 
we consider the performance of this estimator when the outcomes follow a linear factor model with covariates for the special case where $\lambda^\ridge \to \infty$ and the weights $\hat{\bm{\gamma}}^\cov$ do not extrapolate from the convex hull after residualization. 
We find that, 
if $K$ is small relative to $N_0$,
 exactly balancing a small number of auxiliary covariates and targeting imbalance in the residuals $\check{\bm{X}}$ can lead to decreased error due to pre-treatment fit, with only a small increase in approximation error and error due to post-treatment noise.
However, with larger numbers of auxiliary covariates, the approach that incorporates auxiliary covariates in parallel to lagged outcomes would be more appropriate.

\section{Simulations and empirical illustrations}
\label{sec:empirical}

We now turn to simulations and an empirical illustration. First, we conduct extensive simulation studies to assess the performance of different methods, finding substantial gains from ASCM.
We then use our approach to examine the effect of an aggressive tax cut on economic output in Kansas in 2012.


\subsection{Calibrated simulation studies}
\label{sec:sim_results}

We now present simulation studies calibrated to our empirical illustration in Section \ref{sec:illustration}. 
Specifically, we use the Generalized Synthetic Control Method \citep{Xu2017} to estimate a factor model with three latent factors based on the series of log GSP per capita ($N=50$, $T_0=89$).  
We then simulate outcomes using the distribution of estimated parameters and model selection into treatment as a function of the latent factors; see Appendix \ref{sec:sim_details} for additional details. We also present results from three additional DGPs, each calibrated to estimates from the same data: (1) the factor model with quadruple the standard deviation of the noise term, (2) a unit and time fixed effects model, and (3) an autoregressive model with 3 lags.

We explore the role of augmentation using simple outcome estimators. For each DGP, we consider five estimators: (1) SCM alone, (2) ridge regression alone, (3) Ridge ASCM, (4) fixed effects alone, and (5) SCM augmented with fixed effects (i.e., de-meaned SCM), as shown in Equation \eqref{eq:demeaned_scm_tau}.
Figure \ref{fig:overview_bias_plots} shows the Monte Carlo estimate of the absolute bias as a percentage of the absolute bias for SCM, with one panel for each simulation DGP; Appendix Figure \ref{fig:overview_rmse_plots} shows the corresponding estimator root mean squared error (RMSE).

\begin{figure}[!bt]
{\centering \includegraphics[width=\maxwidth]{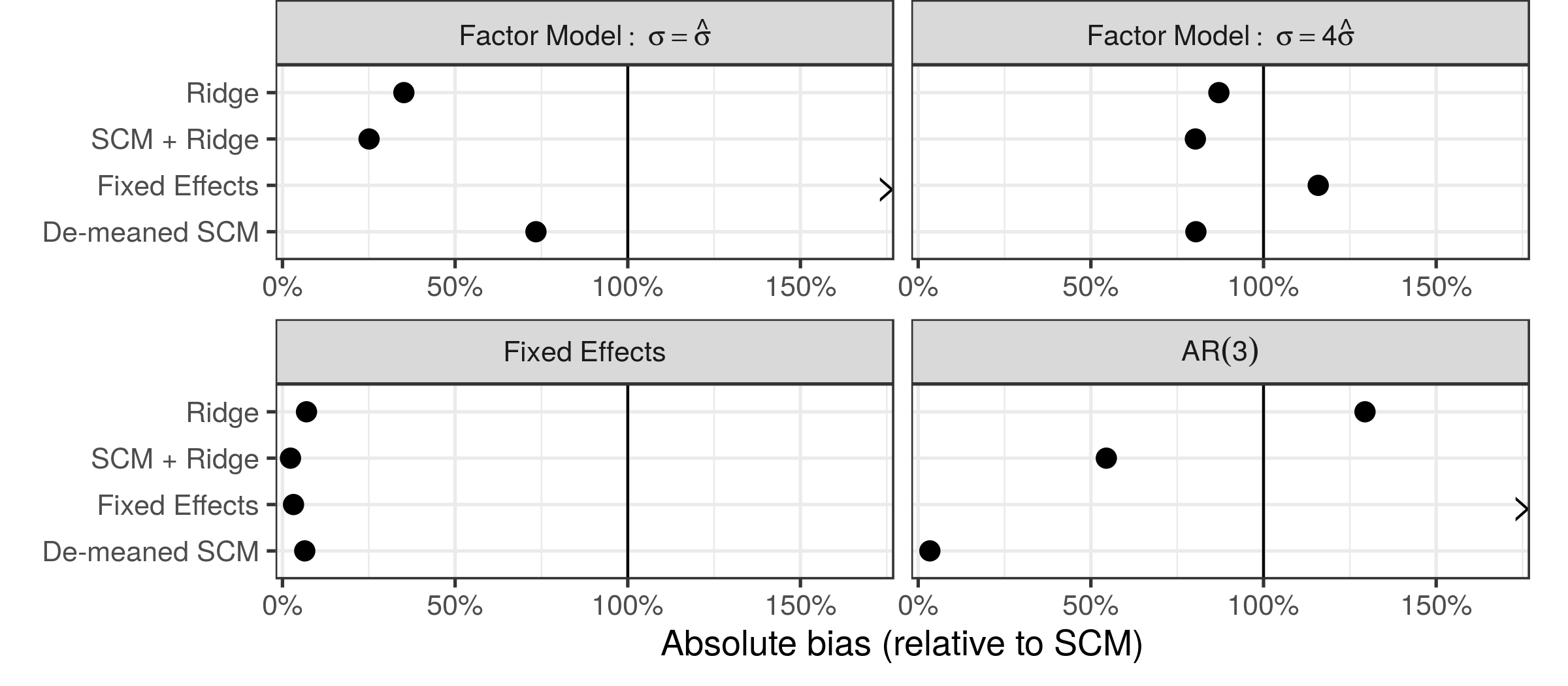} 
}
\caption{\label{fig:overview_bias_plots} Overall absolute bias, normalized to SCM bias for (a) the factor model simulation, (b) the factor model simulation with quadruple the standard deviation, (c) the fixed effects simulation, and (d) the AR simulation. The SCM estimates reported here are \emph{not} restricted to simulation draws with excellent pre-treatment fit; \citet{Abadie2015} advise against using SCM in such settings.}
\end{figure}

There are several takeaways.
First, 
augmenting SCM with a ridge outcome regression reduces bias relative to SCM alone --- \emph{without} conditioning on excellent pre-treatment fit --- in all four simulations.
This underscores the importance of the recommendation in \citet{AbadieAlbertoDiamond2010, Abadie2015} to use SCM only in settings with excellent pre-treatment fit.\footnote{\citet{AbadieAlbertoDiamond2010, Abadie2015} also strongly recommend incorporating auxiliary covariates, weighted by their predictive power, into the procedure, noting that this is important for further reducing bias. For simplicity, the simulations do not include auxiliary covariates.}
Under the baseline factor model and the fixed effect model, the ridge augmentation greatly reduces bias, by more than 75\% in the factor model simulation and over 90\% in the fixed effects simulation. 
In the AR(3) model and in the factor model with greater noise, the gains to augmentation relative to SCM are more limited.
Second, Ridge ASCM has slightly lower bias than ridge regression alone across all of the simulation settings.
Third, when the fixed effects estimator is incorrectly specified, combining SCM with a fixed effects estimator has much lower bias than either method alone. And even when the fixed effects estimator is correctly specified, de-meaned SCM has similar bias to the (correctly specified) fixed effects approach.
Finally, Appendix Figure \ref{fig:overview_rmse_plots} shows that in all simulations ASCM has lower RMSE than SCM, as the large decrease in bias more than makes up for the slight increase in variance.

\begin{figure}[!bt]
{\centering \includegraphics[width=\maxwidth]{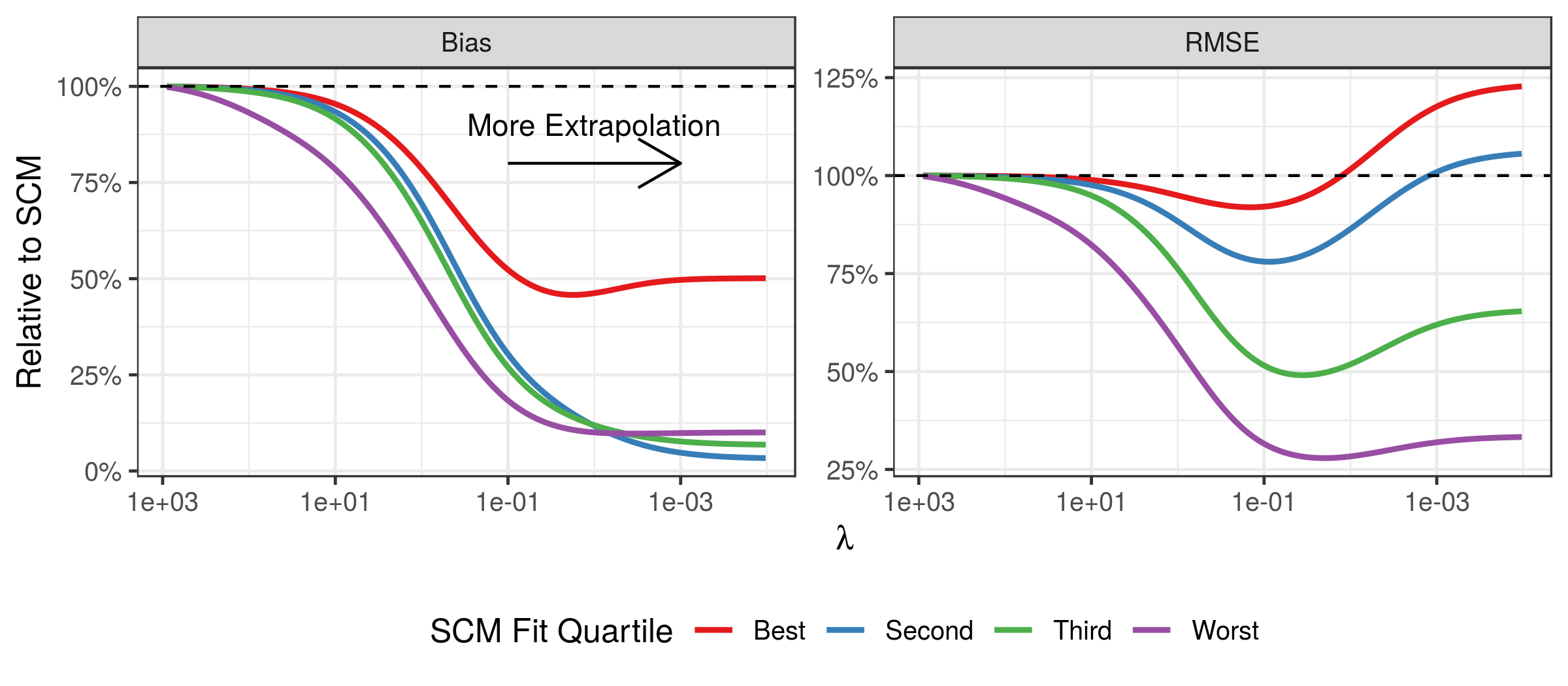} 
}
\caption{\label{fig:bias_vs_lambda} Bias and RMSE of Ridge ASCM, as a percentage of SCM bias and RMSE, versus $\lambda$ under a linear factor model. Results are divided by the quartile of the SCM fit across all simulations.}
\end{figure}

Complementing the worst-case analysis in Section \ref{sec:bias_section}, we now consider how the typical performance of augmentation relates to the amount of extrapolation and the quality of the original SCM fit. Figure \ref{fig:bias_vs_lambda} shows the bias and RMSE as a function of $\lambda$ for the primary factor model simulation, conditional on the quartile of SCM fit. Larger values of $\lambda$ (and hence smaller adjustments) are to the left, with the left-most points in the plots representing SCM. 
First, as expected, Augmented SCM substantially reduces bias regardless of SCM pre-treatment fit. However, the gains are more modest when the SCM fit is in the best quartile: 
in this case the bias is non-monotonic in $\lambda$ and there is some optimal choice of $\lambda$ that minimizes the bias.
Second, it is possible to under-regularize with ASCM, as evident in 
the RMSE achieving a minimum for an intermediate value of $\lambda$. Furthermore, when pre-treatment fit is good, augmentation with too-small $\lambda$ leads to higher RMSE than SCM alone.

Next, we evaluate alternative outcome models for use in ASCM. For each DGP we consider SCM augmented with (1) LASSO, (2) a random forest, (3) \texttt{CausalImpact} \citep{Brodersen2015}, (4) matrix completion using \texttt{MCPanel} \citep{athey2017mcp} and (5) fitting the factor model directly with \texttt{gsynth} \citep{Xu2017}. We compare ASCM to the pure outcome models as well as pure SCM.
Figure \ref{fig:bias_ml} shows the absolute bias for these methods, again as a percentage of the absolute bias for SCM alone; Appendix Figure \ref{fig:overview_rmse_plots} shows the RMSE. We broadly see the same results as with Ridge ASCM. In our simulations, augmenting SCM almost always reduces the bias relative to SCM (unconditional on good pre-treatment fit) with some models improving SCM more than others. Additionally, in nearly every case ASCM also has lower bias than outcome modeling alone. Appendix Figures \ref{fig:bias_good_fit} and \ref{fig:rmse_good_fit} show the bias and RMSE when SCM fit is in the top quintile. As before, conditioned on good SCM pre-treatment fit the gains to augmentation with flexible outcome models are more limited, except with the oracle \texttt{gsynth} estimator.

Overall we find that SCM augmented with a penalized regression model has consistently good performance across data generating processes.
Due to this performance and the method's relative simplicity, 
we therefore recommend augmenting SCM with penalized regression as a reasonable default in settings where SCM alone has poor pre-treatment fit. In particular, we suggest using ridge regression; among the other benefits, Ridge ASCM allows the practitioner to diagnose the level of extrapolation due to the outcome model.

\begin{figure}[!tb]
{\centering \includegraphics[width=\maxwidth]{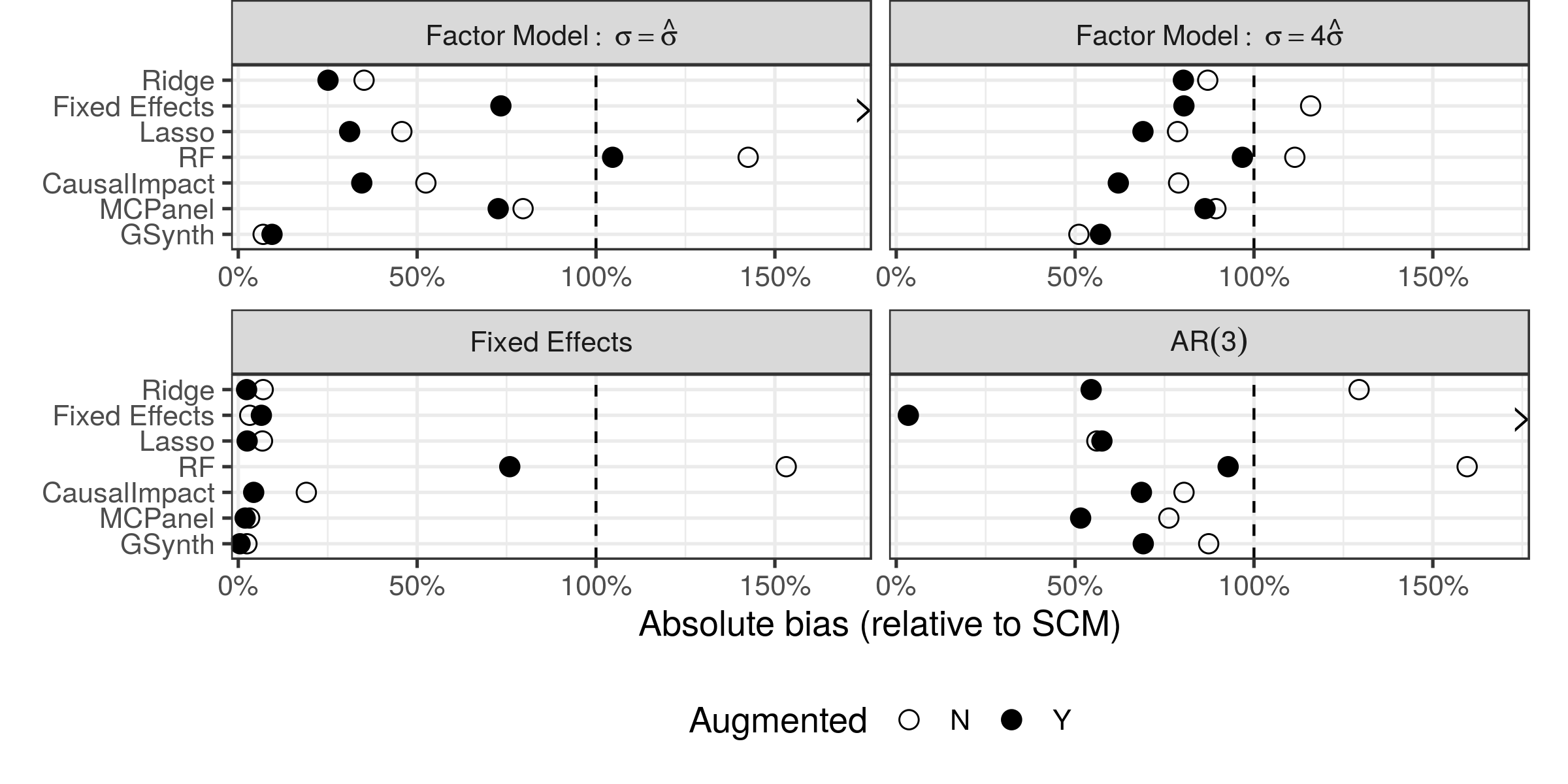} 
}
  \caption{Absolute bias (as a percentage of SCM bias) for ridge, fixed effects, and several machine learning and panel data outcome models, and their augmented versions using the same data generating processes as Figure \ref{fig:overview_bias_plots}.}
  \label{fig:bias_ml}
\end{figure}

\subsection{Illustration: 2012 Kansas tax cuts}
\label{sec:illustration}

In 2010, Sam Brownback was elected governor of Kansas, having run on a platform emphasizing tax cuts and deficit reduction \citep[see][for further discussion and analysis]{rickman2018two}.
Upon taking office, he implemented a substantial personal income tax cut, both lowering rates and reducing credits and deductions. 
This is a valuable test of ``supply side'' models: Brownback argued that the tax cuts would increase business activity in Kansas, generating economic growth and additional tax revenues that would make up for the static revenue losses. Kansas' subsequent economic performance has not been impressive relative to its neighbors; however, potentially confounding factors include a drought and declines in the locally important aerospace industry. Finding a credible control for Kansas is thus challenging, and SCM-type models offer a potential solution.

We estimate the effect of the tax cuts on log gross state product (GSP) per capita using the second quarter of 2012 --- when Brownback signed the tax cut bill into law --- as the intervention time. Results are consistent using outcomes scales other than the standard normalization of log GSP per capita (see Appendix \ref{sec:additional_plots}).
We use four primary estimators: (1) SCM alone fit on the entire vector of lagged outcomes, (2) Ridge ASCM, (3) Ridge ASCM including auxiliary covariates in parallel to lagged outcomes and (4) Ridge ASCM on residualized outcomes, as in Section \ref{sec:extensions}.%
\footnote{The covariates we include are the pre-treatment averages of (1) log state and local revenue per capita, (2) log average weekly wages, (3) number of establishments per capita, (4) the employment level, and (5) log GSP per capita. 
For the augmented estimator on the lagged outcomes we select the hyperparameter $\lambda^\ridge$ as the largest $\lambda$ within one standard error of the $\lambda$ that minimizes the cross-validation placebo fit $CV(\lambda)$; see Section \ref{sec:cv}.
Appendix Figure \ref{fig:lambda_cv} plots $CV(\lambda)$. When including the auxiliary covariates we use the minimal $\lambda$.}
These estimators rely on the ignorability assumption in Equation \eqref{eq:ignore}; substantively, this assumes that post-treatment shocks for Kansas will be the same as for other states in expectation. This also rules out unobserved confounders that affect both post-treatment shocks and the decision to enact the Brownback tax cut bill.

Figure \ref{fig:synth_estimates}, known as a ``gap plot'', shows the difference between Kansas and its synthetic control before and after the passage of the tax cuts along with 95\% point-wise confidence intervals computed via the conformal inference procedure from \citet{chernozhukov2017exact}; see Appendix \ref{sec:inference}. 
Appendix \ref{sec:additional_plots} shows additional results, including estimates plotted on the raw outcome scale and results with alternative estimators. 
First, SCM alone achieves fairly poor pre-treatment fit; the synthetic control exceeds the treatment unit by two to four percent in 2004--2005, on the same scale as the estimated average post-treatment effect of a 3 percent decrease. 
This lack of pre-treatment fit should make us wary of the validity of the SCM effect estimates, and suggests that there may be gains to augmentation. 
Augmenting SCM with ridge regression indeed improves pre-treatment fit, especially in the mid 2000s.
To better understand this, we can inspect the ridge regression coefficients for lagged outcomes (see Appendix Figure \ref{fig:ridge_coefs}), which put the most weight on the two most recent years. 
Adding auxiliary covariates and augmenting further improves both pre-treatment fit and balance on the covariates; see Figure \ref{fig:cov_imbalance}.
Finally, balancing the auxiliary covariates via residualization also improves pre-treatment fit.
Overall, the estimated impact is consistently negative for all four approaches, with weaker evidence that the effect persists to the end of the observation period.

To check against over-fitting, Figure \ref{fig:lngdpcapita_placebo_covascm} shows in-time placebo estimates for the Ridge ASCM estimator with covariates, with placebo treatment times in the second quarter of 2009, 2010, and 2011. 
We estimate placebo effects that are near zero with all three placebo treatment times.
Appendix Figures \ref{fig:lngdpcapita_placebo_scm} and \ref{fig:lngdpcapita_placebo_ascm} show the corresponding placebo estimates for SCM alone and Ridge ASCM without covariates.


\begin{figure}
{\centering \includegraphics[width=\maxwidth]{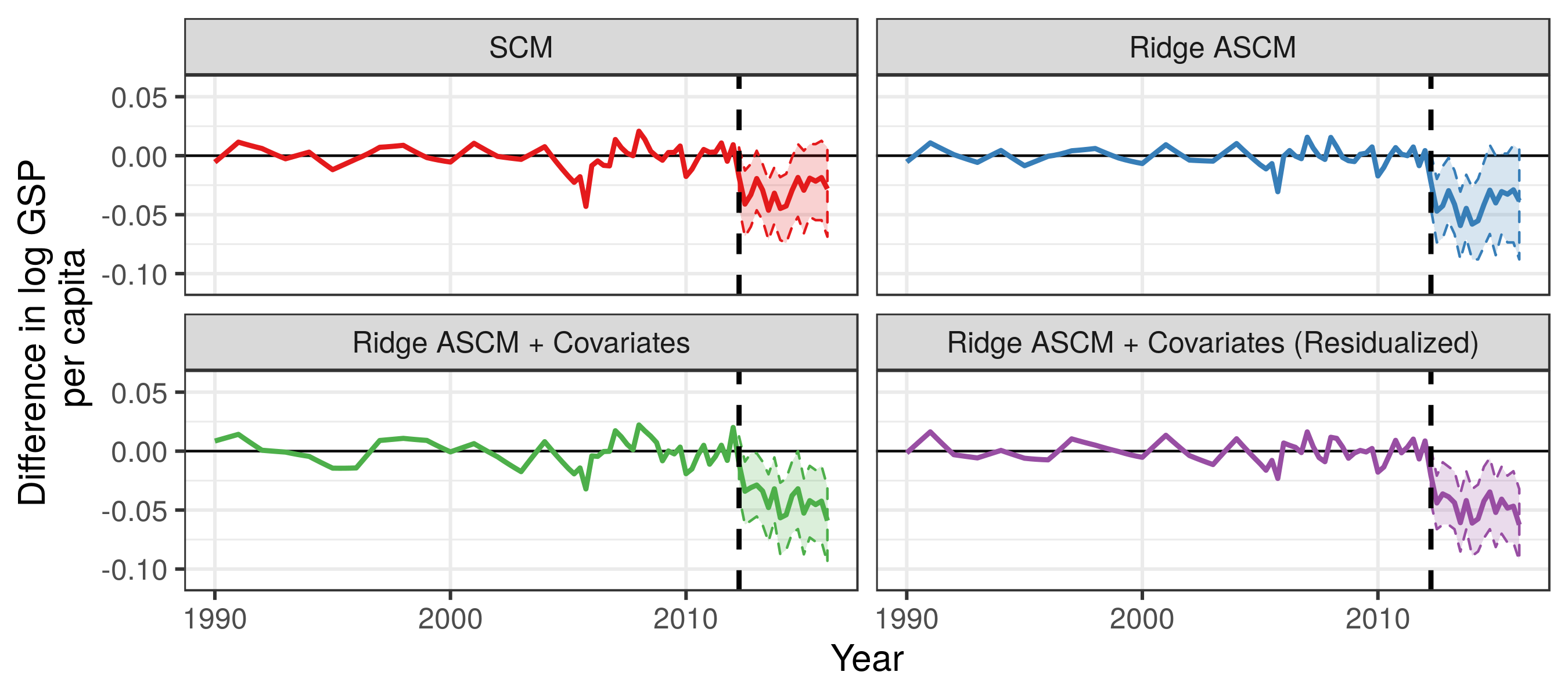}}
  \caption{Point estimates along with point-wise 95\% conformal confidence intervals for the effect of the tax cuts on log GSP per capita using SCM, Ridge ASCM, and Ridge ASCM with covariates.}
    \label{fig:synth_estimates}
\end{figure}

\begin{figure}
{\centering \includegraphics[width=\maxwidth]{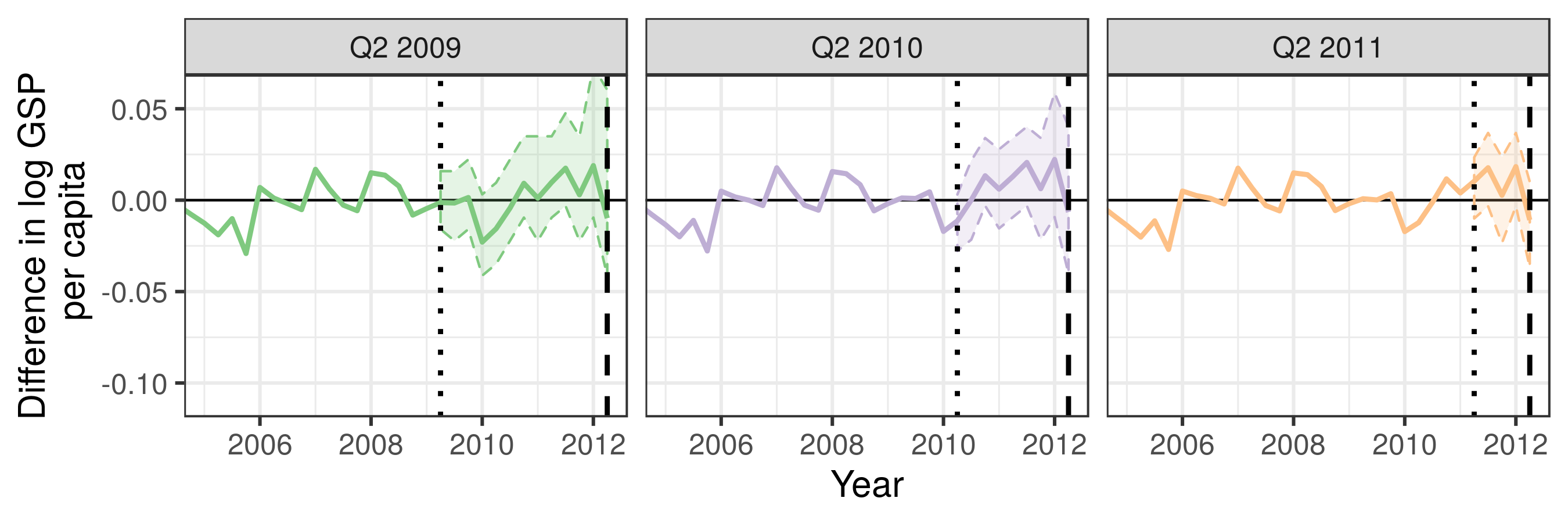}}
  \caption{Placebo point estimates along with 95\% conformal confidence intervals for Ridge ASCM with covariates with placebo treatment times in Q2 2009, 2010, and 2011. The time period begins in 2005 and ends in Q1 2012 to highlight placebo estimates.}
    \label{fig:lngdpcapita_placebo_covascm}
\end{figure}

Figure \ref{fig:cov_imbalance} shows the covariate balance for the four estimators. While SCM and Ridge ASCM achieve excellent fit for the pre-treatment average log GSP per capita, neither estimator achieves good balance on the other covariates, most notably the average employment level across the quarters of the pre-period. 
In contrast, including the auxiliary covariates into both the SCM and ridge optimization problems greatly improves the covariate balance, and --- by design --- 
residualizing on the auxiliary covariates perfectly balances them. Moreover, Ridge ASCM on residualized outcomes 
achieves very good pre-treatment fit on the lagged outcomes as shown in Figure \ref{fig:synth_estimates}.

\begin{figure}
\begin{subfigure}[t]{0.47\textwidth}
{\centering \includegraphics[width=\maxwidth]{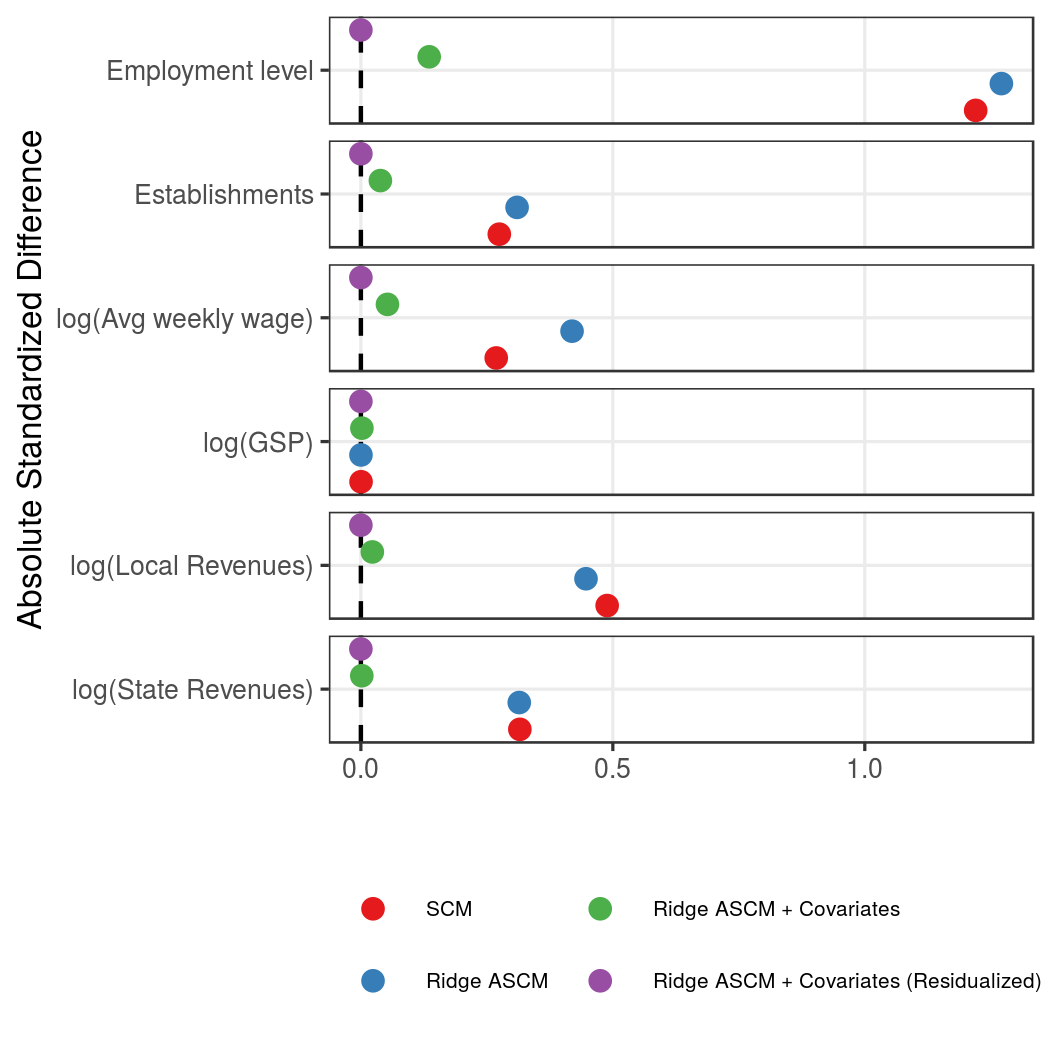} 
\caption{\label{fig:cov_imbalance}}
}
\end{subfigure}%
\qquad \begin{subfigure}[t]{0.47\textwidth}    
{\centering \includegraphics[width=\maxwidth]{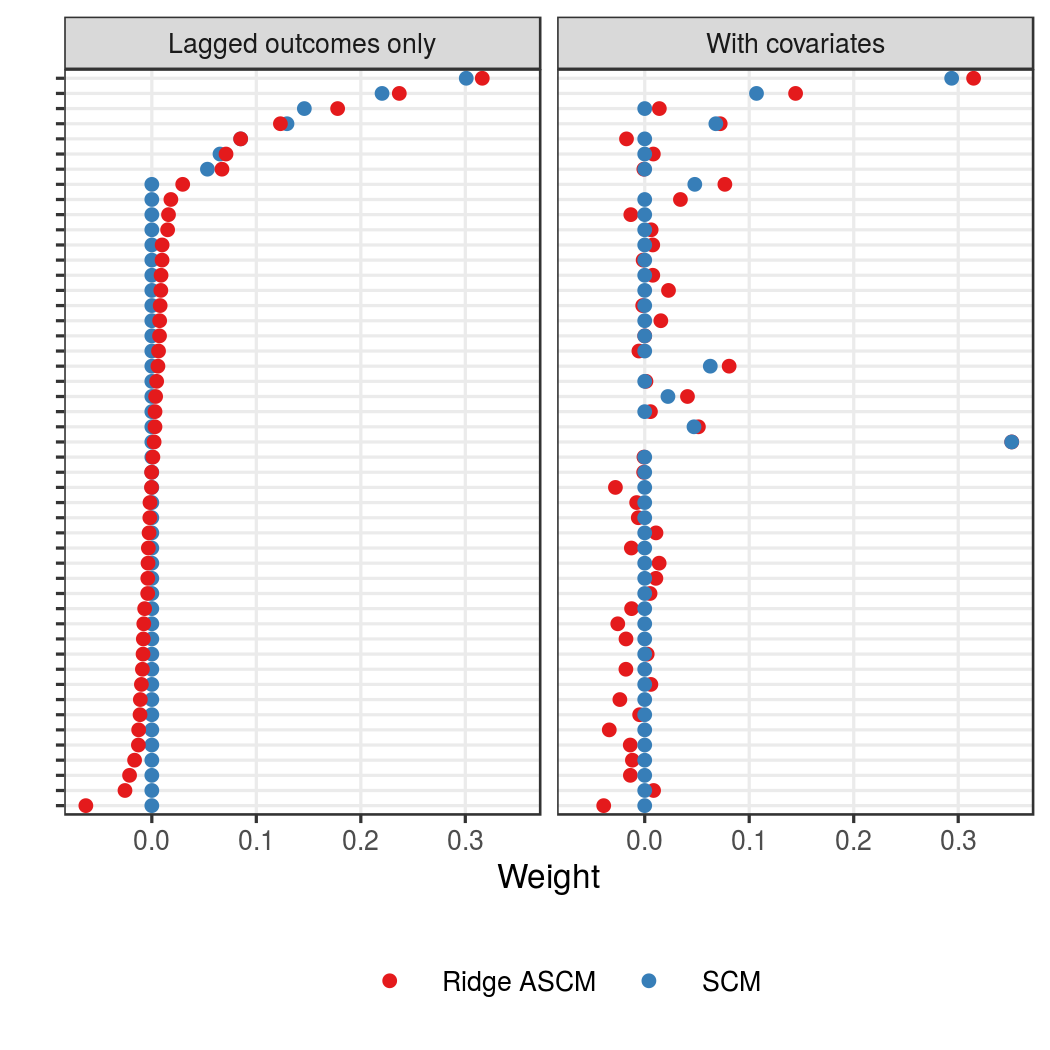} 
\caption{\label{fig:weight_plot}}
}
\end{subfigure}
\caption{(a) Covariate balance for SCM, Ridge ASCM, and ASCM with covariates. Each covariate is standardized to have mean zero and standard deviation one; we plot the absolute difference between the treated unit's covariate and the weighted control units' covariates $\left |Z_{1k} - \sum_{W_i=0} \hat{\gamma}Z_{ik}\right|$. (b) Donor unit weights for SCM, Ridge ASCM, using lagged outcomes only or including auxiliary covariates.}
\label{fig:fit_cov_wts}
\end{figure}

Finally, Figure \ref{fig:weight_plot} shows the weights on donor units for SCM and Ridge ASCM as well as SCM and Ridge ASCM weights when including covariates jointly with the lagged outcomes. Appendix Figure \ref{fig:weight_plot_resid} shows the SCM and Ridge ASCM weights fit on the lagged outcomes after residualizing out the auxiliary covariates.
Here we see the minimal extrapolation property of the ASCM weights. The SCM weights are zero for all but six donor states. The Ridge ASCM weights are similar but deviate slightly from the simplex, with only Louisiana receiving a meaningful negative weight. 
In addition, the Ridge ASCM weights retain some of the interpretability of the SCM weights. For the donor units with positive SCM weight, Ridge ASCM places close to the same weight. For the majority of those with zero SCM weight, Ridge ASCM also places a close to zero weight, with relatively few donor units with non-negligible negative weight. By contrast, Appendix Figure \ref{fig:weight_plot_appendix} shows the weights from ridge regression alone: many of the weights are negative and the weights are far from sparse.
Including auxiliary covariates changes the relative importance of different states by adding new information, but the minimal extrapolation property remains.

\section{Discussion}\label{sec:discussion}

SCM is a popular approach for estimating policy impacts at the jurisdiction level, such as the city or state. By design, however, the method is limited to settings where excellent pre-treatment fit is possible. 
For settings when this is infeasible, we introduce Augmented SCM, which controls pre-treatment fit while minimizing extrapolation. 
We show that this approach controls error under a linear factor model and propose several extensions, including to incorporate auxiliary covariates.

There are several directions for future work. 
The most immediate is to explore robust inferential methods for settings where pre-treatment SCM fit is imperfect. In Appendix \ref{sec:inference}, we outline how to apply the conformal inference approach of \citet{chernozhukov2017exact} to Augmented SCM, as well as a possible modification based on the jackknife+ approach of \citet{Barber2019}. More work is needed, however, to understand the performance of this in more general settings. 
We could also extend this to allow for sensitivity analysis that directly parameterizes departures from, say, the linear factor model, as in recent approaches for sensitivity analysis for balancing weights \citep{soriano2019sensitivity}.

A second area for future inquiry is the application of the ASCM framework to settings with multiple treated units.
For instance, there are different approaches in settings when all treated units are treated at the same time: some papers propose to fit SCM separately for each treated unit \citep[e.g.,][]{Abadie_LHour}, while others simply average the units together \citep[e.g.,][]{kreif2016examination, Robbins2017}. The situation is more complicated with staggered adoption, when units take up the treatment at different times \citep[e.g.,][]{dube2015pooling, donohue2017right}. We explore this extension in \citet{benmichael2019multisynth}.

A third potential extension is to more complex data structures, such as applications with multiple outcomes series for the same units (e.g., measures of both earnings and total employment in minimum wage studies); hierarchical data structures with outcome information at both the individual and aggregate level (e.g., students within schools); or discrete or count outcomes.


\clearpage
\singlespacing
\bibliographystyle{chicago}
\bibliography{syn_ctrls}

\clearpage
\section*{Supplementary Materials for ``The Augmented Synthetic Control Method''}

\maketitle

\singlespacing
\appendix
\renewcommand\thefigure{\thesection.\arabic{figure}}
\renewcommand\thetable{\thesection.\arabic{table}}
\renewcommand\thetheorem{A.\arabic{theorem}}
\renewcommand\thecorollary{A.\arabic{corollary}}
\renewcommand\thelemma{A.\arabic{lemma}}
\renewcommand\theproposition{A.\arabic{proposition}}
\renewcommand\theequation{A.\arabic{equation}}
\setcounter{figure}{0}    
\setcounter{table}{0}
\setcounter{theorem}{0}
\setcounter{corollary}{0}
\setcounter{lemma}{0}
\setcounter{proposition}{0}
\setcounter{equation}{0}

\section{Inference}
\label{sec:inference}
There is a large and growing literature on inference for the synthetic control method and variants, going beyond the original proposal in \citet{Abadie2003} and \citet{AbadieAlbertoDiamond2010}; see, for example, \citet{Li2017}, \citet{toulis2018testing},  \citet{cattaneo2019prediction}, \citet{toulis2018testing}, and \citet{chernozhukov2018inference}. Here, we consider the conformal inference approach of \citet{chernozhukov2017exact}, which is tailored to this setting, as well as an adaptation of the jackknife+ approach of \citet{Barber2019}.

We now briefly describe the conformal inference approach of \citet{chernozhukov2017exact}:
\begin{enumerate}
    \item For a given sharp null hypothesis, $H_0: \tau = \tau_0$:
    \begin{enumerate}
        \item Enforce the null hypothesis by creating an adjusted post-treatment outcome for the treated unit $\tilde{Y}_{1T} = Y_{1T} - \tau_0$.
        \item Augment the original data set to include the post-treatment time period $T$, with the adjusted outcome $\tilde{Y}_{1T}$; use the estimator \eqref{eq:greg} to obtain adjusted weights $\hat{\bm{\gamma}}(\tau_0)$. 
        \item Compute a $p$-value by assessing whether the adjusted residual $Y_{1T} - \tau_0 - \sum_{W_i=0}\hat{\gamma}_i(\tau_0)Y_{iT}$ ``conforms'' with the pre-treatment residuals:\footnote{There are several choices, such as the test statistic and the form of permutation across time periods, that reduce to Equation \eqref{eq:conformal_p} with a single post treatment time period. See \citet{chernozhukov2017exact} for further details.}
        \begin{equation}
          \label{eq:conformal_p}
          \hat{p}(\tau_0) = \frac{1}{T}\sum_{t=1}^{T_0} \bbone\left\{ \left|Y_{1T} - \tau_0 - \sum_{W_i=0}\hat{\gamma}_i(\tau_0)Y_{iT}\right| \leq \left|Y_{1t} - \sum_{W_i=0}\hat{\gamma}_i(\tau_0)Y_{it}\right|\right\} + \frac{1}{T}.
        \end{equation}
    \end{enumerate}

    \item Compute a level $\alpha$ confidence interval for $\tau$ by inverting the hypothesis test and constructing the set $$\widehat{C}_\tau^{\text{conf}}(\alpha) = \{\tau_0 \mid \hat{p}_{\tau_0} \geq \alpha\}.$$ 
    
\end{enumerate}

Since the counterfactual outcome $Y_{1T}(0)$ is random, this is equivalent to constructing a conformal \emph{prediction} set \citep{Vovk2005} for $Y_{1T}(0)$ by using the quantiles of pre-treatment residuals:
\begin{equation}
  \label{eq:conf_prediction_interval}
  \widehat{C}_{Y}^{\text{conf}} = \left\{y \in \R \; \middle \vert \; \left|y - \sum_{W_i = 0}\hat{\gamma}_i(Y_{1T} - y)Y_{iT} \right| \leq q_{T, \alpha}^+\left(\left|Y_{1t} - \sum_{W_i=0}\hat{\gamma}_i(Y_{1T} - y)Y_{it} \right|\right)\right\},
\end{equation}
where $q_{T,\alpha}^+(x_t)$ is the $\lceil(1-\alpha) T \rceil$\super{th} order statistic of $\bm{X}_1,\ldots,x_T$. 
Note that $\tau \in \widehat{C}_\tau^{\text{conf}} \Leftrightarrow Y_{1T}(0) \in \widehat{C}_Y^{\text{conf}}$.
If $\tau \in \widehat{C}_\tau^{\text{conf}}$, then the adjusted residual is less than or equal to the $\lceil(1-\alpha) T \rceil$\super{th} smallest pre-treatment residual and so $Y_{1T}(0) \in \widehat{C}_Y^{\text{conf}}$. Conversely if $\tau \not \in \widehat{C}_\tau^\text{conf}$, then the adjusted residual must be larger than the $\lceil(1-\alpha) T \rceil$\super{th} smallest pre-treatment residual and so $Y_{1T}(0) \not \in \widehat{C}_Y^{\text{conf}}$.

\citet{chernozhukov2017exact} provide several conditions for which approximate or exact finite-sample validity of the $p$-values (and hence coverage of the prediction interval $\widehat{C}_Y^{\text{conf}})$ can be achieved; our setup in Section \ref{sec:setup} follows theirs. First, they show that approximate validity under an additive noise model with either i.i.d. or stationary noise depends on the estimation error. Applying Proposition \ref{cor:ascm_error_ar} and Theorem \ref{thm:ascm_error}, we can characterize the finite-sample coverage probability of $\widehat{C}_{Y}^{\text{conf}}$ and see that the true coverage will be close to the nominal coverage level $\alpha$ if the pre-treatment fit is good and the approximation error is small. 

Second, they show exact validity if the residuals are exchangeable. Intuitively, pre-treatment residuals must be a good proxy for post-treatment residuals; under the linear factor model this can hold when the factor values do not differ much across time periods (i.e., for a two way fixed effects model), but can be violated if the factor values differ widely in the pre- and post-intervention periods. Additionally, if the ignorability assumption does not hold and treatment adoption is correlated with the shocks $\varepsilon_{it}$, then pre-treatment residuals will be poor proxies for post-treatment residuals, leading to undercoverage.

One drawback of the ``full'' conformal approach is that it is computationally intensive: to construct the prediction interval $\widehat{C}_Y^{\text{conf}}$, we must refit the weights $\hat{\gamma}_i^{\text{aug}}(Y_{1T} - y)$ for a grid of possible values of the true counterfactual outcome $y$. 
A recent alternative conformal approach is the jackknife+ \citep{Barber2019}. This procedure uses the leave one out residuals $Y_{1t} - \hat{Y}_{1t}^{(-t)}$ and the estimate of the post-treatment period $\hat{Y}_{1T}^{(-t)}$ after dropping period $t$. The prediction interval for $Y_{1T}(0)$ is
\begin{equation}
  \label{eq:jackknife+} 
  \widehat{C}_Y^{\text{jackknife+}} = \left[q_{T, \alpha / 2}^-\left(\hat{Y}_{1T}^{(-t)} - \left|Y_{1t} - \hat{Y}_{1t}^{(-t)}\right|\right), q_{T, \alpha / 2}^+\left(\hat{Y}_{1T}^{(-t)} + \left|Y_{1t} - \hat{Y}_{1t}^{(-t)}\right|\right) \right],
\end{equation}
where $q_{T,\alpha}^-(x_t)$ is the $\lfloor \alpha T \rfloor$\super{th} order statistic of $\bm{X}_1,\ldots,x_T$. As with the full conformal method, $\widehat{C}_Y^{\text{jackknife+}}$ will have exact finite sample coverage when the time periods or residuals are exhangeable. However, the jackknife+ procedure only requires that we re-fit the estimator for each of the $T_0$ pre-treatment time periods. While the full conformal method enforces a sharp null hypothesis when estimating the weights, the jackknife+ uses the leave-one-out residuals as a proxy for the distribution of of $Y_{1T}(0)$, incorporating variability in the estimator by including the leave-one-out estimate of the post-treatment outcome $\hat{Y}_{1T}^{(-t)}$. We anticipate that it is possible to extend the results of \citet{chernozhukov2017exact} to show approximate validity when residuals are not exchangeable; we leave this to future work.

\paragraph{Simulation study.} We assess the finite sample coverage of the conformal prediction intervals for $Y_{1T}(0)$ using both the full conformal method \eqref{eq:conformal_p} and the jackknife+ procedure \eqref{eq:jackknife+}. For the four simulation settings we compute 95\% prediction intervals for the first post-treatment counterfactual outcome $Y_{1T_0+1}(0)$  using both conformal methods and the both the SCM and ridge ASCM estimators. Table \ref{tab:inf_coverage} shows the results. We see that the intervals for SCM sometimes undercover, due to finite sample bias from poor treatment fit. In contrast, the intervals for ridge ASCM have greater than or close to nominal coverage for $Y_{1T_0+1}$.

\begin{table}
    \centering
  
\begin{tabular}{lllr}
\toprule
Model & Estimation method & Inference method & Coverage\\
\midrule
 &  & Full conformal & 0.934\\
\cmidrule{3-4}
 & \multirow{-2}{*}{\raggedright\arraybackslash SCM} & Jackknife+ & 0.950\\
\cmidrule{2-4}
 &  & Full conformal & 0.932\\
\cmidrule{3-4}
\multirow{-4}{*}{\raggedright\arraybackslash AR(3)} & \multirow{-2}{*}{\raggedright\arraybackslash SCM + Ridge} & Jackknife+ & 0.947\\
\cmidrule{1-4}
 &  & Full conformal & 0.926\\
\cmidrule{3-4}
 & \multirow{-2}{*}{\raggedright\arraybackslash SCM} & Jackknife+ & 0.954\\
\cmidrule{2-4}
 &  & Full conformal & 0.950\\
\cmidrule{3-4}
\multirow{-4}{*}{\raggedright\arraybackslash Factor Model} & \multirow{-2}{*}{\raggedright\arraybackslash SCM + Ridge} & Jackknife+ & 0.966\\
\cmidrule{1-4}
 &  & Full conformal & 0.930\\
\cmidrule{3-4}
 & \multirow{-2}{*}{\raggedright\arraybackslash SCM} & Jackknife+ & 0.956\\
\cmidrule{2-4}
 &  & Full conformal & 0.936\\
\cmidrule{3-4}
\multirow{-4}{*}{\raggedright\arraybackslash Factor Model (More Noise)} & \multirow{-2}{*}{\raggedright\arraybackslash SCM + Ridge} & Jackknife+ & 0.957\\
\cmidrule{1-4}
 &  & Full conformal & 0.889\\
\cmidrule{3-4}
 & \multirow{-2}{*}{\raggedright\arraybackslash SCM} & Jackknife+ & 0.957\\
\cmidrule{2-4}
 &  & Full conformal & 0.939\\
\cmidrule{3-4}
\multirow{-4}{*}{\raggedright\arraybackslash Fixed Effects} & \multirow{-2}{*}{\raggedright\arraybackslash SCM + Ridge} & Jackknife+ & 0.956\\
\bottomrule
\end{tabular}

  \caption{Coverage for the full conformal \eqref{eq:conformal_p} and jackknife+ \eqref{eq:jackknife+} prediction intervals.}
    \label{tab:inf_coverage}
\end{table}

\clearpage
\section{Additional results}

\subsection{Specialization of Ridge ASCM results to SCM}
\label{sec:additional_scm}

This appendix section specializes select results from the main text for Ridge ASCM for the special case of SCM, with $\lambda \to \infty$.

First we specialize Proposition \ref{cor:ascm_error_ar} to SCM weights by taking $\lambda \to \infty$.
\begin{corollary}
  \label{cor:synth_error_ar}
  Under the linear model \eqref{eq:ark} with independent sub-Gaussian noise with scale parameter $\sigma$, for any $\delta > 0$, for weights $\bm{\gamma} \in \Delta^{N_0}$ independent of the post-treatment outcomes $(Y_{1T},\ldots,Y_{NT})$ and for any $\delta > 0$, 
  \begin{equation}
    \label{eq:synth_err_cor_ar}
  Y_{1T}(0) - \sum_{W_i = 0}\hat{\gamma}_i Y_{iT} \leq \underbrace{\|\bm{\beta}\|_2\left\|\bm{X}_1 - \sum_{W_i=0}\hat{\gamma}_i\bm{X}_i\right\|_2}_{\text{imbalance in} \bm{X}} + \underbrace{\delta\sigma \left(1 + \|\hat{\bm{\gamma}}\|_2\right)}_{\text{post-treatment noise}},
  \end{equation}
  with probability at least $1 - 2e^{-\frac{\delta^2}{2}}$.
\end{corollary}

\noindent We can similarly specialize Theorem \ref{thm:ascm_error}.
\begin{corollary}
  \label{cor:synth_error}
  Under the linear factor model \eqref{eq:scm_factor_model} with independent sub-Gaussian noise with scale parameter $\sigma$, for weights $\bm{\gamma} \in \Delta^{N_0}$ independent of the post-treatment outcomes $(Y_{1T},\ldots,Y_{NT})$ and for any $\delta > 0$, 
  \begin{equation}
    \label{eq:synth_err_cor}
  Y_{1T}(0) - \sum_{W_i = 0}\hat{\gamma}_i Y_{iT} \leq \underbrace{\frac{JM^2}{\sqrt{T_0}}\left\|\bm{X}_1 - \sum_{W_i=0}\hat{\gamma}_i\bm{X}_i\right\|_2}_{\text{imbalance in} \bm{X}} + \underbrace{ \frac{2JM^2 \sigma}{\sqrt{T_0}}\left(\sqrt{\log 2 N_0} + \delta \right)}_{\text{approximation error}} + \underbrace{\delta\sigma \left(1 + \|\hat{\bm{\gamma}}\|_2\right)}_{\text{post-treatment noise}},
  \end{equation}
  with probability at least $1 - 6e^{-\frac{\delta^2}{2}}$.
\end{corollary}

\subsection{Error under a partially linear model with Lipshitz deviations from linearity}
\label{sec:lipshitz}

We now bound the estimation error for SCM and Ridge ASCM when the outcome is only partially linear, with Lipshitz deviations from linearity. Specifically, assume that the control potential outcome in time period $T$ satisfies
\begin{equation}
  \label{eq:approx_linear_lip}
  Y_{iT}(0) = \bm{\beta} \cdot \bm{X}_i + f(\bm{X}_i) + \varepsilon_{iT},
\end{equation}
where $f:\R^{T_0} \to \R$ is $L$-Lipshitz and the noise terms $\varepsilon_{it}$ are independent sub-Gaussian random variables with scale parameter $\sigma$, and are ignorable $\E_{\bm{\varepsilon}_T}[W_i \varepsilon_{iT}] = \E_{\bm{\varepsilon}_T}[(1-W_i) \varepsilon_{iT}] = \E_{\bm{\varepsilon}_T}[\varepsilon_{iT}] = 0$, as above.

Under this model, the $L$-Lipshitz function $f(\cdot)$ will induce an approximation error from deviating away from the nearest neighbor match.
\begin{theorem}
  \label{thm:lipshitz_approx_error_ascm}
  Let $C = \max_{W_i = 0} \|\bm{X}_i\|_2$. For any $\delta > 0$, the estimation error for the ridge ASCM weights $\hat{\bm{\gamma}}^\aug$ \eqref{eq:greg} with hyperparameter $\lambda^\ridge = N_0 \lambda$ is 

\begin{equation}
  \label{eq:lipshitz_approx_error_ascm}
  \begin{aligned}
    \left|Y_{1T}(0) - \sum_{W_i=0}\hat{\bm{\gamma}}^\aug_iY_{1T}\right| & 
    \leq \|\bm{\beta}\|_2
    \underbrace{\left\|\text{diag}\left(\frac{\lambda}{d_j^2 + \lambda}\right)(\widetilde{\bm{X}}_1 - \widetilde{\bm{X}}_{0\cdot}'\hat{\bm{\gamma}}^\scm)\right\|_2}_{\text{imbalance in $X$}} + \\[1em] 
    & \qquad\underbrace{CL\left\|\text{diag}\left(\frac{d_j}{d_j^2 + \lambda}\right)(\widetilde{\bm{X}}_1 - \widetilde{\bm{X}}_{0\cdot}'\hat{\bm{\gamma}}^\scm)\right\|_2}_{\text{excess approximation error}} + \\[1em]
    & \qquad \underbrace{L\sum_{W_i = 0}\hat{\gamma}_i^\scm \|\bm{X}_1 - \bm{X}_i\|_2}_{\text{SCM approximation error}}  \qquad +  \qquad \underbrace{\vphantom{\frac{JM^2}{\sqrt{T_0}}}\delta\sigma \left(1 + \|\hat{\bm{\gamma}}^\aug\|_2\right)}_{\text{post-treatment noise}}
  \end{aligned}
\end{equation}
with probability at least $1 - 2e^{-\frac{\delta^2}{2}}$.
\end{theorem}

We can again specialize this to the SCM weights alone by taking $\lambda \to \infty$. 

\begin{corollary}
  \label{thm:lipshitz_approx_error}
  For any $\delta > 0$, the estimation error for weights on the simplex $\hat{\bm{\gamma}} \in \Delta^{N_0}$ independent of the post-treatment outcomes $(Y_{1T},\ldots, Y_{NT})$ is 

\begin{equation}
  \label{eq:lipshitz_approx_error}
  Y_{1T}(0) - \sum_{W_i = 0}\hat{\bm{\gamma}}_iY_i \leq \|\bm{\beta}\|_2 \underbrace{\left\|\bm{X}_1 - \sum_{W_i = 0}\hat{\gamma}_i \bm{X}_i\right\|_2}_{\text{imbalance in } X} + \underbrace{L\sum_{W_i = 0}\hat{\gamma}_i \|\bm{X}_1 - \bm{X}_i\|_2}_{\text{approximation error}} + \underbrace{\delta \sigma (1 + \|\hat{\bm{\gamma}}\|_2)}_{\text{post-treatment noise}}
\end{equation}
with probability at least $1 - 2e^{-\frac{\delta^2}{2}}$.
\end{corollary}

\noindent Inspecting Corollary \ref{thm:lipshitz_approx_error}, we see that in order to control the estimation error, the weights must ensure good pre-treatment fit, while only weighting control units that are near to the treated unit, with the ratio $L / \|\bm{\beta}\|_2$ controlling the relative importance of both terms. 
\citet{Abadie_LHour} propose finding weights by solving the penalized SCM problem,
\begin{equation}
  \label{eq:abadie_lhour}
  \min_{\gamma \in \Delta^{N_0}} \left\|\bm{X}_1 - \sum_{W_i = 0}\hat{\gamma}_i \bm{X}_i\right\|_2^2 + \lambda \sum_{W_i = 0}\hat{\gamma}_i \|\bm{X}_1 - \bm{X}_i\|_2^2.
\end{equation}
Comparing this to Corollary \ref{thm:lipshitz_approx_error}, we see that under the partially linear model \eqref{eq:approx_linear_lip} where $f(\cdot)$ is $L$-Lipshitz, finding weights that limit interpolation error by controling both the overall imbalance in the lagged outcomes as well as the weighted sum of the distances is sufficient to control the error. In the above optimization problem, the hyperparameter $\lambda$ takes the role of $L/\|\bm{\beta}\|_2$.

\subsection{Error under a linear factor model with covariates}
\label{sec:cov_theory}

We can quantify the behavior of the two-step procedure from Lemma \ref{lem:aug_aux} in controlling the error under a more general form of the linear factor model \eqref{eq:scm_factor_model} with covariates \citep[see][for additional discussion]{AbadieAlbertoDiamond2010, botosaru2019role}. We can also consider the error under a linear model with auxiliary covariates, as a direct consequence of Lemma \ref{lem:aug_aux}.

For each time period $t$, the covariates in the linear factor model enter through a time-varying function $f_t:\R^K \to \R$ in the model for the outcomes at $Y_{it}$: 
 \begin{equation}
  \label{eq:scm_factor_model_cov_general}
  Y_{it}(0) = \sum_{j=1}^J \phi_{ij} \mu_{jt} + f_t(\bm{Z}_i) + \varepsilon_{it}.
  \end{equation} 
For this model we again assume that $\varepsilon_{iT}$ are independent, mean-zero sub-Gaussian random variables with scale parameter $\sigma$, and are uncorrelated with treatment assignment.

To characterize how well the covariates approximate the true function $f(\bm{Z}_i)$, we will consider the best linear approximation in our data, and define the residual for unit $i$ and time $t$ as $e_{it} = f_t(\bm{Z}_i) - \bm{Z}_i'(\bm{Z}'\bm{Z})^{-1}\bm{Z}'f_t(Z)$, where $\bm{Z} \in \R^{N \times K}$ is the matrix of all auxiliary covariates for all units. For each time period we will characterize the additional approximation error incurred by only balancing the covariates linearly with the \emph{residual sum of squares} $RSS_t = \sum_{i=1}^n e_{it}^2$.
For ease of exposition,  we assume that the control covariates are standardized and rotated, which can always be true after pre-processing, and present results for the simpler case in which we fit SCM on the residualized pre-treatment outcomes rather than ridge ASCM (i.e., we let $\lambda^\ridge \to \infty$); the more general version follows immediately by applying Theorem \ref{thm:ascm_error}.

\begin{theorem}
  \label{thm:ascm_err_covs_general}
  Under the linear factor model with covariates \eqref{eq:scm_factor_model_cov} with $\frac{1}{N_0}\bm{Z}_{0\cdot}' \bm{Z}_{0\cdot} = \bm{I}_K$ and sub-Gaussian noise, for any $\delta>0$, $\hat{\bm{\gamma}}^\cov$ in Equation \eqref{eq:ols_scm_weights} with $\lambda^\ridge \to \infty$ satisfies the bound
  \begin{equation}
  \label{eq:ascm_err_covs}
  \begin{aligned}
    \left|Y_{1T}(0) - \sum_{W_i=0}\hat{\bm{\gamma}}^\cov Y_{iT} \right| &
    \leq \frac{JM^2}{\sqrt{T_0}}\left(\vphantom{\sqrt{\frac{K}{T_0}}}\right.
    \underbrace{\left\|\check{\bm{X}}_1 - \check{\bm{X}}_{0\cdot}'\hat{\bm{\gamma}}\right\|_2}_{\text{imbalance in $\check{\bm{X}}$}} 
      + \underbrace{4\sigma\sqrt{\frac{K}{N_0}}\|\bm{Z}_1 - \bm{Z}_{0\cdot}'\hat{\bm{\gamma}}\|_2}_{\text{excess approximation error}}\left.\vphantom{\sqrt{\frac{K}{T_0}}}\right) + \\[1em]
    & \qquad \underbrace{\frac{2 J M^2 \sigma}{\sqrt{T_0}} \left(\sqrt{\log N_0} + \frac{\delta}{2}\right)}_{\text{SCM approximation error}}  + \underbrace{(JM^2 + 1)e_{1\text{max}} + (JM^2 + 1)\sqrt{RSS_{\text{max}}}\|\hat{\bm{\gamma}}^\cov\|_2}_{\text{covariate approximation error}} \\
    & \qquad + \underbrace{\delta \sigma (1 + \|\hat{\bm{\gamma}}^\cov\|_2)}_{\text{post-treatment noise}}
  \end{aligned}
  \end{equation}
  with probability at least $1-6e^{-\frac{\delta^2}{2}} - 2e^{-\frac{KN_0(2-\sqrt{\log 5})^2}{2}}$, where $e_{1\text{max}} = \max_{t}|e_{1t}|$ is the maximal residual for the treated unit and $RSS_\text{max} = \max_t RSS_t$ is the maximal residual sum of squares
\end{theorem}

We can also consider the special case of Theorem \ref{thm:ascm_err_covs_general} when $f_t(\bm{Z}_i) = \sum_{k=1}^KB_{tk} Z_{ik}$ is a linear function of the covariates, and so
\begin{equation}\label{eq:scm_factor_model_cov}
  Y_{it}(0) = \sum_{j=1}^J \phi_{ij} \mu_{jt} + \sum_{k=1}^KB_{tk}Z_{ik} + \varepsilon_{it} = \bm{\phi}_i'\bm{\mu}_T + \bm{B}_t'\bm{Z}_i + \varepsilon_{it}.
\end{equation}
In this case the residuals $e_{it} = 0 \;\;\; \forall i,t$.

\begin{corollary}
  \label{thm:ascm_err_covs}
  Under the linear factor model with covariates \eqref{eq:scm_factor_model_cov} with $\frac{1}{N_0}\bm{Z}_{0\cdot}' \bm{Z}_{0\cdot} = \bm{I}_K$ and sub-Gaussian noise, for any $\delta>0$, $\hat{\bm{\gamma}}^\cov$ in Equation \eqref{eq:ols_scm_weights} with $\lambda^\ridge \to \infty$ satisfies the bound
  \begin{equation}
  \label{eq:ascm_err_covs}
  \begin{aligned}
    \left|Y_{1T}(0) - \sum_{W_i=0}\hat{\bm{\gamma}}^\cov Y_{iT} \right| &
    \leq \frac{JM^2}{\sqrt{T_0}}\left(\vphantom{\sqrt{\frac{K}{T_0}}}\right.
    \underbrace{\left\|\check{\bm{X}}_1 - \check{\bm{X}}_{0\cdot}'\hat{\bm{\gamma}}\right\|_2}_{\text{imbalance in $\check{\bm{X}}$}} 
      + \underbrace{4\sigma\sqrt{\frac{K}{N_0}}\|\bm{Z}_1 - \bm{Z}_{0\cdot}'\hat{\bm{\gamma}}\|_2}_{\text{excess approximation error}}\left.\vphantom{\sqrt{\frac{K}{T_0}}}\right) + \\[1em]
    & \qquad \underbrace{\frac{2 J M^2 \sigma}{\sqrt{T_0}} \left(\sqrt{\log N_0} + \frac{\delta}{2}\right)}_{\text{SCM approximation error}} \qquad + \underbrace{\delta \sigma (1 + \|\hat{\bm{\gamma}}^\cov\|_2)}_{\text{post-treatment noise}}
  \end{aligned}
  \end{equation}
  with probability at least $1-6e^{-\frac{\delta^2}{2}} - 2e^{-\frac{KN_0(2-\sqrt{\log 5})^2}{2}}$.
\end{corollary}

Building on Lemma \ref{lem:aug_aux}, Theorem \ref{thm:ascm_err_covs_general} and Corollary \ref{thm:ascm_err_covs} show that due to the additive, separable structure of the auxiliary covariates in Equation \eqref{eq:scm_factor_model_cov},  controlling the pre-treatment fit in the \emph{residualized} lagged outcomes $\check{\bm{X}}$ partially controls the error.
This justifies directly targeting fit in the residualzed lagged outcomes $\check{\bm{X}}$ rather than targeting raw lagged outcomes $\bm{X}$.
Moreover, the excess approximation error will be small since
since the number of covariates $K$ is small relative to $N_0$ and the auxiliary covariates are measured without noise. 
As in Section \ref{sec:better_fit}, we can achieve better balance when we apply ridge ASCM to $\check{\bm{X}}$ than when we apply SCM alone. 
Because $\check{\bm{X}}$ are orthogonal to $\bm{Z}$ by construction, this comes at no cost in terms of imbalance in $\bm{Z}$.
However, the fundamental challenge of over-fitting to noise still remains, and, as in the case without auxiliary covariates, selecting the tuning parameter remains important.
We again propose to follow the cross validation approach in Section \ref{sec:cv}, here using the residualized lagged outcomes $\check{\bm{X}}$ rather than the raw lagged outcomes $\bm{X}$.

\clearpage
\section{Simulation data generating process}
\label{sec:sim_details}
We now describe the simulations in detail. We use the Generalized Synthetic Control Method \citep{Xu2017} to fit the following linear factor model to the observed series of log GSP per capita $(N = 50, T_0 = 89, T = 105)$, setting $J=3$:
\begin{equation}\label{eq:sim_dgp}
Y_{it} = \alpha_i + \nu_t + \sum_{j=1}^J \phi_{ij} \mu_{jt} + \varepsilon_{it}.
\end{equation}
We then use these estimates as the basis for simulating data. Appendix 
Figure \ref{fig:factors}  shows the estimated factors $\hat{\bm{\bm{\mu}}}$.
We use the estimated time fixed effects $\hat{\bm{\nu}}$ and factors $\hat{\bm{\bm{\mu}}}$ and then simulate data using Equation \eqref{eq:sim_dgp}, drawing:
\begin{align*}
    \alpha_i &\sim N(\hat{\bar{\alpha}}, \; \hat{\sigma}_\alpha)\\
    \phi &\sim N(0, \; \hat{\bm{\Sigma}}_\phi)\\
    \varepsilon_{it} &\sim N(0, \; \hat{\sigma}_\varepsilon),
\end{align*}
where $\hat{\bar{\alpha}}$ and $\hat{\sigma}_\alpha$ are the estimated mean and standard deviation of the unit-fixed effects, $\hat{\bm{\Sigma}}_\phi$ is the sample covariance of the estimated factor loadings, and $\hat{\sigma}_\varepsilon$ is the estimated residual standard deviation. We also simulate outcomes with quadruple the standard deviation, $\text{sd}(\varepsilon_{it}) = 4\hat{\sigma}_\varepsilon$. We assume a sharp null of zero treatment effect in all DGPs and estimate the ATT at the final time point.

To model selection, we compute the (marginal) propensity scores as 
\[\text{logit}^{-1}\left\{\pi_i\right\} = \text{logit}^{-1}\left\{ \mathbb{P}(T=1\mid \alpha_i, \bm{\phi}_i) \right\} = \theta\left(\alpha_i + \sum_{j} \phi_{ij}\right),\]
where we set $\theta = 1/2$ and re-scale the factors and fixed effects to have unit variance. Finally, we restrict each simulation to have a single treated unit and therefore normalize the selection probabilities as $\frac{\pi_i}{\sum_{j}\pi_j}$.

We also consider an alternative data generating process that specializes the linear factor model to only include unit- and time-fixed effects:
$$Y_{it}(0) = \alpha_i + \nu_t + \varepsilon_{it}.$$
We calibrate this data generating process by fitting the fixed effects with \texttt{gsynth} and drawing new unit-fixed effects from $\alpha_i \sim N(\hat{\bar{\alpha}}, \hat{\sigma}_\alpha)$. We then model selection proportional to the fixed effect as above with $\theta = \frac{3}{2}$. 
Second, we generate data from an AR(3) model: 
$$Y_{it}(0) = \beta_0 + \sum_{j=1}^3\beta_jY_{i(t-j)} + \varepsilon_{it},$$
where we fit $\beta_0,\bm{\beta}$ to the observed series of log GSP per capita. We model selection as proportional to the last 3 outcomes $\text{logit}^{-1}\pi_i = \theta\left(\sum_{j=1}^4Y_{i(T_0-j+1)}\right)$ and set $\theta = \frac{5}{2}$. For this simulation we estimate the ATT at time $T_0 + 1$.

\clearpage
\section{Proofs}
\label{sec:proofs}

\subsection{Proofs for Section \ref{sec:ridge_ascm}}

\begin{lemma}
  \label{lem:ridge_weights}
  With $\hat{\eta}_0^{\ridge}$ and $\hat{\bm{\eta}}^{\ridge}$, the solutions to \eqref{eq:ridge_params}, the ridge estimate can be written as a weighting estimator: 
  \begin{equation}
      \hat{Y}_{1T}^{\ridge}(0) = \hat{\eta}_0^{\ridge} + \hat{\bm{\eta}}^{\ridge \prime}\bm{X}_1 =  \sum_{W_i=0}\hat{\gamma}_i^{\ridge}Y_{iT},
  \end{equation}
  where
  \begin{equation}
      \label{eq:ridge_weights}
      \hat{\gamma}_i^{\ridge} = \frac{1}{N_0} + (\bm{X}_1 - \bar{X}_0)'(\bm{X}_{0\cdot}'\bm{X}_{0\cdot} + \lambda^\ridge \bm{I}_{T_0})^{-1} \bm{X}_i.
  \end{equation}
  Moreover, the ridge weights  $\hat{\bm{\gamma}}^{\ridge}$ are the solution to
    
    \begin{equation}
        \label{eq:ridge_primal}
        \min_{\bm{\gamma}~ \mid~ \sum_i\gamma_i=1} \;\; \frac{1}{2\lambda^\ridge}\|\bm{X}_1 - \bm{X}_{0\cdot}'\bm{\gamma}\|_2^2 + \frac{1}{2}\left\|\bm{\gamma} - \frac{1}{N_0}\right\|_2^2.
    \end{equation}
    \end{lemma}

\begin{proof}[Proof of Lemmas \ref{lem:ridge_ascm_weights} and \ref{lem:ridge_weights}]
Recall that the lagged outcomes are centered by the control averages. Notice that
\begin{equation}
    \begin{aligned}
    \hat{Y}_{1T}^{\aug}(0) & = \hat{m}(\bm{X}_1) + \sum_{W_i=0} \hat{\gamma}_i^\scm (Y_{iT} - \hat{m}(\bm{X}_i))\\
    & = \hat{\eta}_0 + \hat{\eta}'\bm{X}_1 + \sum_{W_i=0} \hat{\gamma}_i^\scm(Y_{iT} - \hat{\eta}_0 - \bm{X}_{i}'\hat{\eta})\\
    & = \sum_{W_i=0}(\hat{\gamma}_i^\scm + (\bm{X}_1 - \bm{X}_{0\cdot}'\hat{\bm{\gamma}}^\scm)(\bm{X}_{0\cdot}'\bm{X}_{0\cdot} + \lambda \bm{I}_{T_0})^{-1} \bm{X}_i) Y_{iT}\\
    & = \sum_{W_i=0} \hat{\gamma}_i^{\aug} Y_{iT}
    \end{aligned}
\end{equation}
The expression for $\hat{Y}_{1T}^{\ridge}(0)$ follows.

We now prove that $\hat{\bm{\gamma}}^\ridge$ and $\hat{\bm{\gamma}}^\scm$ solve the weighting optimization problems \eqref{eq:ridge_primal} and \eqref{eq:ridge_ascm_primal}. First, the Lagrangian dual to \eqref{eq:ridge_primal} is 

\begin{equation}
    \label{eq:ridge_dual}
    \min_{\alpha, \bm{\beta}} \frac{1}{2}\sum_{W_i=0} \left(\alpha + \bm{\beta}'\bm{X}_i + \frac{1}{N_0}\right)^2 - (\alpha + \bm{\beta}'\bm{X}_1) + \frac{\lambda}{2}\|\bm{\beta}\|_2^2,
\end{equation}
where we have used that the convex conjugate of $\frac{1}{2}\left(x - \frac{1}{N_0}\right)^2$ is $\frac{1}{2}\left(y + \frac{1}{N_0}\right)^2 - \frac{1}{2N_0^2}$.
Solving for $\alpha$ we see that 
\[\sum_{W_i=0} \hat{\alpha} + \hat{\bm{\beta}}'\bm{X}_i  + 1 = 1\]
Since the lagged outcomes are centered, this implies that 
\[\hat{\alpha} = 0\]
Now solving for $\bm{\beta}$ we see that 
\[\bm{X}_{0\cdot}'\left(\mathbf{1}\frac{1}{N_0} + \bm{X}_{0\cdot}\hat{\bm{\beta}}\right) + \lambda \hat{\bm{\beta}}= \bm{X}_1\]
This implies that 
\[\hat{\bm{\beta}} = (\bm{X}_{0\cdot}'\bm{X}_{0\cdot} + \lambda I)^{-1}\bm{X}_1\]
Finally, the weights are the ridge weights
\[\hat{\gamma}_i = \frac{1}{N_0} + \bm{X}_1'(\bm{X}_{0\cdot}'\bm{X}_{0\cdot} + \lambda I)^{-1}\bm{X}_i = \hat{\bm{\gamma}}^\ridge_i\]
Similarly, the Lagrangian dual to \eqref{eq:ridge_ascm_primal} is 
\begin{equation}
    \label{eq:ridge_dual}
    \min_{\alpha, \bm{\beta}} \frac{1}{2}\sum_{W_i=0} \left(\alpha + \bm{\beta}'\bm{X}_i + \hat{\gamma}^\scm_i\right)^2 - (\alpha + \bm{\beta}'\bm{X}_1) + \frac{\lambda}{2}\|\bm{\beta}\|_2^2,
\end{equation}
where we have used that the convex conjugate of $\frac{1}{2}\left(x - \hat{\gamma}_i^\scm\right)^2$ is $\frac{1}{2}\left(y + \hat{\gamma}_i^\scm\right)^2 - \frac{1}{2}\hat{\gamma}_i^{\scm 2}$.
 Solving for $\alpha$ we see that $\hat{\alpha} = 0$.
Now solving for $\bm{\beta}$ we see that 
\[\hat{\bm{\beta}} = (\bm{X}_{0\cdot}'\bm{X}_{0\cdot} + \lambda I)^{-1}(\bm{X}_1-\bm{X}_{0\cdot}'\hat{\bm{\gamma}}^\scm)\]
Finally, the weights are the ridge ASCM weights
\[\hat{\gamma}_i = \hat{\gamma}^\scm_i + (\bm{X}_1 - \bm{X}_{0\cdot}'\hat{\bm{\gamma}}^\scm)'(\bm{X}_{0\cdot}'\bm{X}_{0\cdot} + \lambda I)^{-1}\bm{X}_i = \hat{\gamma}^\aug_i\]
\end{proof}

\begin{proof}[Proof of Lemma \ref{lem:aug_imbal}]
    Notice that 
    \[
      \begin{aligned}
        \bm{X}_1 - \bm{X}_{0\cdot}'\hat{\bm{\gamma}}^\aug & = (I - \bm{X}_{0\cdot}'\bm{X}_{0\cdot}(\bm{X}_{0\cdot}'\bm{X}_{0\cdot} + N_0 \lambda I)^{-1})(\bm{X}_1 - \bm{X}_{0\cdot}'\hat{\bm{\gamma}}^\scm)\\ 
        & = N_0\lambda (\bm{X}_{0\cdot}'\bm{X}_{0\cdot} + N_0 \lambda I)^{-1}(\bm{X}_1 - \bm{X}_{0\cdot}'\hat{\bm{\gamma}}^\scm)\\
        & = \bm{V} \text{diag}\left(\frac{\lambda}{d_j^2 + \lambda}\right) \bm{V}'(\bm{X}_1 - \bm{X}_{0\cdot}'\hat{\bm{\gamma}}^\scm)
      \end{aligned}
      \]
    So since $\bm{V}$ is orthogonal, 
    \[ \|\bm{X}_1 - \bm{X}_{0\cdot}'\hat{\bm{\gamma}}^\aug\|_2 = \left\|\text{diag}\left(\frac{\lambda}{d_j^2 + \lambda}\right)(\widetilde{\bm{X}}_1 - \widetilde{\bm{X}}_{0\cdot}'\hat{\bm{\gamma}}^\scm)\right\|_2\]
\end{proof}

\begin{lemma}
  \label{lem:aug_variance}
  The ridge augmented SCM weights with hyperparameter $\lambda N_0$, $\hat{\bm{\gamma}}^{\aug}$, satisfy
  \begin{equation}
      \label{eq:aug_var}
      \begin{aligned}
        \left\|\hat{\bm{\gamma}}^\aug\right\|_2\leq \left\|\hat{\bm{\gamma}}^\scm\right\|_2 + \frac{1}{\sqrt{N_0}}\left\|\text{diag}\left(\frac{d_j}{d_j^2 + \lambda}\right) (\widetilde{\bm{X}}_1 - \widetilde{\bm{X}}_{0\cdot}'\hat{\bm{\gamma}}^\scm)\right\|_2,
          \end{aligned}   
  \end{equation}
  with $\widetilde{\bm{X}}_i = V'\bm{X}_i$ as defined in Lemma \ref{lem:aug_imbal}.
  \end{lemma}

\begin{proof}[Proof of Lemma \ref{lem:aug_variance}]
  Notice that using the singular value decomposition and by the triangle inequality,
  \[
    \begin{aligned}
      \|\hat{\bm{\gamma}}^\aug\|_2 & = \left\|\hat{\bm{\gamma}}^\scm + \bm{X}_{0\cdot}(\bm{X}_{0\cdot}'\bm{X}_{0\cdot} + \lambda I)^{-1}(\bm{X}_1 - \bm{X}_{0\cdot}'\hat{\bm{\gamma}}^\scm)\right\|_2\\ 
      & = \left\|\hat{\bm{\gamma}}^\scm + \bm{U} \text{diag}\left(\frac{\sqrt{N_0}d_j}{N_0d_j^2 + \lambda N_0}\right) \bm{V}'(\bm{X}_1 - \bm{X}_{0\cdot}'\hat{\bm{\gamma}}^\scm)\right\|_2\\ 
      & \leq \|\hat{\bm{\gamma}}^\scm\|_2 + \left\|\text{diag}\left(\frac{d_j}{(d_j^2 + \lambda)\sqrt{N_0}}\right) (\widetilde{\bm{X}}_1 - \widetilde{\bm{X}}_{0\cdot}'\hat{\bm{\gamma}}^\scm)\right\|_2.
    \end{aligned}
    \]
\end{proof}

\subsection{Proofs for Sections \ref{sec:bias_section}, \ref{sec:additional_scm}, and \ref{sec:lipshitz}}

First, consider a model where the post-treatment control potential outcomes at time $T$ are linear in the lagged outcomes and include a unit specific term $\xi_i$, i.e.
\begin{equation}
  \label{eq:approx_linear}
  Y_{iT}(0) = \bm{\beta} \cdot \bm{X}_i + \xi_i + \varepsilon_{iT},
\end{equation}
where $\varepsilon_{it}$ are independent sub-Gaussian random variables with scale parameter $\sigma$, and are ignorable $\E_{\bm{\varepsilon}_T}[W_i \varepsilon_{iT}] = \E_{\bm{\varepsilon}_T}[(1-W_i) \varepsilon_{iT}] = \E_{\bm{\varepsilon}_T}[\varepsilon_{iT}] = 0$. Below we will put structure on the unit-specific terms $\xi_i$, first we write a general finite-sample bound.

\begin{proposition}
  \label{prop:general_bound}
  Under model \eqref{eq:approx_linear} with independent sub-Gaussian noise, for weights $\hat{\bm{\gamma}}$ independent of the post-treatment residuals $(\varepsilon_{1T},\ldots,\varepsilon_{NT})$ and for any $\delta > 0$,
  \begin{equation}
    \label{eq:general_bound}
    Y_{1T}(0) - \sum_{W_i = 0}\hat{\gamma}_i Y_{iT} \leq \|\bm{\beta}\|_2 \underbrace{\left\|\bm{X}_1 - \sum_{W_i=0}\hat{\gamma}_i\bm{X}_i\right\|_2}_{\text{imbalance in} X} + \underbrace{\left|\xi_1 - \sum_{W_i = 0}\hat{\gamma}_i\xi_i\right|}_{\text{approximation error}} + \underbrace{\delta\sigma (1 + \|\hat{\bm{\gamma}}\|_2)}_{\text{post-treatment noise}},
  \end{equation}
  with probability at least $1 - 2e^{-\frac{\delta^2}{2}}$.
\end{proposition}
\begin{proof}
  First, note that the estimation error is 

  \begin{equation}
    \label{eq:approx_lin_err_decomp}
    Y_{1T}(0) - \sum_{W_i = 0}\hat{\gamma}_i Y_{iT} = \bm{\beta} \cdot \left(X_{1} - \sum_{W_i = 0}\hat{\gamma}_i \bm{X}_i\right) + \left(\rho_{1} - \sum_{W_i = 0}\hat{\gamma}_i \xi_i\right) + \left(\varepsilon_{1T} - \sum_{W_i = 0}\hat{\gamma}_i \varepsilon_{iT}\right)
  \end{equation}  
  
  Now since the weights are independent of $\varepsilon_{iT}$, by the ignorability condition \eqref{eq:ignore} and sub-Gaussianity and independence of $\varepsilon_{iT}$, we see that $\varepsilon_{1T} - \sum_{W_i=0}\hat{\gamma}_i\varepsilon_{iT}$ is sub-Gaussian with scale parameter $\sigma\sqrt{1 + \|\hat{\bm{\gamma}}\|_2^2} \leq \sigma\left(1+ \|\hat{\bm{\gamma}}\|_2\right)$. Therefore we can bound the second term:

  \[
    P\left(\left|\varepsilon_{1T} - \sum_{W_i=0}\hat{\gamma}_i\varepsilon_{iT} \right| \geq \delta \sigma \left(1 + \|\hat{\bm{\gamma}}\|_2\right)\right) \leq 2 \exp\left(-\frac{\delta^2}{2}\right)
  \]
  The result follows from the triangle inequality and the Cauchy-Schwartz inequality.
\end{proof}

\begin{proof}[Proof of Proposition \ref{cor:ascm_error_ar}]
    Note that under the linear model \eqref{eq:ark}, $\xi_i = 0$ for all $i$. Now from Lemma \ref{lem:aug_imbal} we have that 
    \[
    \|\bm{X}_1 - \bm{X}_{0\cdot}'\hat{\bm{\gamma}}^\aug\|_2 = \left\|\text{diag}\left(\frac{\lambda}{d_j^2 + \lambda}\right)(\widetilde{\bm{X}}_1 - \widetilde{\bm{X}}_{0\cdot}'\hat{\bm{\gamma}}^\scm)\right\|_2.
    \]
    Plugging this in to Equation \eqref{eq:general_bound} completes the proof.
\end{proof}

\begin{proof}[Proof of Corollary \ref{cor:synth_error_ar}]
  This is a direct consequence of Proposition \ref{prop:general_bound} noting that under the linear model \eqref{eq:ark}, $\xi_i = 0$ for all $i$.
\end{proof}

\paragraph{Random approximation error}
We now consider the case where $\xi_i$ are random. We can use Proposition \ref{prop:general_bound} to further bound the approximation error. In particular, we'll assume that $\xi_i$ are sub-Gaussian random variables with scale parameter $\varpi$ and $\E_\xi[W_i \xi_i] = \E_\xi[(1 - W_i) \xi_i] = \E_\xi[\xi_i] = 0$.

\begin{lemma}
  \label{lem:random_approx}
  If $\xi_i$ are mean-zero sub-Gaussian random variables with scale parameter $\varpi$, then for weights $\hat{\bm{\gamma}}$ and any $\delta > 0$ the approximation error satisfies
  \begin{equation}
    \label{eq:random_approx}
    \left|\xi_1 - \sum_{W_i = 0}\hat{\gamma}_i \xi_i\right| \leq \delta \varpi + 2 \|\hat{\bm{\gamma}}\|_1 \varpi \left(\sqrt{\log 2 N_0} + \frac{\delta}{2}\right),
  \end{equation}
  with probability at least $1 - 4e^{-\frac{\delta^2}{2}}$.
\end{lemma}

\begin{proof}[Proof of Lemma \ref{lem:random_approx}]
  From the triangle inequality and H\"older's inequality we see that 
  \[
    \left|\xi_1 - \sum_{W_i = 0}\hat{\gamma}_i  \xi_i\right|  \leq |\xi_1| +  \|\hat{\bm{\gamma}}\|_1 \max_{W_i = 0} |\xi_i|.
  \]

  Now since the $\xi_i$ are mean-zero sub-Gaussian with scale parameter $\varpi$, we have that
  \[
    P\left(|\xi_1| \geq \delta \varpi \right) \leq 2e^{-\frac{\delta^2}{2}}
  \]
  Next, from the union bound, the maximum of the $N_0$ sub-Gaussian variables $\rho_2,\ldots,\rho_N$ satisfies 

  \[
    P\left(\max_{W_i = 0} |\xi_i| \geq 2 \varpi \sqrt{\log 2 N_0} + \delta\right) \leq 2 e^{-\frac{\delta^2}{2\varpi^2}}.
  \]

  Setting $\delta = \delta\varpi$ and combining the two probabilities with the union bound gives the result.

\end{proof}

\begin{lemma}
  \label{lem:random_approx_ridge_ascm}
  If $\xi_i$ are mean-zero sub-Gaussian random variables with scale parameter $\varpi$, for the ridge ASCM weights $\hat{\bm{\gamma}}^\aug$ with hyper-parameter $\lambda^\ridge = \lambda N_0$ and for any $\delta > 0$ the approximation error satisfies
  \begin{equation}
    \label{eq:random_approx_ridge_ascm}
    \left|\xi_1 - \sum_{W_i = 0}\hat{\gamma}_i \xi_i\right| \leq 2 \varpi \left(\sqrt{\log 2 N_0} + \frac{\delta}{2}\right) + \underbrace{(1 + \delta) 4 \varpi \left\| \text{diag}\left(\frac{d_j}{d_j^2 + \lambda}\right)\left(\widetilde{\bm{X}}_1 - \widetilde{\bm{X}}_{0\cdot}'\hat{\bm{\gamma}}^\scm\right)\right\|_2}_{\text{excess approximation error}},
  \end{equation}
  with probability at least $1 - 4e^{-\frac{\delta^2}{2}} - e^{-2(\log 2 + N_0 \log 5)\delta^2}$.
\end{lemma}

\begin{proof}[Proof of Lemma \ref{lem:random_approx_ridge_ascm}]
  Again from H\"older's inequality we see that 
  \[
    \begin{aligned}
      \left|\xi_1 - \sum_{W_i = 0}\hat{\gamma}_i^\aug  \xi_i\right| & = |\xi_1| + \left |\sum_{W_i=0}(\hat{\gamma}_i^\scm + \hat{\gamma}_i^\aug - \hat{\gamma}_i^\scm) \xi_i\right|\\
      & \leq |\xi_1| + \|\hat{\bm{\gamma}}^\scm\|_1 \max_{W_i = 0} |\xi_i| + \|\hat{\bm{\gamma}}^\aug - \hat{\bm{\gamma}}^\scm\|_2 \sqrt{\sum_{W_i=0}\xi_i ^ 2}.  
    \end{aligned}
  \]

  We have bounded the first two terms in Lemma \ref{lem:random_approx}, now it sufficies to bound the third term. First, from Lemma \ref{lem:aug_variance} we see that 
  \[
    \|\hat{\bm{\gamma}}^\aug - \hat{\bm{\gamma}}^\scm\|_2 = \frac{1}{\sqrt{N_0}}\left\| \text{diag}\left(\frac{d_j}{d_j^2 + \lambda}\right)\left(\widetilde{\bm{X}}_1 - \widetilde{\bm{X}}_{0\cdot}'\hat{\bm{\gamma}}^\scm\right)\right\|_2.
  \]
  Second, via a standard discretization argument \citep{wainright2018high}, we can bound the $L^2$ norm of the vector $(\xi_2,\ldots, \xi_N)$ as
  \[
    P\left(\sqrt{\sum_{W_i=0}\xi_i^ 2} \geq 2\varpi \sqrt{\log 2 + N_0 \log 5} + \delta \right) \leq 2\exp\left(-\frac{\delta^2}{2\varpi^2}\right).
  \]
    Setting $\delta = 2\delta\varpi\sqrt{\log 2 + N_0\log 5}$, noting that $\log 2 + N_0 \log 5 < 4 N_0$, we have that 
  \[
    \|\hat{\bm{\gamma}}^\aug - \hat{\bm{\gamma}}^\scm\|_2 \sqrt{\sum_{W_i=0}\xi_i ^ 2} \leq (1 + \delta) \varpi 4\left\| \text{diag}\left(\frac{d_j}{d_j^2 + \lambda}\right)\left(\widetilde{\bm{X}}_1 - \widetilde{\bm{X}}_{0\cdot}'\hat{\bm{\gamma}}^\scm\right)\right\|_2
  \] 
    with probability at least $1 - 2e^{-2(\log 2 + N_0 \log 5)\delta^2}$.
    Since $\|\hat{\bm{\gamma}}^\scm\|_1 = 1$, combining with Lemma \ref{lem:random_approx} via the union bound gives the result.

\end{proof}

\begin{theorem}
  \label{thm:random_approx_error}
  Under model \eqref{eq:approx_linear} with independent sub-Gaussian noise $\varepsilon_{iT}$, and $\xi_i$ mean-zero sub-Gaussian with scale parameter $\varpi$, for weights $\hat{\bm{\gamma}}$ independent of the post-treatment outcomes $(Y_{1T},\ldots,Y_{NT})$ and for any $\delta > 0$,
  \begin{equation}
    \label{eq:random_approx_error}
    Y_{1T}(0) - \sum_{W_i = 0}\hat{\gamma}_i Y_{iT} \leq \|\bm{\beta}\|_2 \underbrace{\left\|\bm{X}_1 - \sum_{W_i=0}\hat{\gamma}_i\bm{X}_i\right\|_2}_{\text{imbalance in } X} + \underbrace{\delta \varpi + 2\|\hat{\bm{\gamma}}\|_1 \varpi \left(\sqrt{\log 2 N_0} + \frac{\delta}{2} \right)}_{\text{approximation error}} + \underbrace{\delta\sigma \left(1 + \|\hat{\bm{\gamma}}\|_2\right)}_{\text{post-treatment noise}},
  \end{equation}
  with probability at least $1 - 6e^{-\frac{\delta^2}{2}}$.
\end{theorem}

\begin{proof}[Proof of Theorem \ref{thm:random_approx_error}]
  The Theorem directly follows from Proposition \ref{prop:general_bound} and Lemma \ref{lem:random_approx}, combining the two probabilistic bounds via the union bound.
\end{proof}

\begin{theorem}
  \label{thm:random_approx_error_ascm}
  Under model \eqref{eq:approx_linear} with independent sub-Gaussian noise $\varepsilon_{iT}$, and $\xi_i$ mean-zero sub-Gaussian with scale parameter $\varpi$, for any $\delta > 0$, the ridge ASCM weights with hyperparameter  $\lambda^\ridge = \lambda N_0$ satisfy the bound
  \begin{equation}
    \label{eq:random_approx_error}
    \begin{aligned}
      Y_{1T}(0) - \sum_{W_i = 0}\hat{\gamma}_i Y_{iT} & \leq \|\bm{\beta}\|_2 \underbrace{\left\|\text{diag}\left(\frac{\lambda}{d_j^2 + \lambda}\right) \left(\widetilde{\bm{X}}_1 - \sum_{W_i=0}\hat{\gamma}_i^\scm\widetilde{\bm{X}}_i\right)\right\|_2}_{\text{imbalance in } X} + \underbrace{2\varpi \left(\sqrt{\log 2 N_0} + \frac{\delta}{2} \right)}_{\text{approximation error}}\\
      & \qquad \underbrace{(1 + \delta)  4 \varpi \left\| \text{diag}\left(\frac{d_j}{d_j^2 + \lambda}\right)\left(\widetilde{\bm{X}}_1 - \widetilde{\bm{X}}_{0\cdot}'\hat{\bm{\gamma}}^\scm\right)\right\|_2}_{\text{excess approximation error}} + \underbrace{\delta\sigma \left(1 + \|\hat{\bm{\gamma}}\|_2\right)}_{\text{post-treatment noise}},
    \end{aligned}
  \end{equation}
  with probability at least $1 - 6e^{-\frac{\delta^2}{2}} - e^{-2(\log 2 + N_0 \log 5)\delta^2}$.
\end{theorem}

\begin{proof}[Proof of Theorem \ref{thm:random_approx_error_ascm}]
  First note that from Lemma \ref{lem:aug_imbal} we have that
  \[
    \|\bm{X}_1 - \bm{X}_{0\cdot}'\hat{\bm{\gamma}}^\aug\|_2 = \left\|\text{diag}\left(\frac{\lambda}{d_j^2 + \lambda}\right)(\widetilde{\bm{X}}_1 - \widetilde{\bm{X}}_{0\cdot}'\hat{\bm{\gamma}}^\scm)\right\|_2. 
  \]
  The Theorem directly follows from Proposition \ref{prop:general_bound} and Lemma \ref{lem:random_approx_ridge_ascm}, combining the two probabilistic bounds via the union bound.
\end{proof}

Theorems \ref{thm:random_approx_error} and \ref{thm:random_approx_error_ascm} have several implications when the outcomes follow a linear factor model \eqref{eq:scm_factor_model}. To see this, we need one more lemma:

\begin{lemma}
  \label{lem:factor_model_approx}
  The linear factor model is a special case of  model \eqref{eq:approx_linear} with $\bm{\beta} = \frac{1}{T_0}\bm{\mu}\bm{\mu}_T$ and $\xi_i = \frac{1}{T_0}\bm{\mu}_T'\bm{\mu}\bm{\varepsilon}_{i(1:T_0)}$. $\|\bm{\beta}\|_2 \leq \frac{MJ^2}{\sqrt{T_0}}$, and if $\bm{\varepsilon}_{i(1:T_0)}$ are independent sub-Gaussian vectors with scale parameter $\sigma_{T_0}$, then $\frac{1}{T_0} \bm{\bm{\mu}}_T'\bm{\bm{\mu}}'\bm{\varepsilon}_{i(1:T_0)}$ is sub-Gaussian with scale parameter $\frac{JM^2\sigma_{T_0}}{\sqrt{T_0}}$.
\end{lemma}

\begin{proof}[Proof of Lemma \ref{lem:factor_model_approx}]
  Notice that under the linear factor model, the pre-treatment covariates for unit $i$ satisfy:
  \[
    \bm{X}_i = \bm{\bm{\mu}} \bm{\phi}_i + \bm{\varepsilon}_{i(1:T_0)}.
  \]
  Multiplying both sides by $(\bm{\mu}'\bm{\mu})^{-1}\bm{\mu}' = \frac{1}{T_0} \bm{\mu}'$ and rearranging gives
  \[
    \frac{1}{T_0} \bm{\mu}' \bm{X}_i - \frac{1}{T_0} \bm{\mu}' \bm{\varepsilon}_{i(1:T_0)} = \bm{\phi}_i.
  \]
  Then we see that the post treatment outcomes are
  \[
    Y_{iT}(0) = \frac{1}{T_0}\bm{\mu}_T'\bm{\mu}' \bm{X}_i + \frac{1}{T_0}\bm{\mu}_T'\bm{\mu}'\bm{\varepsilon}_{i(1:T_0)}.
  \]

  Since $\bm{\varepsilon}_{i(1:T_0)}$ is a sub-Gaussian vector $v'\bm{\varepsilon}_{i(1:T_0)}$ is sub-Gaussian with scale parameter $\sigma_{T_0}$ for any $v \in \R^{T_0}$ such that $\|v\|_2 = 1$. Now notice that $\|\bm{\mu}_T'\bm{\mu}'\|_2 \leq \|\bm{\mu}_T\|_2 \|\bm{\mu}\|_2 \leq MJ^2\sqrt{T_0}$. This completes the proof.
\end{proof}

\begin{proof}[Proof of Corollary \ref{cor:synth_error}]
  From Lemma \ref{lem:factor_model_approx} we can apply Theorem \ref{thm:random_approx_error} with $\bm{\beta} = \frac{1}{T_0}\bm{\mu}_T'\bm{\mu}'$ and $\xi_i = \frac{1}{T_0}\bm{\mu}_T'\bm{\mu}'\bm{\varepsilon}_{i(1:T_0)}$. Since $\varepsilon_{it}$ are independent sub-Gaussian random variables, $\bm{\varepsilon}_{i(1:T_0)}$ is a sub-Gaussian vector with scale parameter $\sigma_{T_0} = \sigma$. Noting that $\|\hat{\bm{\gamma}}\|_1 = \sum_{W_i = 0} |\hat{\gamma}_i| = 1$ and applying Lemma \ref{lem:factor_model_approx} gives the result.
\end{proof}

\begin{proof}[Proof of Theorem \ref{thm:ascm_error}]
  Again from Lemma \ref{lem:factor_model_approx} we can apply Theorem \ref{thm:random_approx_error_ascm} with $\bm{\beta} = \frac{1}{T_0}\bm{\mu}_T'\bm{\mu}'$ and $\xi_i = \frac{1}{T_0}\bm{\mu}_T'\bm{\mu}'\bm{\varepsilon}_{i(1:T_0)}$, so $\varpi = \frac{JM^2}{\sqrt{T_0}}$.
  Plugging these values into Theorem \ref{thm:random_approx_error} gives the result.
\end{proof}

\begin{corollary}[Approximation error for ridge ASCM with dependent errors]
  \label{cor:dependent_approx_error}
  Under the linear factor model \eqref{eq:scm_factor_model} with time-dependent errors satisfying $\bm{\varepsilon}_{i(1:T_0)} \overset{iid}{\sim} N(0, \sigma^2 \Omega)$
  the approximation error satisfies
  \begin{equation}
    \label{eq:factor_model_dependent_error_approx}
    \begin{aligned}
      \left|\xi_1 - \sum_{W_i = 0}\hat{\gamma}_i \xi_i\right| & = \left|\frac{1}{T_0}\bm{\mu}_T'\bm{\mu}'\left(\bm{\varepsilon}_{1(1:T_0)} - \sum_{W_i = 0}\hat{\gamma}_i \bm{\varepsilon}_{i(1:T_0)}\right)\right|\\
       & \leq 2\sqrt{\frac{\|\Omega\|_2}{T_0}}JM^2\sigma \left(\sqrt{\log 2 N_0} + \delta + (1 + \delta) 2 \left\| \text{diag}\left(\frac{d_j}{d_j^2 + \lambda}\right)\left(\widetilde{\bm{X}}_1 - \widetilde{\bm{X}}_{0\cdot}'\hat{\bm{\gamma}}^\scm\right)\right\|_2\right),  
    \end{aligned}
  \end{equation}
\end{corollary}

\begin{proof}[Proof of Corollary \ref{cor:dependent_approx_error}]
  From Lemma \ref{lem:factor_model_approx}, we see that $\xi_i = \frac{1}{T_0}\bm{\mu}_T'\bm{\mu}'\bm{\varepsilon}_{i(1:T_0)}$ is sub-Guassian with scale parameter $JM^2\sqrt{\frac{\|\Omega\|_2}{T_0}}$. Plugging in to Lemma \ref{lem:random_approx_ridge_ascm} gives the result.
\end{proof}

\paragraph{Lipshitz approximation error}
If $\xi_i$ are Lipshitz functions, we can can also bound the approximation error. Specifically, we assume that $\xi_i = f(\bm{X}_i)$ where $f:\R^{T_0} \to \R$ is an $L$-Lipshitz function.

\begin{lemma}
  \label{lem:lipshitz_approx}
  If $\xi_i = f(\bm{X}_i)$ where $f:\R^{T_0} \to \R$ is an $L$-Lipshitz function, then for weights on the simplex $\hat{\bm{\gamma}} \in \Delta^{N_0}$, the approximation error satisfies
  \begin{equation}
    \label{eq:lipshitz_approx}
    \left|\xi_1 - \sum_{W_i = 0} \hat{\gamma}_i \xi_i \right| \leq L \sum_{W_i=0}\hat{\gamma}_i \|\bm{X}_1 - \bm{X}_i\|_2
  \end{equation}
\end{lemma}

\begin{proof}[Proof of Lemma \ref{lem:lipshitz_approx}]
  Since the weights sum to one, we have that 
  \[
    \left|\xi_1 - \sum_{W_i = 0} \hat{\gamma}_i \xi_i \right| \leq \left|\sum_{W_i = 0} \hat{\gamma}_i (f(\bm{X}_1) - f(\bm{X}_i))\right|.
  \]
  Now from the Lipshitz property, $|f(\bm{X}_1) - f(\bm{X}_i)| \leq L\|\bm{X}_1 - \bm{X}_i\|_2$, and so by Jensen's inequalty:
  \[
    \left|\sum_{W_i = 0} \hat{\gamma}_i (f(\bm{X}_1) - f(\bm{X}_i))\right| \leq L \sum_{W_i = 0}\hat{\gamma}_i \|\bm{X}_1 - \bm{X}_i\|_2
  \]
\end{proof}

\begin{proof}[Proof of Theorem \ref{thm:lipshitz_approx_error}]
  The proof follows directly from Proposition \ref{prop:general_bound} and Lemma \ref{lem:lipshitz_approx}.
\end{proof}

\begin{lemma}
  \label{lem:lipshitz_approx_ridge_ascm}
  Let $C = \max_{W_i = 0} \|\bm{X}_i\|_2$. If $\xi_i = f(\bm{X}_i)$ where $f:\R^{T_0} \to \R$ is an $L$-Lipshitz function, then the ridge ASCM weights $\hat{\bm{\gamma}}^\aug$ \eqref{eq:greg} with hyperparameter $\lambda^\ridge = N_0 \lambda$ satisfy
  \begin{equation}
    \label{eq:lipshitz_approx_ridge_ascm}
    \left|\xi_1 - \sum_{W_i = 0} \hat{\gamma}_i^\aug \xi_i \right| \leq L \sum_{W_i=0}\hat{\gamma}_i^\scm \|\bm{X}_1 - \bm{X}_i\|_2 + CL\left\|\text{diag}\left(\frac{d_j}{d_j^2 + \lambda}\right)\left(\widetilde{\bm{X}}_1 - \widetilde{\bm{X}}_{0\cdot}'\hat{\bm{\gamma}}^\scm\right)\right\|_2
  \end{equation}
\end{lemma}

\begin{proof}[Proof of Lemma \ref{lem:lipshitz_approx_ridge_ascm}]
  From the triangle inequality we have that
  \[
    \left|\xi_1 - \sum_{W_i = 0} \hat{\gamma}_i^\aug \xi_i \right| \leq \left|\sum_{W_i = 0} \hat{\gamma}_i^\scm (f(\bm{X}_1) - f(\bm{X}_i))\right| + \left|\sum_{W_i = 0} \bm{X}_i \left(\bm{X}_{0\cdot}'\bm{X}_{0\cdot} + \lambda I\right)^{-1} (\bm{X}_1 - \bm{X}_{0\cdot}'\hat{\bm{\gamma}}^\scm)f(\bm{X}_i)\right|.
  \]
  We have already bounded the first term in Lemma \ref{lem:lipshitz_approx}, now we bound the second term. From the Cauchy-Schwartz inequality and since $\|x\|_2 \leq \sqrt{N_0}\|x\|_\infty $ for all $x \in \R^{N_0}$ we see that 
  \[
    \begin{aligned}
      \left|\sum_{W_i = 0} \bm{X}_i \left(\bm{X}_{0\cdot}'\bm{X}_{0\cdot} + \lambda I\right)^{-1} (\bm{X}_1 - \bm{X}_{0\cdot}'\hat{\bm{\gamma}}^\scm)f(\bm{X}_i)\right| & \leq \sqrt{N_0}\left\|\bm{X}_{0\cdot} \left(\bm{X}_{0\cdot}'\bm{X}_{0\cdot} + \lambda I\right)^{-1} (\bm{X}_1 - \bm{X}_{0\cdot}'\hat{\bm{\gamma}}^\scm)\right\|_2 |f(\bm{X}_i)|\\
      & = \left\|\text{diag}\left(\frac{d_j}{d_j^2 + \lambda}\right)\left(\widetilde{\bm{X}}_1 - \widetilde{\bm{X}}_{0\cdot}'\hat{\bm{\gamma}}^\scm\right)\right\|_2 |f(\bm{X}_i)|\\
      & \leq CL\left\|\text{diag}\left(\frac{d_j}{d_j^2 + \lambda}\right)\left(\widetilde{\bm{X}}_1 - \widetilde{\bm{X}}_{0\cdot}'\hat{\bm{\gamma}}^\scm\right)\right\|_2,
    \end{aligned}
  \]
  where the second line comes from Lemma \ref{lem:aug_variance} and the third line from the Lipshitz property.  
\end{proof}

\begin{proof}[Proof of Theorem \ref{thm:lipshitz_approx_error_ascm}]
  The proof follows directly from Proposition \ref{prop:general_bound} and Lemma \ref{lem:lipshitz_approx_ridge_ascm}.
\end{proof}

\subsection{Proofs for Sections \ref{sec:extensions} and \ref{sec:cov_theory}}

\begin{proof}[Proof of Lemma \ref{lem:aug_aux}]
  The regression parameters $\hat{\bm{\eta}}_x$ and $\hat{\bm{\eta}}_z$ in Equation \eqref{eq:cov_regression} are:
  \begin{equation}
      \hat{\bm{\eta}}_x = (\check{\bm{X}}_{0\cdot}'\check{\bm{X}}_{0\cdot} + \lambda^\ridge I)^{-1}\check{\bm{X}}_{0\cdot}' Y_{0T} \;\; \text{ and } \;\; \hat{\bm{\eta}}_z = (\bm{Z}_{0\cdot}'\bm{Z}_{0\cdot})^{-1} \bm{Z}_{0\cdot}' Y_{0T}
  \end{equation}
  Now notice that 
  \begin{equation}
      \begin{aligned}
          \hat{Y}_{0T}^\cov & = \hat{\bm{\eta}}_x' \bm{X}_1 + \hat{\bm{\eta}}_z'\bm{Z}_1 + \sum_{W_i=0}(Y_{iT} - \hat{\bm{\eta}}_x'\bm{X}_i - \hat{\bm{\eta}}_z \bm{Z}_i)\hat{\gamma}_i\\
          & = \hat{\bm{\eta}}_x'(\bm{X}_1 - \bm{X}_{0\cdot}'\hat{\bm{\gamma}}) + \hat{\bm{\eta}}_z(\bm{Z}_1 - \bm{Z}_{0\cdot}'\hat{\bm{\gamma}}) + \bm{Y}_{0T}'\hat{\bm{\gamma}}\\
          & = \hat{\bm{\eta}}_x'(\bm{X}_1 - \bm{X}_{0\cdot}'\hat{\bm{\gamma}}) - \hat{\bm{\eta}}_x'\bm{X}_{0\cdot}(\bm{Z}_{0\cdot}'\bm{Z}_{0\cdot})^{-1} (\bm{Z}_1 - \bm{Z}_{0\cdot}'\hat{\bm{\gamma}}) +  Y_{0T}'\bm{Z}_{0\cdot}(\bm{Z}_{0\cdot}'\bm{Z}_{0\cdot})^{-1} (\bm{Z}_1 - \bm{Z}_{0\cdot}'\hat{\bm{\gamma}}) + Y_{0T}'\hat{\bm{\gamma}}\\
          & = \hat{\bm{\eta}}_x'(\check{\bm{X}}_1 - \check{\bm{X}}_{0\cdot}'\hat{\bm{\gamma}}) +  Y_{0T}' \bm{Z}_{0\cdot}(\bm{Z}_{0\cdot}'\bm{Z}_{0\cdot})^{-1} (\bm{Z}_1 - \bm{Z}_{0\cdot}'\hat{\bm{\gamma}}) + Y_{0T}'\hat{\bm{\gamma}}\\
          & = Y_{0T}' \left(\hat{\bm{\gamma}} + \check{\bm{X}}_{0\cdot}(\check{\bm{X}}_{0\cdot}'\check{\bm{X}}_{0\cdot} + \lambda^\ridge \bm{I}_{T_0})^{-1}(\check{\bm{X}}_1 - \check{\bm{X}}_{0\cdot}'\hat{\bm{\gamma}}) +  \bm{Z}_{0\cdot}(\bm{Z}_{0\cdot}'\bm{Z}_{0\cdot})^{-1} (\bm{Z}_1 - \bm{Z}_{0\cdot}'\hat{\bm{\gamma}})\right).
      \end{aligned}
  \end{equation}
  This gives the form of $\hat{\bm{\gamma}}^\cov$.
  The imbalance in $Z$ is
  \begin{equation}
      \begin{aligned}
      \bm{Z}_1 - \bm{Z}_{0\cdot}'\hat{\bm{\gamma}}^{\cov} & = \left(\bm{Z}_1 - \bm{Z}_{0\cdot}'  \bm{Z}_{0\cdot}(\bm{Z}_{0\cdot}'\bm{Z}_{0\cdot})^{-1} \bm{Z}_1\right) + \left(\bm{Z}_{0\cdot} - \bm{Z}_{0\cdot}'  \bm{Z}_{0\cdot}(\bm{Z}_{0\cdot}'\bm{Z}_{0\cdot})^{-1} \bm{Z}_{0\cdot}\right)'\hat{\bm{\gamma}} \\
      & \qquad- \bm{Z}_{0\cdot}'\check{\bm{X}}_{0\cdot}(\check{\bm{X}}_{0\cdot}'\check{\bm{X}}_{0\cdot}+ \lambda^\ridge I)^{-1}(\check{\bm{X}}_1 - \check{\bm{X}}_{0\cdot}'\hat{\bm{\gamma}})\\
      & = 0.
      \end{aligned}
  \end{equation}
  The pre-treatment fit is
  \begin{equation}
      \begin{aligned}
      \bm{X}_1 - \bm{X}_{0\cdot}'\hat{\bm{\gamma}}^{\cov} & = \left(\bm{X}_1 - \bm{X}_{0\cdot}'  \bm{Z}_{0\cdot}(\bm{Z}_{0\cdot}'\bm{Z}_{0\cdot})^{-1} \bm{Z}_1\right) + \left( \bm{X}_{0\cdot} - \bm{X}_{0\cdot}'  \bm{Z}_{0\cdot}(\bm{Z}_{0\cdot}'\bm{Z}_{0\cdot})^{-1} \bm{Z}_{0\cdot}\right)'\hat{\bm{\gamma}} \\
      & \qquad- \bm{X}_{0\cdot}'\check{\bm{X}}_{0\cdot}(\check{\bm{X}}_{0\cdot}'\check{\bm{X}}_{0\cdot}+ \lambda^\ridge \bm{I}_{T_0})^{-1}(\check{\bm{X}}_1 - \check{\bm{X}}_{0\cdot}'\hat{\bm{\gamma}})\\
      & = \left(\bm{I}_{T_0} - \bm{X}_{0\cdot}'\check{\bm{X}}_{0\cdot}(\check{\bm{X}}_{0\cdot}'\check{\bm{X}}_{0\cdot} + \lambda^\ridge \bm{I}_{T_0})^{-1}\right)\left(\check{\bm{X}}_1 - \check{\bm{X}}_{0\cdot}'\hat{\bm{\gamma}}\right)\\
      &= \left(\bm{I}_{T_0} - \check{\bm{X}}_{0\cdot}'\check{\bm{X}}_{0\cdot}(\check{\bm{X}}_{0\cdot}'\check{\bm{X}}_{0\cdot} + \lambda^\ridge \bm{I}_{T_0})^{-1}\right)\left(\check{\bm{X}}_1 - \check{\bm{X}}_{0\cdot}'\hat{\bm{\gamma}}\right).
      \end{aligned}
  \end{equation}
  This gives the bound on the pre-treatment fit.
  
  \end{proof}

  \begin{proof}[Proof of Theorem \ref{thm:ascm_err_covs_general}]
      
    First, we will separate $f(\bm{Z})$ into the projection onto $\bm{Z}$ and a residual. Defining $\bm{B}_t = (\bm{Z}'\bm{Z})^{-1}\bm{Z}'f_t(\bm{Z}) \in \R^K$ as the regression coefficient, the projection of $f_t(\bm{Z}_i)$ is $\bm{Z}_i'\bm{B}_t$ and the residual is $e_{it} = f_t(\bm{Z}_i) - \bm{Z}_i'\bm{B}_t$. We will denote the matrix of regression coefficients over $t=1,\ldots,T_0$ as $\bm{B} = [\bm{B}_1,\ldots,\bm{B}_{T_0}] \in \R^{K\times T_0}$ and denote the matrix of residuals as $\bm{E} \in \R^{n \times T_0}$, with $\bm{E}_{1\cdot} = (e_{11},\ldots,e_{1T_0})$ as the vector of residuals for the treated unit and $\bm{E}_{0 \cdot}$ as the matrix of residuals for the control units.
  
    Then the error is
    \[
      \begin{aligned}
        \left|Y_{1T}(0) - \sum_{W_i = 0} \hat{\bm{\gamma}}^\cov_i Y_{iT}\right| & \leq \left|\bm{\mu}_T \cdot \left(\bm{\phi}_1 - \sum_{W_i=0}\hat{\bm{\gamma}}^\cov_i\bm{\phi}_i\right)\right| + \left| \bm{B}_t \cdot \left(\bm{Z}_1 - \sum_{W_i = 0} \hat{\bm{\gamma}}^\cov_i \bm{Z}_i\right)\right|\\
        & + \left| e_{1T} - \sum_{W_i = 0} \hat{\bm{\gamma}}^\cov e_{iT}\right| + \left|\varepsilon_{1T} - \sum_{W_i = 0}\hat{\bm{\gamma}}^\cov_i \varepsilon_{iT}\right|
      \end{aligned}
    \]
      Since $\hat{\bm{\gamma}}^\cov_i$ exactly balances the covariates, the second term is equal to zero. We can bound the third term with H\"{o}lder's inequality:
      \[
        \left| e_{1T} - \sum_{W_i = 0} \hat{\bm{\gamma}}^\cov e_{iT}\right| \leq |e_{1T}| + \sqrt{RSS}_T \|\hat{\bm{\gamma}}^\cov\|_2
      \]
      In previous theorems we have bounded the last term with high probability. Only the error due to imbalance remains.
  
    Denote $\bm{\varepsilon}_{0(1:T_0)}$ as the matrix of pre-treatment noise for the control units, where the rows correspond to $\bm{\varepsilon}_{2(1:T_0)},\ldots, \bm{\varepsilon}_{N_0(1:T_0)}$. Building on Lemma \ref{lem:factor_model_approx}, we can see that the error due to imbalance in $\phi$ is equal to
      \begin{equation}
        \label{eq:factor_bias_cov}
        \begin{aligned}
          \bm{\mu}_T \cdot \left(\bm{\phi}_1 - \sum_{W_i=0}\hat{\bm{\gamma}}^\cov_i\bm{\phi}_i\right) & = \frac{1}{T_0}\bm{\mu}_T'\bm{\mu}'(\bm{X}_1 - \bm{X}_{0\cdot}'\hat{\bm{\gamma}}^\cov) - \frac{1}{T_0}\bm{\mu}_T'\bm{\mu}'(\bm{\varepsilon}_{1(1:T_0)} - \bm{\varepsilon}_{0(1:T_0)}'\hat{\bm{\gamma}}^\cov)\\
          & - \frac{1}{T_0}\bm{\mu}_T'\bm{\mu}'B'(\bm{Z}_1 - \bm{Z}_{0\cdot}'\hat{\bm{\gamma}}^\cov) - \frac{1}{T_0}\bm{\mu}_T'\bm{\mu}'(\bm{E}_{1\cdot} - \bm{E}_{0\cdot}'\hat{\bm{\gamma}}^\cov).
        \end{aligned}
    \end{equation}
    By construction, $\hat{\bm{\gamma}}^\cov$ perfectly balances the covariates, and combined with Lemma \ref{lem:aug_aux}, the error due to imbalance in $\phi$ simplifies to
    \[
     \bm{\mu}_T \cdot \left(\bm{\phi}_1 - \sum_{W_i=0}\gamma_i\bm{\phi}_i\right) = \frac{1}{T_0}\bm{\mu}_T'\bm{\mu}'(\check{\bm{X}}_1 - \check{\bm{X}}_{0\cdot}'\hat{\bm{\gamma}}) - \frac{1}{T_0}\bm{\mu}_T'\bm{\mu}'(\bm{\varepsilon}_{1(1:T_0)} - \bm{\varepsilon}_{0(1:T_0)}'\hat{\bm{\gamma}}^\cov) - \frac{1}{T_0}\bm{\mu}_T'\bm{\mu}'(\bm{E}_{1\cdot} - \bm{E}_{0\cdot}'\hat{\bm{\gamma}}^\cov).
    \]
    We now turn to bounding the noise term and the error due to the projection of $f(Z)$ on to $Z$. First, notice that 
    \[
      \frac{1}{T_0}\bm{\mu}_T'\bm{\mu}'\bm{\varepsilon}_{0(1:T_0)}'\hat{\bm{\gamma}}^\cov = \frac{1}{T_0}\bm{\mu}_T'\bm{\mu}'\bm{\varepsilon}_{0(1:T_0)}'\hat{\bm{\gamma}}^\scm + \frac{1}{T_0}\bm{\mu}_T'\bm{\mu}'\bm{\varepsilon}_{0(1:T_0)}'\bm{Z}_{0\cdot}(\bm{Z}_{0\cdot}'\bm{Z}_{0\cdot})^{-1}(\bm{Z}_1 - \bm{Z}_{0\cdot}'\hat{\bm{\gamma}}^\scm).
    \]
    We have bounded the first term on the right hand side in Lemma \ref{lem:random_approx}. To bound the second term, notice that $\sum_{W_i=0}\sum_{t=1}^{T_0} \bm{\mu}_T'\bm{\mu}_{t\cdot} Z_{ik} \varepsilon_{it}$ is sub-Gaussian with scale parameter $ \sigma MJ^2 \sqrt{T_0}\|Z_{\cdot k}\|_2 = MJ^2\sigma \sqrt{T_0 N_0}$. We can now bound the $L^2$ norm of $\frac{1}{T_0}\bm{\mu}_T'\bm{\mu}'\bm{\varepsilon}_{0(1:T_0)}'\bm{Z}_{0\cdot} \in \R^K$:
  
    \[
    P\left(\frac{1}{T_0} \|\bm{\mu}_T'\bm{\mu}'\bm{\varepsilon}_{0(1:T_0)}'\bm{Z}_{0\cdot}\|_2 \geq 2 J M^2 \sigma \left(\sqrt{\frac{N_0 K \log 5}{T_0}} + \delta\right)\right) \leq 2\exp\left(-\frac{T_0\delta^2}{2}\right)
    \]
    Replacing $\delta$ with $\sqrt{\frac{KN_0}{T_0}}(2-\sqrt{\log 5})$ and with the Cauchy-Schwarz inequality we see that
    \[
      \frac{1}{T_0}\left|\bm{\mu}_T'\bm{\mu}'\bm{\varepsilon}_{0(1:T_0)}'\bm{Z}_{0\cdot}(\bm{Z}_{0\cdot}'\bm{Z}_{0\cdot})^{-1}(\bm{Z}_1 - \bm{Z}_{0\cdot}'\hat{\bm{\gamma}})\right| 
      \leq 4JM^2\sigma\sqrt{\frac{K}{T_0N_0}}\|\bm{Z}_1 - \bm{Z}_{0\cdot}'\hat{\bm{\gamma}}^\scm\|_2
    \]
    with probability at least $1-2\exp\left(-\frac{KN_0(2-\sqrt{\log 5})^2}{2}\right)$. 
    
    Next we turn to the residual term. By H\"{o}lder's inequality and using that for a matrix $\bm{A}$, the operator norm is bounded by $\|\bm{A}\|_2 \leq \sqrt{\trace(\bm{A}'\bm{A})}$ we see that 
    \[
      \begin{aligned}
        \left|\frac{1}{T_0}\bm{\mu}_T'\bm{\mu}'(\bm{E}_{1\cdot} - \bm{E}_{0\cdot}'\hat{\bm{\gamma}}^\cov)\right| & \leq \frac{JM^2}{\sqrt{T_0}}\left(\|\bm{E}_{1\cdot}\|_2 + \|\hat{\bm{\gamma}}^\cov\|_2\|\bm{E}_{0\cdot}\|_2\right)\\
        & \leq JM^2\left(\max_{t=1,\ldots,T_0}|e_{1t}| + \|\hat{\bm{\gamma}}^\cov\|_2\sqrt{\frac{1}{T_0}\sum_{t=1}^{T_0} RSS_t}\right)\\
        & \leq JM^2\left(\max_{t=1,\ldots,T_0}|e_{1t}| + \|\hat{\bm{\gamma}}^\cov\|_2 \sqrt{\max_{t}RSS_t}\right),
      \end{aligned}
    \]
    where we have used that $\frac{1}{\sqrt{T_0}}\|\bm{E}_{1\cdot}\|_2 \leq \max_{t=1,\ldots,T_0}|e_{1t}|$ and $\trace(\bm{E}_{0\cdot}'\bm{E}_{0\cdot}) = \sum_{t=1}^{T_0} RSS_t$.
  
    Combining with Lemma \ref{lem:aug_aux} and putting together the pieces with the union bound gives the result.
  \end{proof}

\clearpage
\section{Connection to balancing weights and IPW}
\label{sec:ipw_connection}

We have motivated Augmented SCM via bias correction. An alternative motivation comes from the connection between SCM and inverse propensity score weighting (IPW).

First, notice that the SCM weights from the constrained optimization problem in Equation \eqref{eq:vanillaSCM} are a form of \emph{approximate balancing weights} \citep[see, for example][]{Zubizarreta2015, Athey2016, Tan2017, Wang2018, Zhao2017}. 
Unlike traditional inverse propensity score weights, which indirectly minimize covariate imbalance by estimating a propensity score model, balancing weights seek to \emph{directly} minimize covariate imbalance, in this case $L^2$ imbalance.
Balancing weights have a Lagrangian dual formulation as inverse propensity score weights \citep[see, for example][]{Zhao2016a, Zhao2017, Chattopadhyay2019}.
Extending these results to the SCM setting, the Lagrangian dual of the SCM optimization problem in Equation \eqref{eq:vanillaSCM} has the form of a propensity score model. Importantly, as we discuss below, it is not always appropriate to interpret this model as a propensity score. 

We first derive the Lagrangian dual for a general class of balancing weights problems, then specialize to the penalized SCM estimator \eqref{eq:vanillaSCM}.
\begin{equation}
  \label{eq:scm_primal_general}
  \begin{aligned}
    \min_{\bm{\gamma}}\;\;   &\underbrace{h_\zeta(\bm{X}_1 - \bm{X}_{0\cdot}' \bm{\gamma})}_{\text{balance criterion}} + \sum_{W_i=0}~ \underbrace{f(\gamma_i)}_{\text{dispersion}}\\
        \text{subject to } & \sum_{W_i=0}\gamma_i = 1.
  \end{aligned}
\end{equation}
This formulation generalizes Equation \eqref{eq:vanillaSCM} in two ways: first, we remove the non-negativity constraint and note that this can be included by restricting the domain of the strongly convex dispersion penalty $f$. Examples include the re-centered $L^2$ dispersion penalties for ridge regression and ridge ASCM, an entropy penalty \citep{Robbins2017}, and an elastic net penalty \citep{Doudchenko2017}. Second, we generalize from the squared $L^2$ norm to a general balance criterion $h_\zeta$; another promiment example is an $L^\infty$ constraint \citep[see e.g.][]{Zubizarreta2015, Athey2016}.

\begin{proposition}
\label{prop:scm_dual}
The Lagrangian dual to Equation \eqref{eq:scm_primal_general} is
\begin{equation}
  \label{eq:scm_dual_general}
  \min_{\alpha, \bm{\beta}} \;\; \underbrace{\sum_{W_i=0} f^{\ast}(\alpha + \bm{\beta}'X_{i\cdot}) - (\alpha + \bm{\beta}'\bm{X}_1)}_{\text{loss function}} + \underbrace{\vphantom{\sum_{W_i=0} }h^\ast_\zeta(\bm{\beta})}_{\text{regularization}},
\end{equation}
where a convex, differentiable function $g$ has convex conjugate $g^\ast(\bm{y}) \equiv \sup_{\bm{x} \in \text{dom}(g)} \{\bm{y}'\bm{x} - g(\bm{x})\}$.
The solutions to the primal problem \eqref{eq:scm_primal_general} are
$  \hat{\gamma}_i = f^{\ast \prime}(\hat{\alpha} + \hat{\bm{\beta}}'\bm{X}_i), $
where $f^{\ast \prime}(\cdot)$ is the first derivative of the convex conjugate, $f^\ast(\cdot)$.
\end{proposition}
There is a large literature relating balancing weights to propensity score weights. This literature shows that the loss function in Equation \eqref{eq:scm_dual_general} is an M-estimator for the propensity score and thus will be consistent for the propensity score parameters under large $N$ asymptotics.
The dispersion measure $f(\cdot)$ determines the link function of the propensity score model, where the odds of treatment are $\frac{\pi(x)}{1-\pi(x)} = f^{\ast \prime}(\alpha + \bm{\beta}'x)$.
Note that un-penalized SCM, which can yield multiple solutions, does not have a well-defined link function.
We extend the duality to a general set of balance criteria so that Equation \eqref{eq:scm_dual_general} is a regularized M-estimator of the propensity score parameters where the balance criterion $h_\zeta(\cdot)$ determines the type of regularization through its conjugate $h_\zeta^\ast(\cdot)$. This formulation recovers the duality between entropy balancing and a logistic link \citep{Zhao2016a}, Oaxaca-Blinder weights and a log-logistic link \citep{Kline2011}, and $L^\infty$ balance and $L^1$ regularization \citep{Wang2018}.
This more general formulation also suggests natural extensions of both SCM and ASCM beyond the $L^2$ setting to other forms, especially $L^1$ regularization.

Specializing proposition \ref{prop:scm_dual} to a squared $L^2$ balance criterion $h_\zeta(x) = \frac{1}{2\zeta}\|x\|_2^2$ as in the penalized SCM problems yields that the dual propensity score coefficients $\bm{\beta}$ are regularized by a ridge penalty. In the case of an entropy dispersion penalty as \citet{Robbins2017} consider, the donor weights $\hat{\bm{\gamma}}$ have the form of IPW weights with a logistic link function, 
where the propensity score is $\pi(\bm{X}_i) = \text{logit}^{-1}(\alpha + \bm{\beta}'\bm{X}_i)$, 
the odds of treatment are $\frac{\pi(\bm{X}_i)}{1-\pi(\bm{X}_i)}=\exp(\alpha + \bm{\beta}'\bm{X}_i)=\gamma_i$.

We emphasize that while Proposition \ref{prop:scm_dual} shows that the the estimated weights have the IPW form, in SCM settings it may not always be appropriate to interpret the dual problem as a propensity score reflecting stochastic selection into treatment. For example, this interpretation would not be appropriate in some canonical SCM examples, such as the analysis of German reunification in \citet{Abadie2015}.

\begin{proof}[Proof of Proposition \ref{prop:scm_dual}]
We can augment the optimization problem \eqref{eq:scm_primal_general} with auxiliary variables $\epsilon$, yielding:

\begin{equation}
  \begin{aligned}
    \min_{\bm{\gamma}, \bm{\epsilon}}\;\;   & h_\zeta(\bm{\epsilon})+ \sum_{W_i=0} f(\gamma_i).    \\
    \text{subject to } & \bm{\epsilon} = \bm{X}_1 - \bm{X}_{0\cdot}'\bm{\gamma}\\
    & \sum_{W_i=0}\gamma_i = 1
  \end{aligned}
\end{equation}

The Lagrangian is
\begin{equation}
\calL(\bm{\gamma}, \bm{\epsilon}, \alpha, \bm{\beta}) = \sum_{i \mid W_i=0} f(\gamma_i) + \alpha (1 - \gamma_i) + h_\zeta(\bm{\epsilon}) +  \bm{\beta}'(\bm{X}_1 - \bm{X}_{0\cdot}' \bm{\gamma} -\bm{\epsilon}).
\label{eq:lagrange}
\end{equation}

The dual maximizes the objective

\begin{equation}
\begin{aligned}
  q(\alpha, \bm{\beta})  & = \min_{\bm{\gamma}, \bm{\epsilon}}\calL(\bm{\gamma}, \epsilon, \alpha, \bm{\beta})\\
  & =\sum_{W_i=0} \min_{\gamma_i} \{f(\gamma_i) - (\alpha + \bm{\beta}'\bm{X}_i )\gamma_i \} + \min_{\bm{\epsilon}} \{h_\zeta(\bm{\epsilon}) - \bm{\beta}'\bm{\epsilon}\} + \alpha + \bm{\beta}' \bm{X}_1\\
    & = -\sum_{W_i =0} f^*(\alpha + \bm{\beta}'\bm{X}_{i}) + \alpha + \bm{\beta}'\bm{X}_1' - h_\zeta^*(\bm{\beta}),
\end{aligned}  
\end{equation}
 By strong duality the general dual problem \eqref{eq:scm_dual_general}, which minimizes $-q(\alpha, \bm{\beta})$, is equivalent to the primal balancing weights problem. Given the $\hat{\alpha}$ and  $\hat{\bm{\beta}}$ that minimize the Lagrangian dual objective, $-q(\alpha, \bm{\beta})$, we recover the donor weights solution to \eqref{eq:scm_primal_general} as
    \begin{equation}
  \label{eq:weight_sol2}
  \hat{\gamma}_i = f^{\ast \prime}(\hat{\alpha} + \hat{\bm{\beta}}'\bm{X}_i).
\end{equation}
\end{proof}

\clearpage
\section{Additional figures}
\label{sec:additional_plots}

\begin{figure}[htbp]
\centering \includegraphics[width=0.5\textwidth]{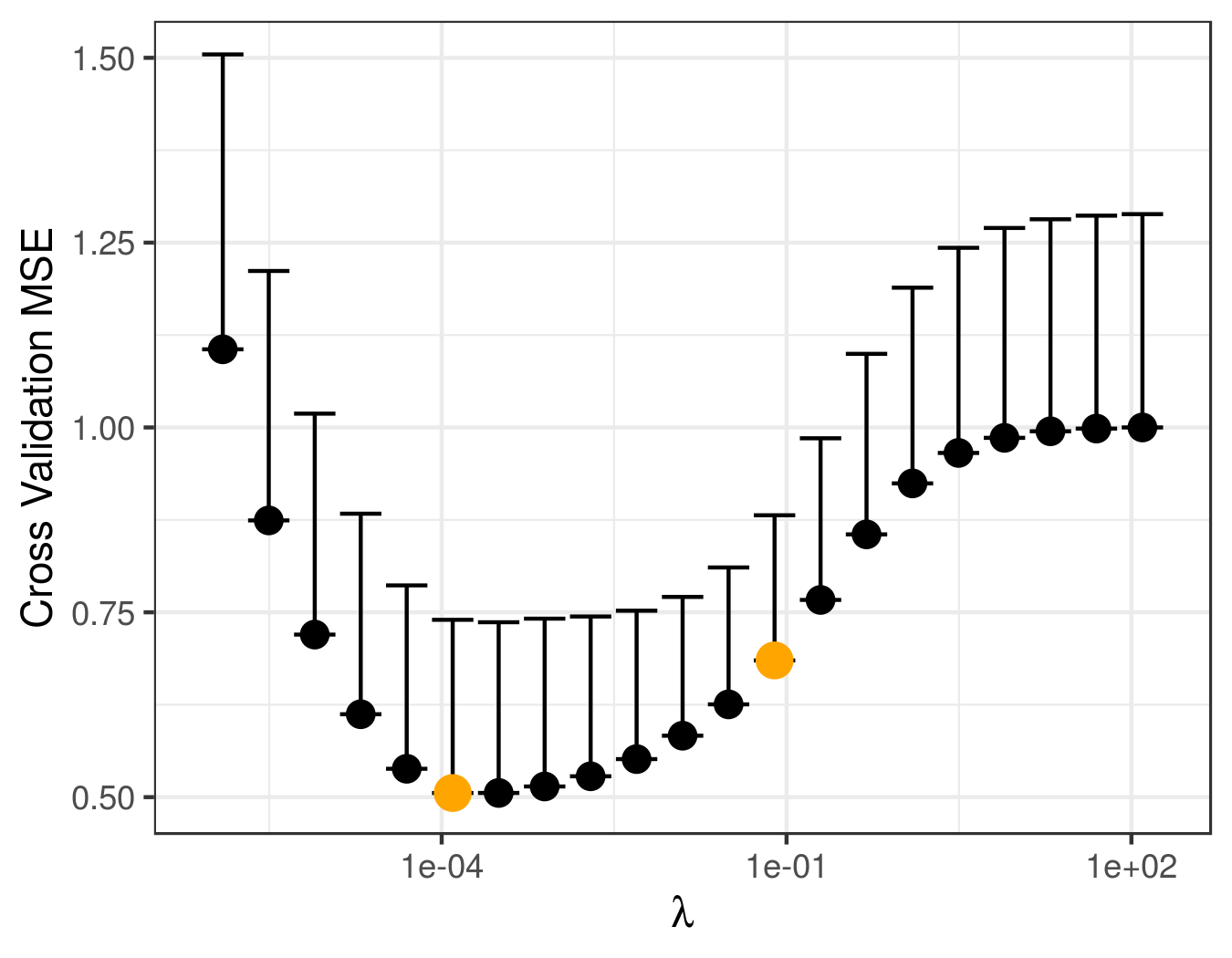}
  \caption{Cross validation MSE and one standard error computed according to Equation \eqref{eq:loocv}. The minimal point, and the maximum $\lambda$ within one standard error of the minimum are highlighted.}
    \label{fig:lambda_cv}
\end{figure}

\begin{figure}
  {\centering \includegraphics[width=\maxwidth]{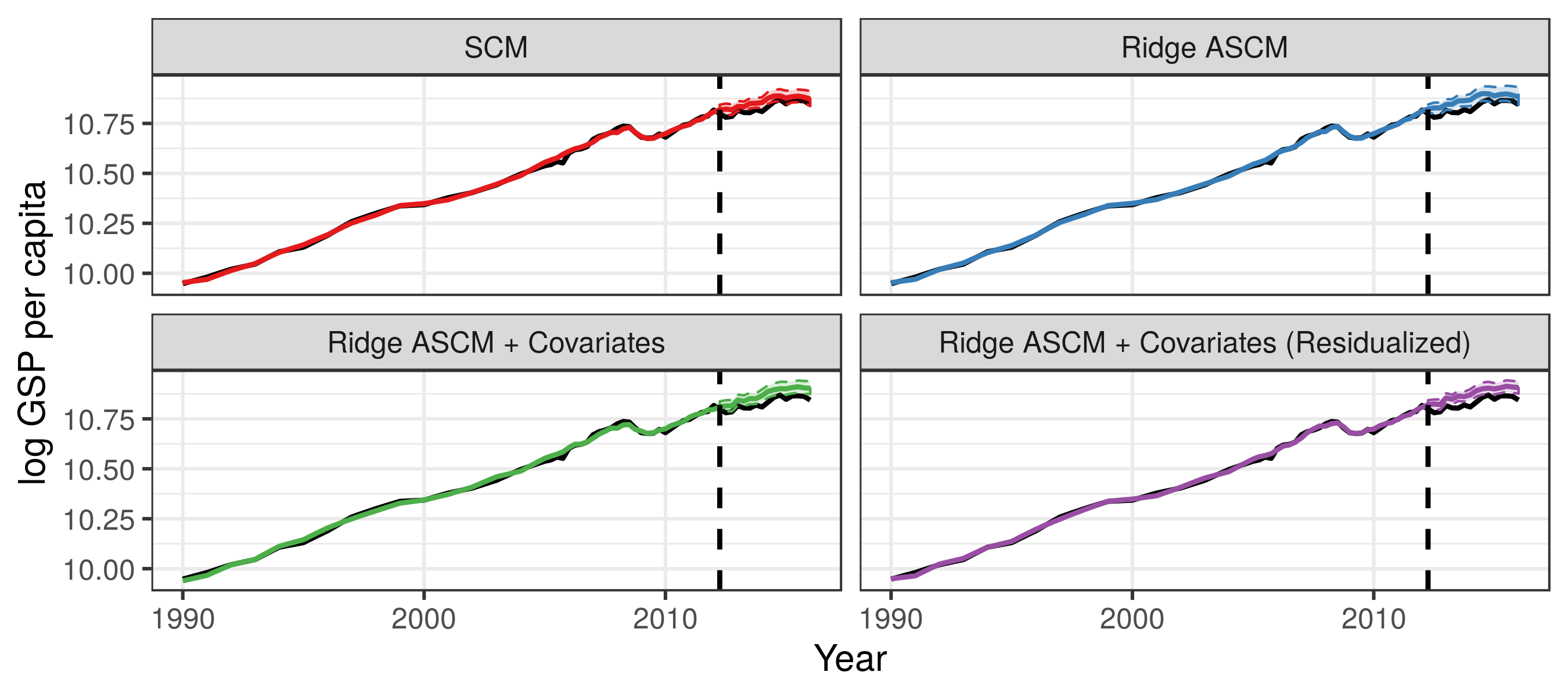}}
    \caption{Point estimates along with point-wise 95\% conformal prediction intervals for counterfactual log GSP per capita without the tax cuts using SCM, ridge ASCM, and ridge ASCM with covariates, plotting with the observed log GSP per capita in black.}
      \label{fig:synth_estimates_raw}
  \end{figure}

  \begin{figure}
    {\centering \includegraphics[width=\maxwidth]{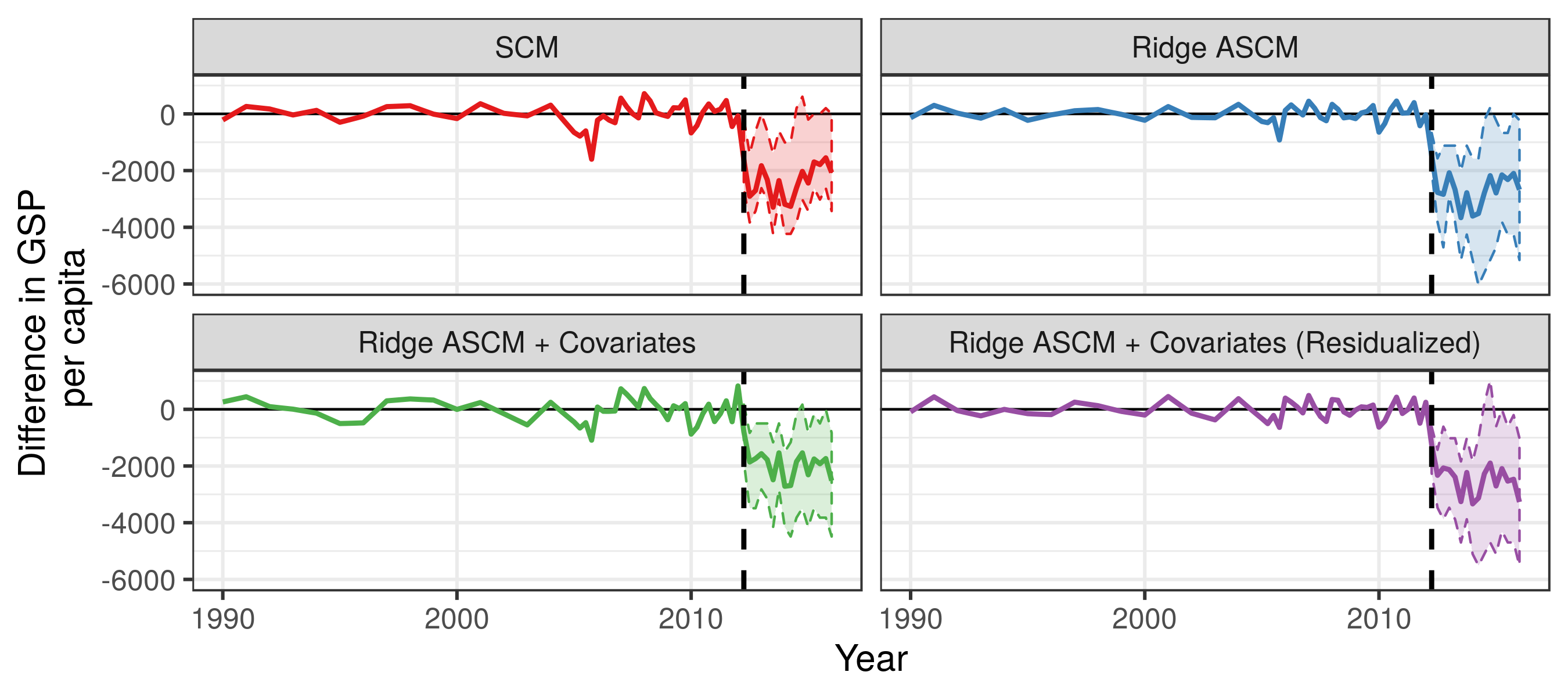}}
      \caption{Point estimates along with point-wise 95\% conformal confidence intervals for the effect of the tax cuts on GSP per capita using SCM, ridge ASCM, and ridge ASCM with covariates.}
      \label{fig:synth_estimates_levels}
    \end{figure}

\begin{figure}[!hbtp]
{\centering \includegraphics[width=\maxwidth]{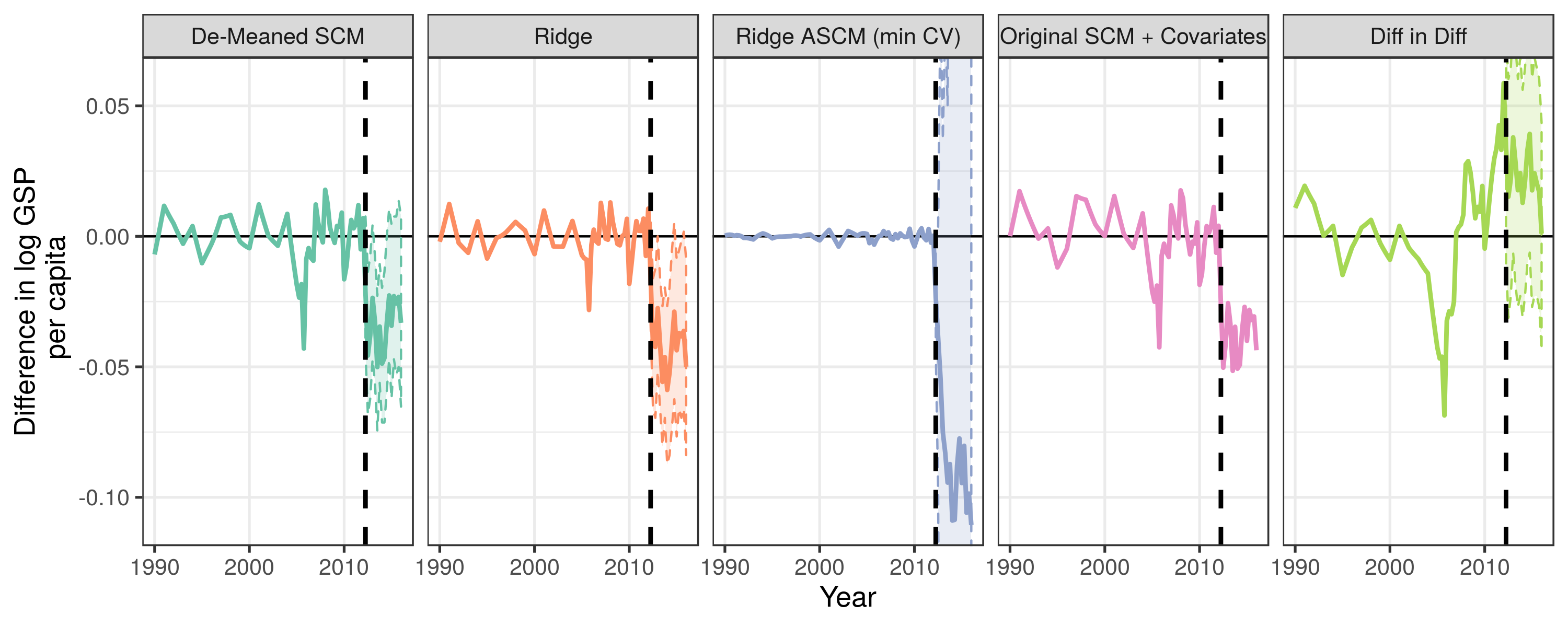} 
}
\caption{Point estimates $\pm$ two standard errors of the ATT for the effect of the tax cuts on log GSP per capita using de-meaned SCM, ridge regression alone, ridge ASCM with $\lambda$ chosen to minimize the cross validated MSE, the original SCM proposal with covariates \citep{AbadieAlbertoDiamond2010}, and a two-way fixed effects differences in differences estimate.}
    \label{fig:synth_estimates_appendix}
\end{figure}

\begin{figure}[!hbtp]
{\centering \includegraphics[width=0.8\maxwidth]{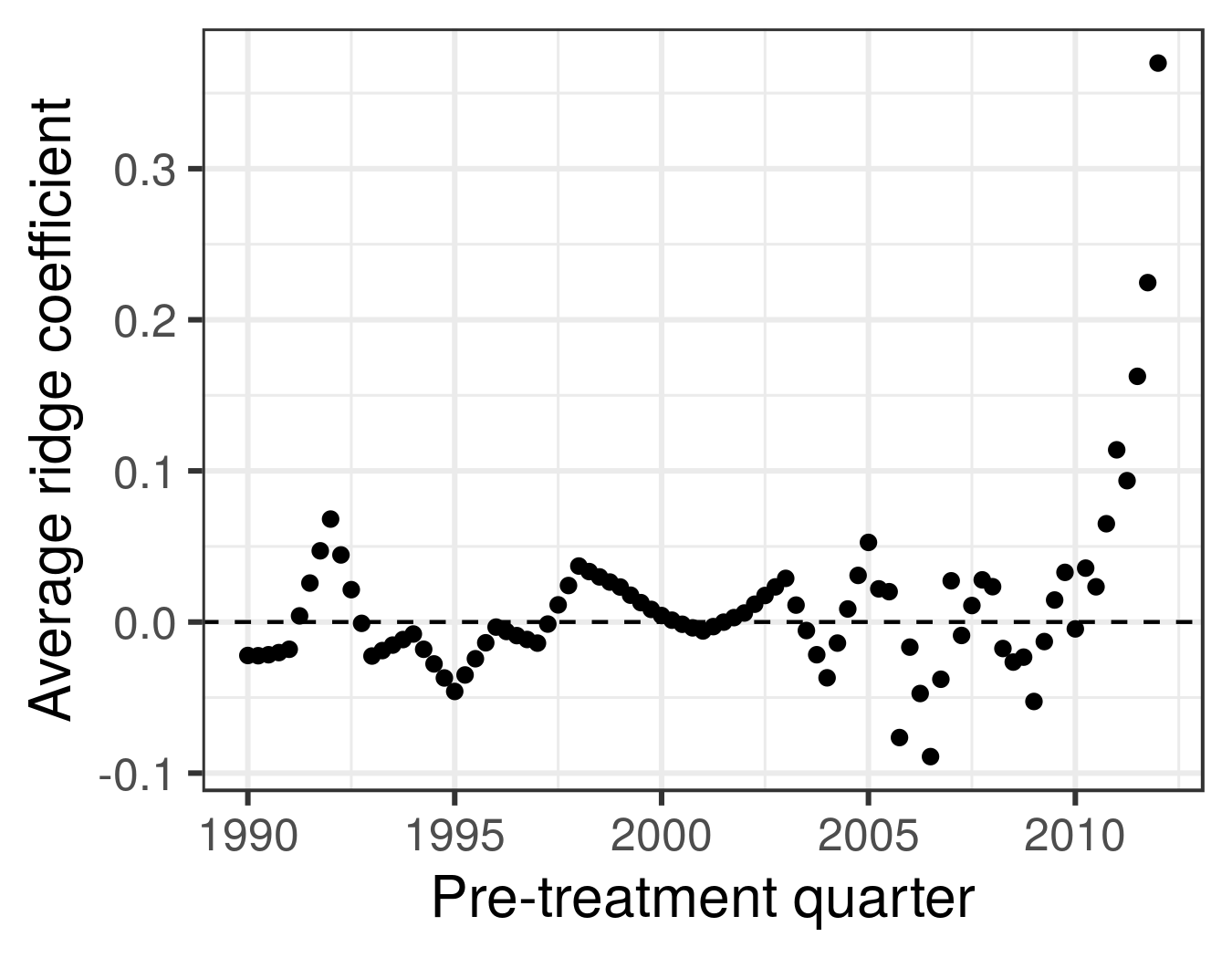}
}
\caption{Ridge regression coefficients for each pre-treatment quarter, averaged across post-treatment quarters.}
    \label{fig:ridge_coefs}
\end{figure}

\begin{figure}[htbp]
{\centering \includegraphics[width=\maxwidth]{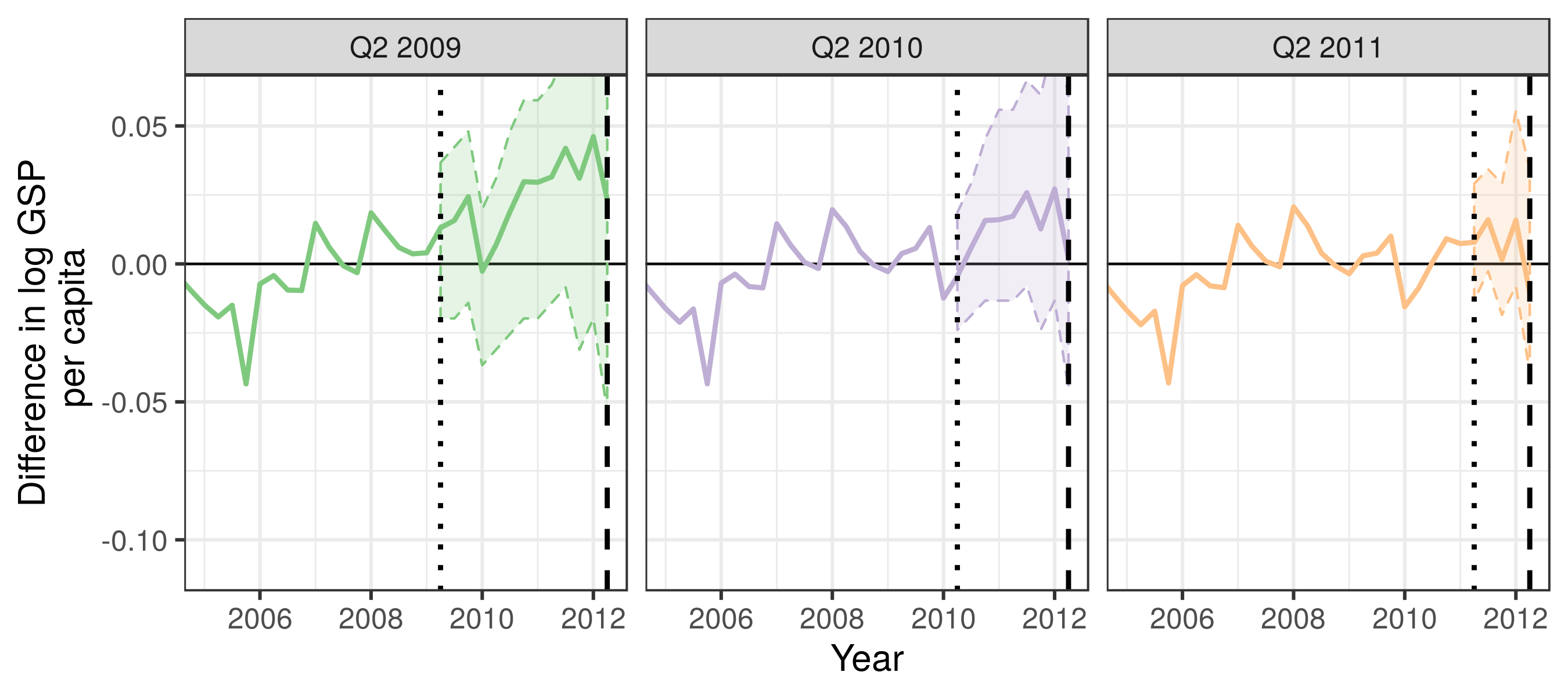}}
  \caption{Placebo point estimates $\pm$ two standard errors for SCM with placebo treatment times in Q2 2009, 2010, and 2011. Scale begins in 2005 to highlight placebo estimates.}
    \label{fig:lngdpcapita_placebo_scm}
\end{figure}

\begin{figure}
{\centering \includegraphics[width=\maxwidth]{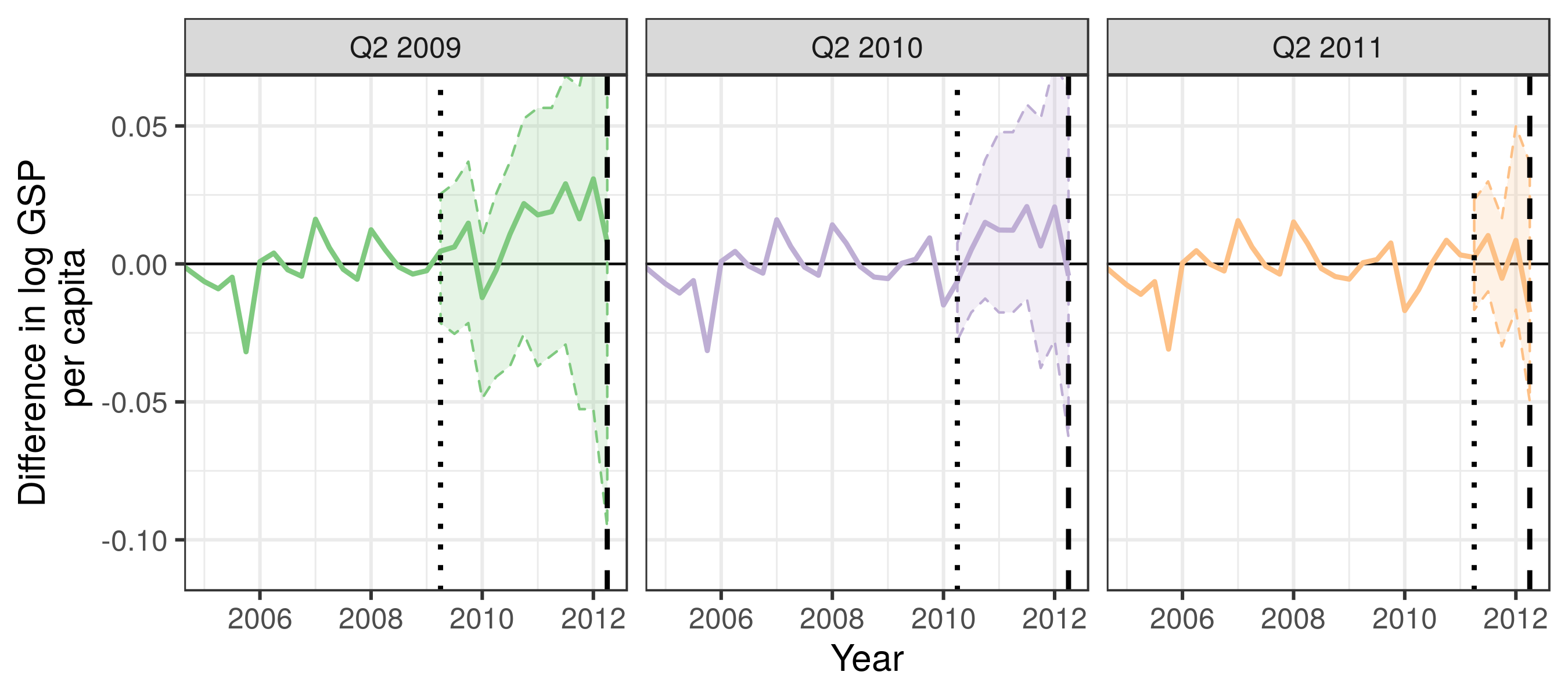}}
  \caption{Placebo point estimates $\pm$ two standard errors for ridge ASCM with placebo treatment times in Q2 2009, 2010, and 2011. Scale begins in 2005 to highlight placebo estimates.}
    \label{fig:lngdpcapita_placebo_ascm}
\end{figure}

\begin{figure}[!hbtp]
\centering \includegraphics[width=\maxwidth]{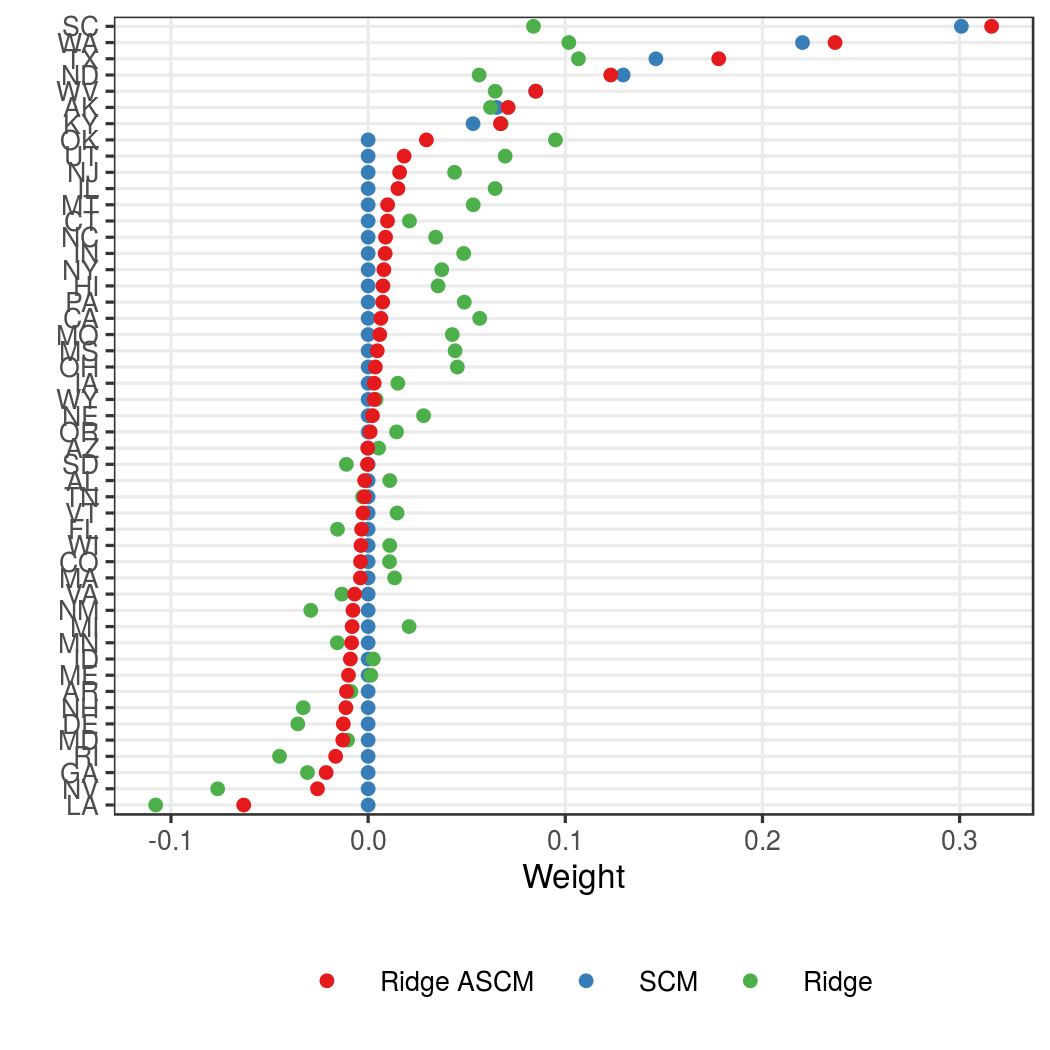}
\caption{Donor unit weights for SCM, ridge regression, and ridge ASCM balancing lagged outcomes.}
\label{fig:weight_plot_appendix}
\end{figure}

\begin{figure}[!hbtp]
  \centering \includegraphics[width=\maxwidth]{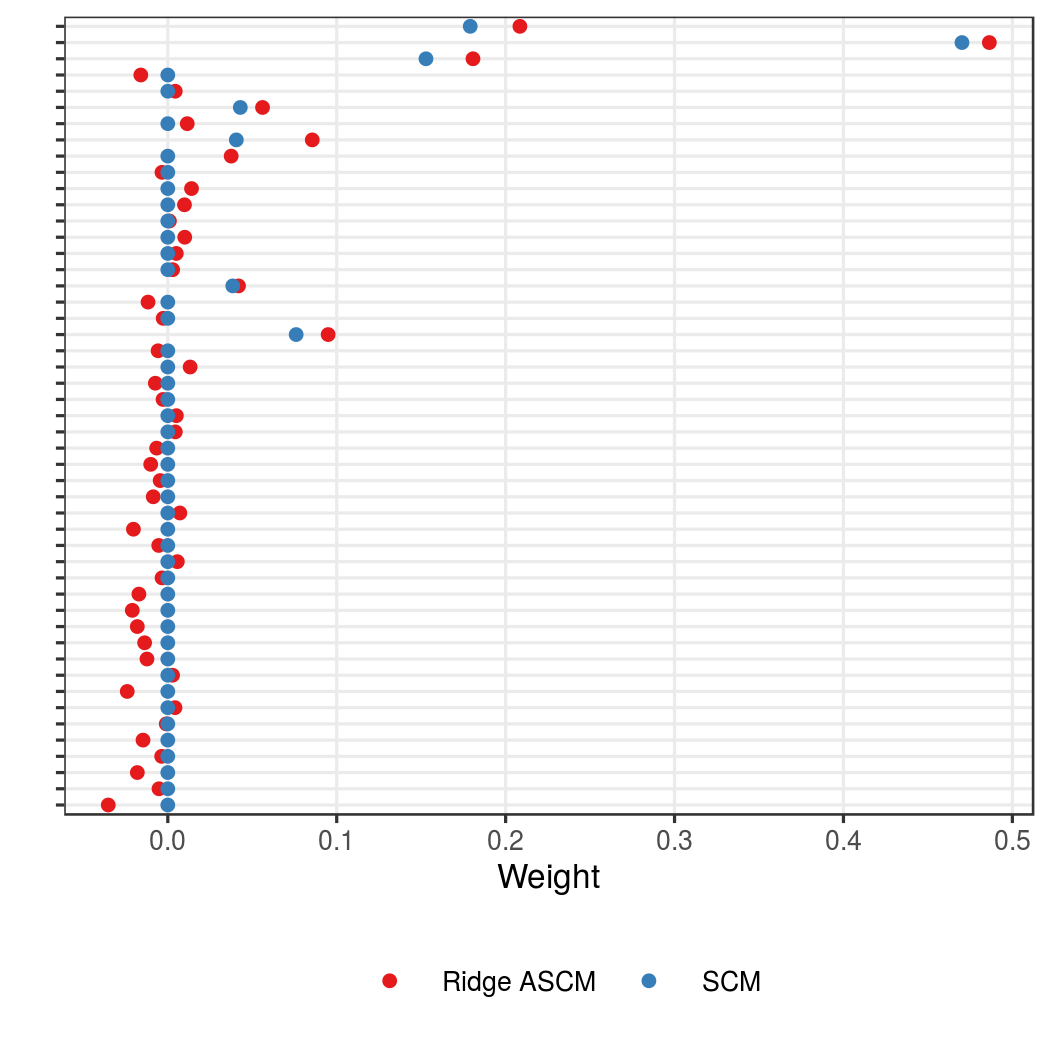}
  \caption{Donor unit weights for SCM and ridge ASCM fit on lagged outcomes after residualizing out auxiliary covariates.}
  \label{fig:weight_plot_resid}
  \end{figure}

\begin{figure}[h]
{\centering \includegraphics[width=\maxwidth]{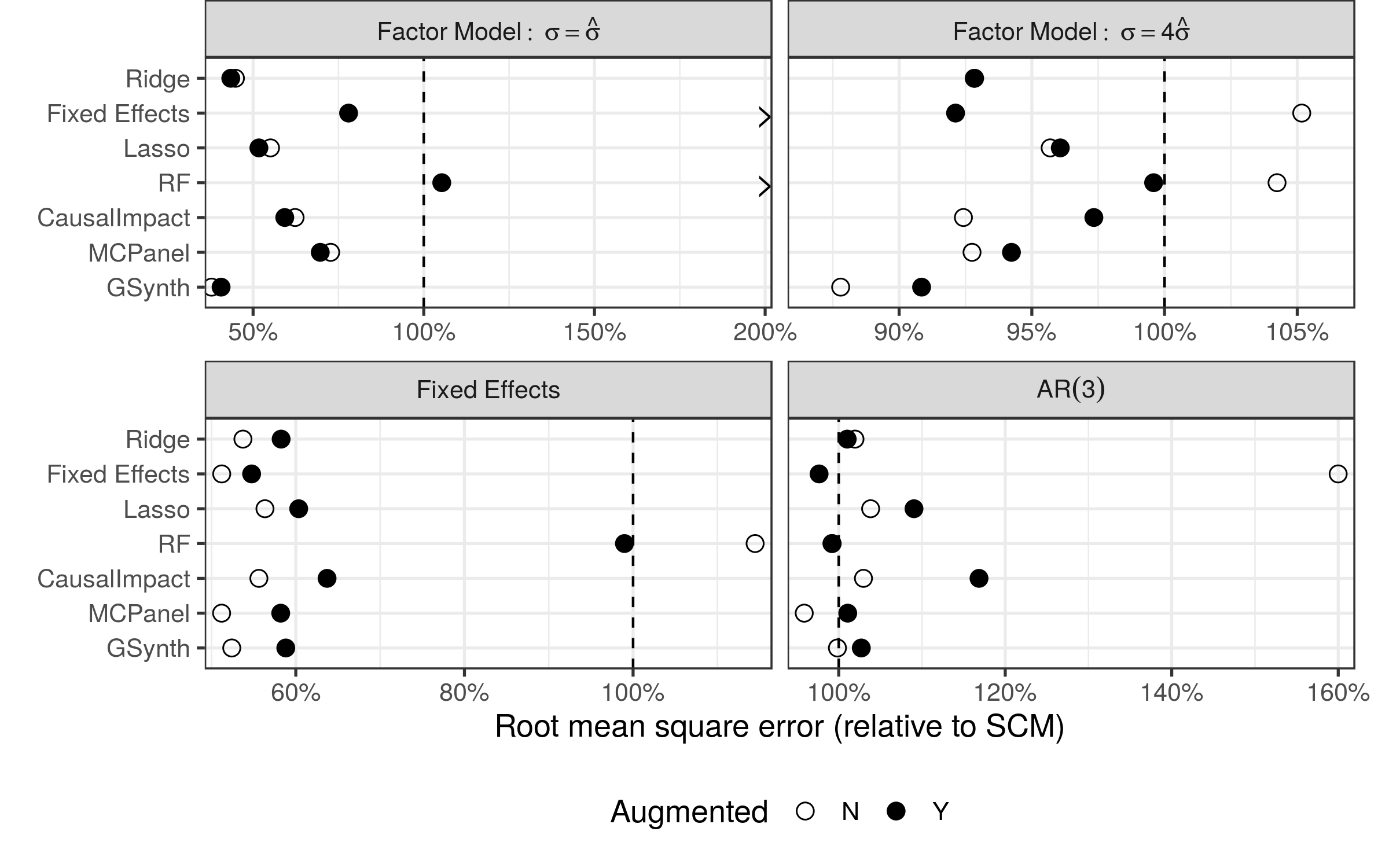} 
}
\caption{RMSE for different augmented and non-augmented estimators across outcome models.}
\label{fig:overview_rmse_plots}
\end{figure}

\begin{figure}[h]
{\centering \includegraphics[width=\maxwidth]{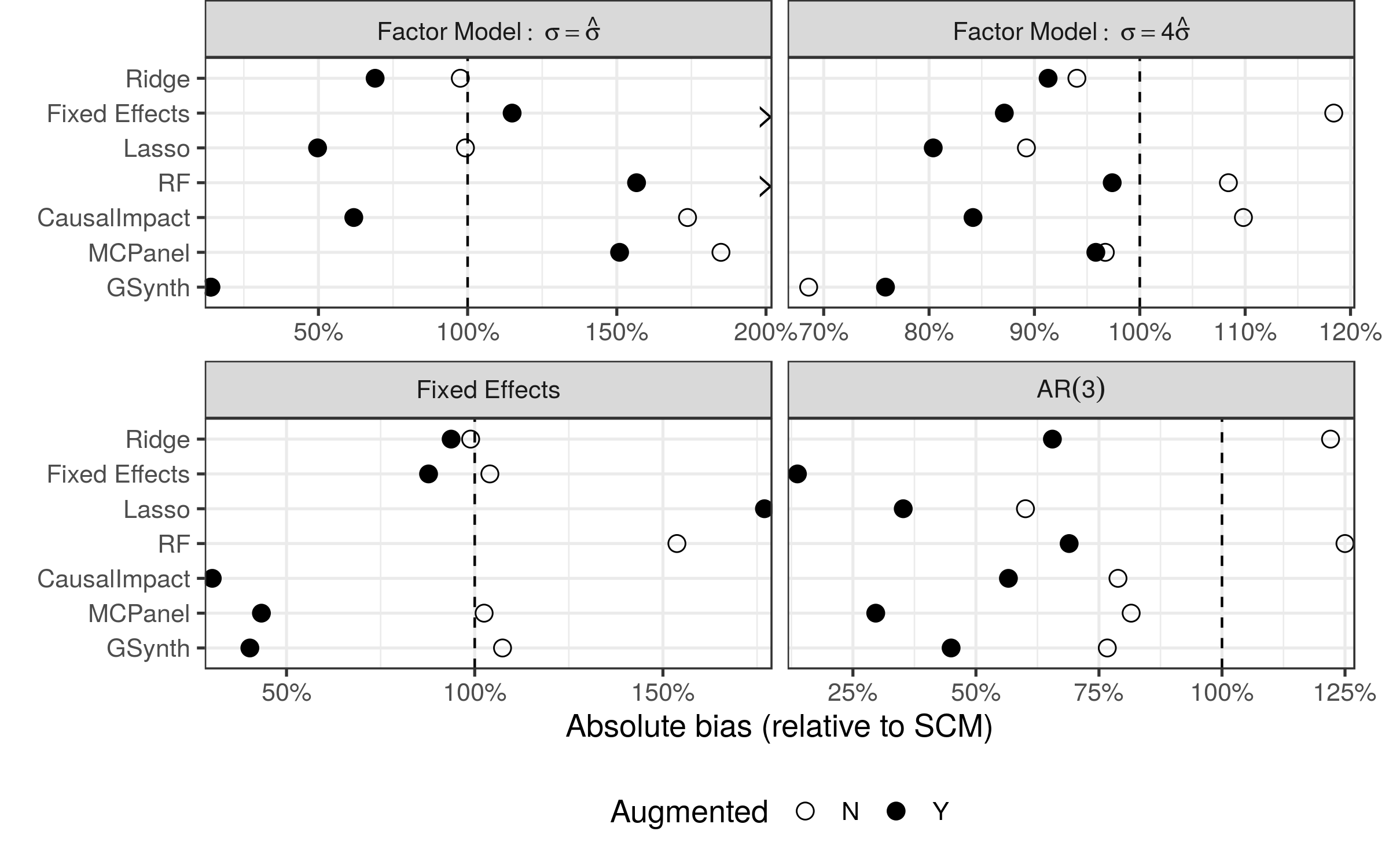} 
}
\caption{Bias for different augmented and non-augmented estimators across outcome models conditioned on SCM fit in the top quintile.}
\label{fig:bias_good_fit}
\end{figure}

\begin{figure}[htbp]
{\centering \includegraphics[width=\maxwidth]{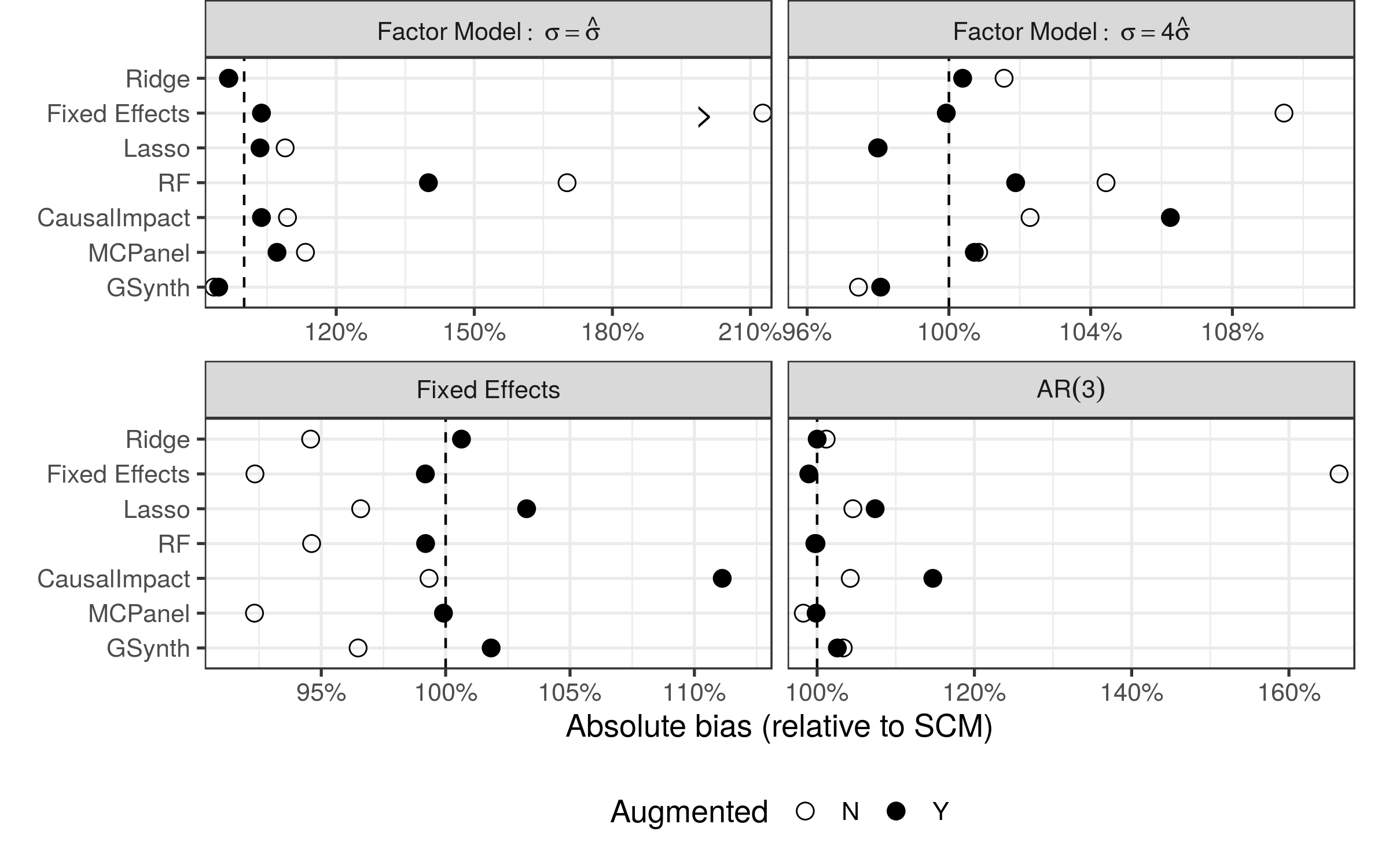} 
}
\caption{RMSE for different augmented and non-augmented estimators across outcome models conditioned on SCM fit in the top quntile.}
\label{fig:rmse_good_fit}
\end{figure}

\begin{figure}[htbp]
{\centering \includegraphics[width=\maxwidth]{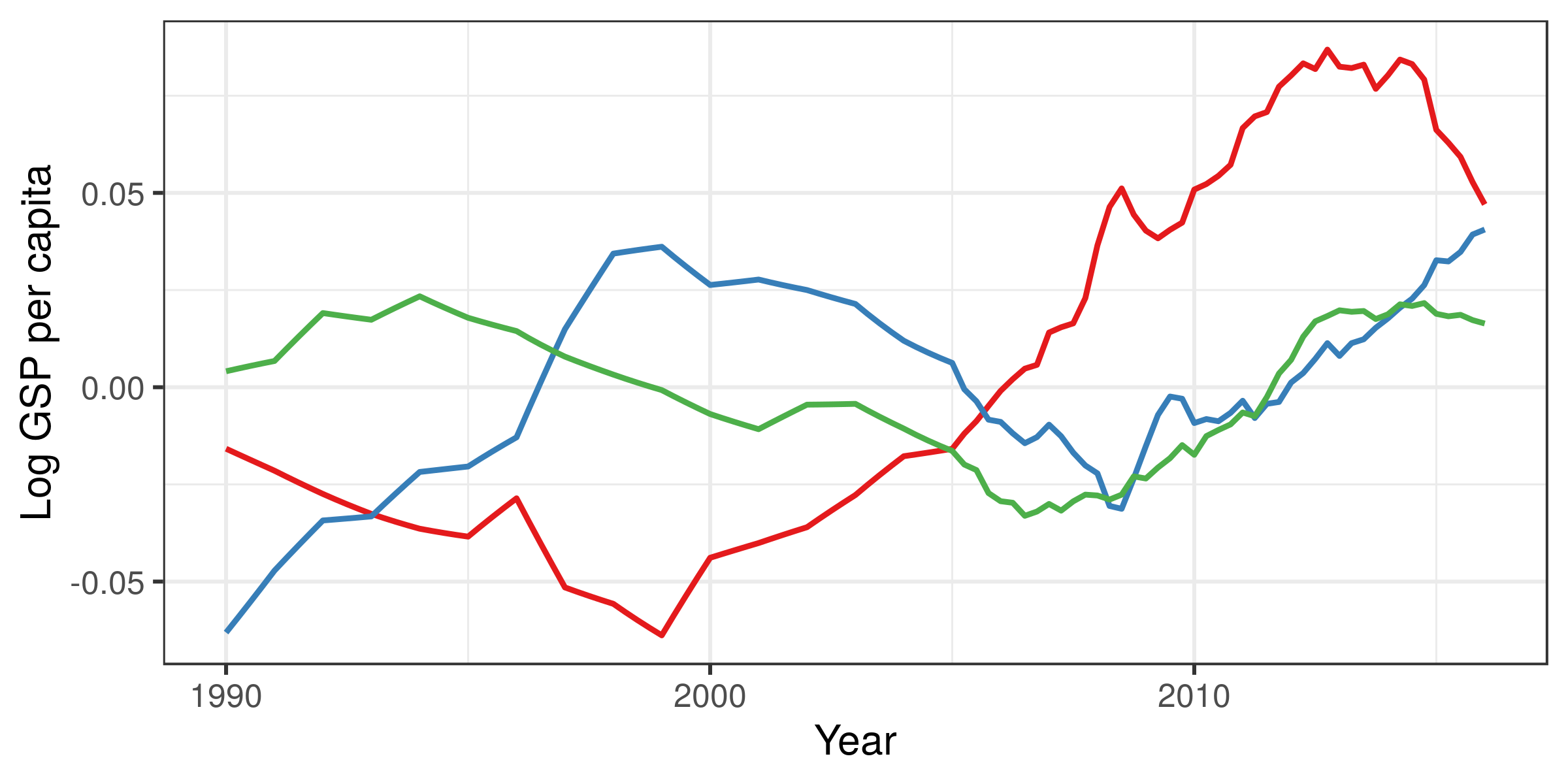} 
}
\caption{Latent factors for calibrated simulation studies.}
\label{fig:factors}
\end{figure}


\clearpage
\singlespacing
\bibliographystyle{chicago}
\bibliography{syn_ctrls}

\end{document}